\documentclass[%
 aip,
 amsmath,amssymb,
 reprint,%
]{revtex4-1}

\usepackage{graphicx}% Include figure files
\usepackage{dcolumn}% Align table columns on decimal point
\usepackage{bm}% bold math
\usepackage[utf8]{inputenc}
\usepackage[T1]{fontenc}
\usepackage{mathptmx}
\usepackage{amsthm}
\usepackage{amsmath}
\usepackage{eqnarray}
\usepackage{tikz}
\usepackage{stackrel}
\usepackage{siunitx}
\usepackage{enumerate}
\usepackage[version=4]{mhchem} % chemistry
\usepackage{braket} % sets
\usepackage{subcaption} % multiple plots in one figure
\usepackage{algorithm}
\usepackage{algpseudocode}
\usepackage{setspace}
\usepackage{tabularx}
\usepackage{dsfont} 
\usepackage{xr}

\usetikzlibrary{arrows,shapes,arrows.meta}%,graphs,graphdrawing,automata}

\newcommand{\floor}[1]{\left\lfloor #1 \right\rfloor}
\newcommand{\maps}{\rightarrow}

\newcommand{\R}{\mathbb{R}}

\newcommand{\Z}{\mathbb{Z}}
\newcommand{\E}{\mathbb{E}}
\newcommand{\Sym}{\mathbb{S}}
\newcommand{\G}{C(X)}
\newcommand{\spn}[1]{\text{span}\left(#1\right)}
\newcommand{\St}[1]{\left\lbrace #1 \right\rbrace}
\newcommand{\order}[1]{$O\left(#1\right)$}

\newcommand{\rb}[1]{\left( #1 \right)}
{	%\theoremstyle{plain}
	\newtheorem{assumption}{Assumption}
}
{	%\theoremstyle{plain}
	\newtheorem{theorem}{Theorem}
}
{	%\theoremstyle{plain}
	\newtheorem{proposition}{Proposition}
}
{	%\theoremstyle{plain}
	\newtheorem{corollary}{Corollary}
}
{	%\theoremstyle{plain}
	\newtheorem{remark}{Remark}
}
{	%\theoremstyle{plain}
	
}
{
	
}
{
	\newtheorem{example}{Example}
}
{
	\newtheorem{definition}{Definition}
}

\begin{document}

\title{Tighter bounds on transient moments of stochastic chemical systems}

\author{Flemming Holtorf}
\email{holtorf@mit.edu}
\author{Paul I. Barton}%
\altaffiliation{Corresponding Author}
\email{pib@mit.edu}
\affiliation{$^{\text{1)}}$Department of Chemical Engineering, Massachusetts Institute of Technology, Cambridge,
	Massachusetts 02139, USA%\\This line break forced with \textbackslash\textbackslash
}%


\date{5 June 2020}% It is always \today, today,

\begin{abstract}
	The use of approximate solution techniques for the Chemical Master Equation is a common practice for the analysis of stochastic chemical systems. Despite their widespread use, however, many such techniques rely on unverifiable assumptions and only a few provide mechanisms to control the approximation error quantitatively. Addressing this gap, Dowdy and Barton [\textit{The Journal of Chemical Physics}, \textbf{149(7)}, 074103 (2018)] proposed a method for the computation of guaranteed bounds on the moment trajectories associated with stochastic chemical systems described by the Chemical Master Equation, thereby providing a general framework for error quantification. Here, we present an extension of this method. The key contribution is a new hierarchy of convex necessary moment conditions crucially reflecting the temporal causality and other regularity conditions that are inherent to the moment trajectories associated with stochastic processes described by the Chemical Master Equation. Analogous to the original method, these conditions generate a hierarchy of semidefinite programs that furnishes monotonically improving bounds on the trajectories of the moments and related statistics. Compared to its predecessor, the presented hierarchy produces bounds that are at least as tight and it often enables the computation of dramatically tighter bounds as it enjoys superior scaling properties and the arising semidefinite programs are highly structured. We analyze the properties of the presented hierarchy in detail, discuss some aspects of its practical implementation and demonstrate its merits with several examples.
\end{abstract}

\maketitle
\section{Tighter Bounds on the Transient Moments of Stochastic Chemical Systems} \label{chap:chemical_kinetics}
\subsection{Introduction}
The analysis of systems undergoing chemical reactions lies at the heart of many scientific and engineering activities. While deterministic models have proved adequate for the analysis of systems at the macroscopic scale, they often fall short for meso- and microscopic systems; in particular for those that feature low molecular counts. In this regime, the complex and chaotic motion of molecules reacting upon collision causes effectively stochastic fluctuations of the molecular counts that are large compared to the mean and, as a consequence, can have a profound effect on the system's characteristics -- a situation frequently encountered in cellular biology \cite{arkin1998stochastic, elowitz2002stochastic, liu2004fluctuations, artyomov2007purely}. In the context of the continuously growing capabilities of synthetic biology, this fact motivates the use of stochastic models for the identification, design and control of biochemical reaction networks. However, while in these applications stochastic models provide the essential fidelity relative to their deterministic counterparts, their analysis is generally more involved.  

Stochastic chemical systems are canonically modeled as jump processes. Specifically, the system state, as encoded by the molecular counts of the individual chemical species, is modeled to change only discretely in response to reaction events as triggered by the arrivals of Poisson processes whose rates depend on the underlying reaction mechanism. The so-called Chemical Master Equation (CME) describes how the probability distribution of the state of such a process evolves over time. To that end, the CME tracks the probability to observe the system in any reachable state over time. This introduces a major challenge for the analysis of stochastic chemical systems which routinely feature millions or even infinitely many reachable states, rendering a direct solution of the CME intractable. As a consequence, sampling techniques such as Gillespie's Stochastic Simulation Algorithm \cite{gillespie1976general,gillespie1977exact} have become the most prominent approach for the analysis of systems described by the CME. Although these techniques perform remarkably well across a wide range of problems, they are inadequate in certain settings. Most notably, they do not scale well for stiff systems, generally do not provide hard error bounds and the evaluation of sensitivity information is challenging \cite{gillespie2013perspective}. Specifically, the two latter shortcomings limit their utility in the context of identification, design and control. Alternatives such as the finite state projection algorithm \cite{munsky2006finite} come with guaranteed error bounds and straightforward sensitivity evaluation, however, suffer generally more severely from high dimensionality.

From a practical perspective, stochastic reaction networks are often sufficiently characterized by only a few low-order moments of the associated distribution of the state, for example through means and variances. In that case, tractability of the CME may be recovered by solving for the moments of its solution directly. The dynamics of a finite sequence of moments associated with the distribution described by the CME, however, do not generally form a closed, square system, hence do not admit a solution by simple numerical integration. Numerous moment closure approximations \cite{ale2013general, keeling2000multiplicative, nasell2003extension, smadbeck2013closure} have been proposed to remedy this problem. A major shortcoming of moment closure approximations, however, is that they generally rely on unverifiable assumptions about the underlying distribution and therefore introduce an uncontrolled error. In fact, it is well-known that their application can lead to unphysical results such as spurious oscillations and negative mean molecular counts \cite{schnoerr2015comparison,schnoerr2014validity,grima2012study}. 
Addressing this shortcoming, several authors have recently proposed methods for the computation of rigorous, theoretically guaranteed bounds for the moments (or related statistics) associated with stochastic reaction networks; such bounding schemes have been proposed for the steady \cite{dowdy2018bounds,ghusinga2017exact,kuntz2019bounding,sakurai2017convex} and transient setting \cite{dowdy2018dynamic,sakurai2018bounding, del2016control, backenkohler2019bounding}, and their utility for the design of biochemical systems has been demonstrated \cite{sakurai2018optimization}. The key insight underpinning these methods, rooted in real algebraic geometry, is that the moment sequence associated with the true solution of the CME must satisfy a rich set of algebraic conditions reflecting the support and dynamics of the true distribution. Crucially, these conditions only depend on the problem data but do not require explicit knowledge of the true solution. The use of mathematical programming to identify a truncated moment sequence which minimizes (maximizes) a given moment or related statistic of interest subject to these conditions then furnishes a valid lower (upper) bound on the true value. While such bounds provide a mechanism to quantify errors or verify the consistency of approximation techniques, they are even frequently found to be sufficiently tight to be used directly as a proxy for the true solution\cite{sakurai2017convex,ghusinga2017exact,dowdy2018bounds,kuntz2019bounding}.  

In this work, we extend the bounding scheme proposed by \citet{dowdy2018dynamic} for transient moments of the solutions of the CME. To that end, we introduce new necessary moment conditions that improve the tightness of the semidefinite relaxations on which Dowdy and Barton's approach is based. These necessary moment conditions reflect crucially the temporal causality that is inherent to solutions of the CME. The conditions lend themselves to be organized in a hierarchy that provides a mechanism to trade-off computational cost for higher-quality bounds. Moreover, we show that the conditions exhibit favorable scaling properties and structure when compared to the conditions employed in the original method. 

This article is organized as follows. In Section \ref{sec:preliminaries}, we introduce definitions and assumptions, formally define the moment bounding problem, and review essential preliminaries. Section \ref{sec:tighter_bounds} is devoted to the development and analysis of the proposed hierarchy of necessary moment conditions. In Section \ref{sec:practicalities}, we discuss certain aspects pertaining to the use of these conditions for computation of moment bounds in practice. The potential of the developed methodology is demonstrated with several examples in Sections \ref{sec:examples} and \ref{sec:bounding_mechanisms} before we conclude with some open questions in Section \ref{sec:conclusion}.

\section{Preliminaries}\label{sec:preliminaries}
\subsection{Notation}
We denote scalars with lowercase symbols without emphasis while vectors and matrices are denoted by bold lower- and uppercase symbols, respectively. Throughout, vectors are assumed to be column vectors. Generic sets are denoted by uppercase symbols without emphasis. For special or commonly used sets we use the standard notation. For example, for the (non-negative) $n$-dimensional reals and integers, we use the usual notation of $\mathbb{R}^n$ ($\mathbb{R}^n_+$) and $\mathbb{Z}^n$ ($\mathbb{Z}^n_+$), respectively. Similarly, we refer to the set of symmetric and symmetric positive semidefinite (psd) $n$-by-$n$ matrices with $\mathbb{S}^n$ and $\mathbb{S}^n_+$, respectively, and use the usual shorthand notation $\bm{A}\succeq\bm{B}$ for $\bm{A}-\bm{B}\in\mathbb{S}_+^n$. The set of $n$-dimensional vector and symmetric matrix polynomials with real coefficients (of degree at most $k$) in the variables $\bm{x} = [x_1 \ \dots \ x_N]^\top$ will be denoted by $\mathbb{R}^n[\bm{x}]$ ($\mathbb{R}^n_k[\bm{x}]$) and $\mathbb{S}^n[\bm{x}]$ ($\mathbb{S}^n_k[\bm{x}]$), respectively. In order to concisely denote multivariate monomials, we employ the multi-index notation: for a monomial in $n$ variables corresponding to the multi-index $\bm{j} = [j_1\,\dots\,j_n]^{\top} \in \mathbb{Z}_+^n$, we write $\bm{x}^{\bm{j}} = \prod_{i=1}^{n} x_i^{j_i}$. The indicator function of a set $A$ is denoted by $\mathds{1}_{A}$. Lastly, we denote the set of $n$ times continuously differentiable functions on an interval $I\subset\mathbb{R}$ by $\mathcal{C}^n(I)$ while the set of absolutely continuous functions will be denoted by $\mathcal{AC}(I)$. The remaining symbols will be defined as they are introduced. 

\subsection{Problem Statement, Definitions \& Assumptions}
We consider a chemical system featuring $n$ chemical species $S_1,\dots,S_n$ undergoing $n_R$ different reactions. The system state $\bm{x}$ is encoded by the molecular counts of the individual species, i.e., $\bm{x} = [x_1\, \dots\, x_n]^{\top} \in \mathbb{Z}^n_+$. It changes in response to reaction events according to the stoichiometry:
\begin{align*}
    \nu^{-}_{1,r}S_1 + \cdots + \nu^{-}_{n,r}S_N \rightarrow \nu^{+}_{1,r} S_1 + \cdots + \nu^{+}_{n,r} S_n, \quad r = 1, \dots, n_R.  
\end{align*}
Thus, the system state changes by $\bm{\nu}_r = [\nu_{1,r}^+-\nu_{1,r}^- \ \dots  \ \nu_{n,r}^+-\nu_{n,r}^-]^{\top}\in \mathbb{Z}^n$ in response to reaction $r$. We will restrict ourselves to the framework of stochastic chemical kinetics for modeling such systems.

The notion of stochastic chemical kinetics treats the position and velocities of all molecules in the system as random variables; reactions are assumed to occur at collisions with a prescribed probability. Consequently, the evolution of the system state is a continuous-time jump process. Here, we will assume that this jump process can be described by the Chemical Master Equation (CME).
\begin{assumption}\label{asmpt:CME}
	Let $P_{\pi}(\bm{x},t)$ be the probability to observe the system in state $\bm{x}$ at time $t$ given the distribution $\pi$ of the initial state of the system. Then, $P_{\pi}(\bm{x},t)$ satisfies
	\begin{multline}
		\frac{\partial P_{\pi}}{\partial t}(\bm{x},t) = \sum_{r=1}^{n_R} a_r(\bm{x}-\bm{\nu}_r) P_{\pi}(\bm{x}-\bm{\nu}_r,t) - a_r(\bm{x}) P_{\pi}(\bm{x}, t), \\ P_{\pi}(\cdot,0) = \pi \tag{CME} \label{eq:CME} 
	\end{multline}
	where $a_r$ denotes the propensity of reaction $r$, i.e., in state $\bm{x}$, $a_r(\bm{x})dt$ quantifies the probability that reaction $r$ occurs in $[t,t+dt)$ as $dt\rightarrow 0$. 
\end{assumption}
Moreover, we will restrict our considerations to the case of polynomial reaction propensities.
\begin{assumption}\label{asmpt:prop}
	The reaction propensities $a_r$ in \eqref{eq:CME} are polynomials.
\end{assumption}
To ensure the moment trajectories remain well-defined at all times, we will further assume that the stochastic process is well-behaved in the following sense.
\begin{assumption}\label{asmpt:boundedness}
	The number of reaction events occurring in the system within finite time is finite with probability 1. 
\end{assumption}
A consequence of Assumption \ref{asmpt:boundedness} is that the continuous-time jump processes associated with \eqref{eq:CME} is regular \cite{resnick1992adventures}, i.e., it does not explode in finite time. We wish to emphasize that Assumptions \ref{asmpt:CME} -- \ref{asmpt:boundedness} are rather weak; Assumptions \ref{asmpt:CME} and \ref{asmpt:prop} are in line with widely accepted microscopic models \cite{gillespie1992rigorous} while Assumption \ref{asmpt:boundedness} should intuitively be satisfied for any practically relevant system for which the CME is a reasonable modeling approach. Furthermore, Assumption \ref{asmpt:boundedness} is formally necessary for \eqref{eq:CME} to be valid on $\R_+$ \cite{resnick1992adventures}. For a detailed, physically motivated derivation of the CME alongside discussion of the underlying assumptions and potential relaxations thereof, the interested reader is referred to \citet{gillespie1992rigorous}.

Instead of studying the probability distribution $P_{\pi}$ as a description of the system behavior, in this paper we will focus on its moments defined as follows.
\begin{definition}\label{def:moment}
	Let $X$ be the reachable set of the system, i.e., $X = \{ \bm{x} \in \mathbb{Z}_+^n \mid\exists t \geq 0 :  P_{\pi}(\bm{x},t) > 0 \}$, and $\bm{j} \in \mathbb{Z}_+^n$ be a multi-index. The $\bm{j}$\textsuperscript{th} moment of $P_{\pi}(\cdot,t)$ is defined as $y_{\bm{j}}(t) = \sum_{x\in X} \bm{x}^{\bm{j}} P_{\pi}(\bm{x},t)$. $y_{\bm{j}}$ is said to be of order $|\bm{j}| = \sum_{i=1}^n j_i$. The function $y_{\bm{j}}(\cdot)$ is called the trajectory of the $\bm{j}$\textsuperscript{th} moment.  
\end{definition}
Additionally, it will prove useful to introduce the following notion of generalized moments.
\begin{definition}\label{def:generalized_moment}
	Let $y_{\bm{j}}$ be as in Definition \ref{def:moment} and $t_T > 0$. Consider a uniformly bounded Lebesgue integrable function $g:[0,t_T] \rightarrow \mathbb{R}$. The $\bm{j}$\textsuperscript{th} generalized moment of $P_{\pi}$ with respect to $g$ is defined by $z_{\bm{j}}(g;t) = \int_{0}^{t} g(\tau)y_{\bm{j}}(\tau)d\tau$ for $t\in [0,t_T]$. We say $g$ is a test function and generates $z_{\bm{j}}(g;t)$. 
\end{definition}

Under Assumptions \ref{asmpt:CME} and \ref{asmpt:prop}, it is well-known that the dynamics of the $\bm{j}$\textsuperscript{th} moment are described by a linear time-invariant ordinary differential equation (ODE) of the form 
\begin{align}
\frac{dy_{\bm{j}}}{dt}(t)  = \sum_{|\bm{k}|\leq |\bm{j}| + q} c_{\bm{k}} y_{\bm{k}}(t) = \bm{c}^{\top} \bm{y}(t) \label{eq:single_moment}
\end{align}
where $q = \max_{1 \leq r \leq n_R} \text{deg}(a_r) - 1$. The coefficient vector $\bm{c}$ can be readily computed from the reaction propensities and stoichiometry; see for example \citet{gillespie2009moment} for details. For $q>0$, it is clear from \eqref{eq:single_moment} that the dynamics of moments of a certain order in general depend on moments of a higher order. This issue is commonly termed the {\em moment closure problem}. If we denote by $\bm{y}_L$ the vector of ``lower'' order moments up to a specified order, say $m$, and by $\bm{y}_H$ the vector of ``higher'' order moments of order $m+1$ to $m+q$, it is clear from \eqref{eq:single_moment} that we obtain a linear time-invariant ODE system of the form
\begin{align}
	\frac{d\bm{y}_L}{dt}(t) = \bm{A}_L \bm{y}_L(t) + \bm{A}_H \bm{y}_H(t) \nonumber 
\end{align}
with $\bm{A}_L \in \mathbb{R}^{n_L\times n_L}$ and $\bm{A}_H \in\mathbb{R}^{n_L\times n_H}$ where $n_L={n+m \choose n}$ and $n_H = {n+m+q \choose n} - n_L$ denote the number of lower and higher order moments, respectively. For the sake of a more concise notation, throughout we will often omit these subscripts and instead write 
\begin{align}
	\bm{K} \frac{d\bm{y}}{dt}(t) = \bm{A} \bm{y}(t)  \tag{mCME} \label{eq:mCME}
\end{align}
where $\bm{A} = \left[\bm{A}_L \,\, \bm{A}_H\right]$, $\bm{K} = \left[\bm{I}_{n_L \times n_L} \,\, \bm{0}_{n_L \times n_H}\right]$ and $\bm{y}= \left[\bm{y}_L^{\top}\, \bm{y}_H^{\top}\right]^{\top}$. 

In the presence of the moment closure problem, it is clear from the setup of Equation \eqref{eq:mCME} that it does not provide sufficient information to determine uniquely the moment trajectories associated with the solution of \eqref{eq:CME}. In the following, we therefore address the question of how to compute hard, theoretically guaranteed bounds on the true moment trajectory $y_{\bm{j}}(\cdot)$ associated with the solution of \eqref{eq:CME} in this setting. To that end, we build on the work of \citet{dowdy2018dynamic} who have recently proposed an approach to answer this question. In broad strokes, they generate upper and lower bounds by optimizing a moment sequence truncated at a given order subject to a set of {\em necessary moment conditions}, i.e., conditions that the true moment trajectories are guaranteed to satisfy. By increasing the truncation order, the bounds can be successively improved. Our contribution is an extension of Dowdy and Barton's work in the form of a hierarchy of new necessary moment conditions. We show that these conditions provide additional, more scalable bound tightening mechanisms beyond increasing the truncation order and moreover give rise to highly structured optimization problems that can potentially be solved more efficiently than the unstructured problems arising in Dowdy and Barton's method. 

\subsection{Necessary Moment Conditions}\label{sec:NMC}
The bounding method proposed by \citet{dowdy2018dynamic} hinges on necessary moment conditions which restrict the set of potential solutions of \eqref{eq:mCME} as much as possible, yet allow efficient computation. Necessary moment conditions in the form of affine equations and linear matrix inequalities (LMI) have proved to fit that bill. Conditions of this form are of particular practical value as they allow for the computation of the desired bounds via semidefinite programming (SDP). As shown by \citet{dowdy2018dynamic} such affine equations arise from the system dynamics while the LMIs reflect constraints on the support of the underlying probability distribution. In the following, we will sketch their derivation and summarize the key properties that will be leveraged in Section \ref{sec:tighter_bounds} to construct additional necessary moment conditions and establish their properties.

\subsubsection{Linear Matrix Inequalities}
As claimed above, the fact that the solution of the CME $P_{\pi}(\cdot,t)$ is a non-negative measure on $\mathbb{R}^n$ and supported only on $X$ implies that its truncated moment sequences satisfy certain LMIs \cite{lasserre2001global,sakurai2017convex,ghusinga2017exact,dowdy2018bounds,kuntz2019bounding}. The following argument reveals this fact: Consider a polynomial $f \in \mathbb{R}[\bm{x}]$ that is non-negative on $X$; further, let $\bm{b}$ be a vector polynomial obtained by arranging the elements of the monomial basis of the polynomials up to degree $d = \lfloor \frac{m + q - \text{deg}(f)}{2} \rfloor$ in a vector. 
Then, the following generalized inequality where $\E$ denotes the expectation with respect to $P_{\pi}(\cdot, t)$ follows immediately  
\begin{multline*}
\E\big[f\bm{b}\bm{b}^{\top} \big] \succeq  P_{\pi}(\hat{\bm{x}},t)f(\hat{\bm{x}})\bm{b}(\hat{\bm{x}})\bm{b}(\hat{\bm{x}})^{\top} \succeq \bm{0}, \\
\forall (\hat{\bm{x}},t) \in \mathbb{R}^n\times \R_+.
\end{multline*}
It is easy to verify that the above relation can be concisely written as an LMI involving the moment trajectory of $P_\pi$. Concretely, we can write
\begin{align}
	\bm{M}_f(\bm{y}(t)) \succeq \bm{0} \tag{LMI} \label{eq:LMI} 
\end{align}
where $\bm{M}_f:\mathbb{R}^{n_L+n_H} \to \mathbb{S}^{{n+d\choose d}}$ is a {\em linear} map. The precise structure of $\bm{M}_f$ depends on $f$ and is immaterial for all arguments presented in this paper; however, the interested reader is referred to \citet{lasserre2010moments} or \citet{dowdy2018bounds} for a detailed and formal description of the structure of $\bm{M}_f$. As clear from the above argument, the construction of valid LMIs of the form \eqref{eq:LMI} relies merely on polynomials that are non-negative on $X$. For stochastic chemical systems, natural choices of such polynomials include $f(\bm{x}) = 1$ and $f(\bm{x}) = x_i$ for $i = 1,\dots,n$,\cite{sakurai2017convex,ghusinga2017exact,dowdy2018bounds,kuntz2019bounding} reflecting that $P_{\pi}$ is non-negative and in particular not supported on states with negative molecular counts, respectively. More generally, the support of $P_{\pi}(\cdot, t)$ on any basic closed semialgebraic set can be reflected this way, most importantly including the special cases of polyhedra and bounded integer lattices. To account for this flexibility while simplifying notation, we will make use of the following definition and shorthand notation.
\begin{definition}\label{def:cone}
	Let $f_0, \dots, f_{n_p}$ be polynomials that are non-negative on the reachable set $X$. The convex cone described by these LMIs is denoted by $C(X)$, i.e., $C(X) = \{ \bm{y} \in \mathbb{R}^{n_L+n_H} \mid \bm{M}_{f_i}(\bm{y}) \succeq \bm{0}, \ i = 0,\dots,n_p \}$.
\end{definition}

Lastly, we note that the validity of LMIs of the form \eqref{eq:LMI} carries over to the generalized moments that are generated by {\em non-negative} test functions. To see this, observe that the linearity of $\bm{M}_{f}$ implies that
\begin{align*}
	\bm{M}_{f}(\bm{z}(g;t))  = \int_{0}^{t} g(\tau) \bm{M}_{f}(\bm{y}(\tau)) \, d\tau 
\end{align*}
holds. Now assuming $g$ is non-negative on $\R_+$ and applying Jensen's inequality to the extended convex indicator function of the positive semidefinite cone, $\mathds{1}^{\infty}_{\mathbb{S}_+}$, therefore yields
\begin{align*}
    0 \leq \mathds{1}^{\infty}_{\mathbb{S}_+}\rb{\bm{M}_{f}(\bm{z}(g;t))} \leq  \int_{0}^t g(\tau) \mathds{1}^{\infty}_{\mathbb{S}_+}\rb{\bm{M}_{f}(\bm{y}(\tau))} \, d\tau = 0
\end{align*}
and hence $\bm{M}_{f}(\bm{z}(g;t)) \succeq \bm{0}$ must hold for any $t \geq 0$ in analogy to \eqref{eq:LMI}.

\subsubsection{Affine Constraints}
As noted in the beginning of this section, the moment dynamics \eqref{eq:mCME} give rise to affine constraints that the moments and generalized moments must satisfy. To see this, consider a test function $g \in \mathcal{AC}([0,t_T])$ and final time $t_f \leq t_T$. Then, as proposed by \citet{dowdy2018dynamic}, integrating $\int_0^{t_f} g(t)\frac{d\bm{y}_L}{dt}(t) \, dt$ by parts yields the following set of affine equations
\begin{align}
	\bm{K}\left(g(t_f)\bm{y}(t_f) - g(0) \bm{y}(0)\right)= \bm{A} \bm{z}(g;t_f) + \bm{K} \bm{z}(g';t_f). \label{eq:affine_dynamics}
\end{align}
We wish to emphasize here that the above constraints are vacuous if $\bm{z}(g;t)$ and $\bm{z}(g';t)$ are no further restricted. This observation motivates necessary restrictions on $g$ to generate ``useful" generalized moments. Recalling the discussion in Section \ref{sec:NMC}, one may be tempted to argue that $g$ and $g'$ shall be non-negative (or non-positive) on $[0,t_f]$ so that the generated generalized moments satisfy LMIs of the form \eqref{eq:LMI}. In fact, \citet{dowdy2018dynamic} as well as \citet{sakurai2018bounding} demonstrate that this is indeed a reasonable strategy; they use exponential and monomial test functions, respectively. However, in principle a wider range of test functions can be used. We defer the discussion of this issue to Section \ref{sec:tighter_bounds}.

\section{Tighter Bounds}\label{sec:tighter_bounds}
\subsection{An Optimal Control Perspective}
Some of the conservatism in the original method of \citet{dowdy2018dynamic} stems from the fact that the moments are only constrained in an {\em integral} or {\em weak} sense, i.e., $\bm{z}(g;t_f) = \int_0^{t_f} g(\tau) \bm{y}(\tau) \, d\tau$ is constrained as opposed to $\bm{y}(t)$ for all $t \in [0,t_f]$. This is potentially a strong relaxation as in fact the entire trajectory must satisfy the necessary moment conditions. Moreover, by Assumption \ref{asmpt:boundedness}, the moment trajectories remain bounded at all times, which, taken together with the fact that they satisfy the ODE \eqref{eq:mCME}, shows that they are guaranteed to be infinitely differentiable. Using these two additional pieces of information, we argue that the following continuous-time optimal control problem provides an elementary starting point for addressing the question of how to bound the moment trajectories associated with a stochastic chemical system evaluated at a given time point $t_f$: 
\begin{align}
    \inf_{\bm{y} \in \mathcal{C}^{\infty}(\R_+)} \quad &y_{\bm{j}}(t_f) \tag{OCP} \label{eq:OCP}\\
    \text{s.t.} \quad &\frac{d\bm{y}_L}{dt}(t) = \bm{A}_L\bm{y}_L(t) + \bm{A}_H\bm{y}_H(t),\quad \forall t \in \R_+, \nonumber \\
    &\bm{y}(0) = \bm{y}_{0}, \nonumber \\
    &\bm{y}(t) \in C(X), \quad \forall t \in \R_+. \nonumber 
\end{align}
Here, the ``lower'' order moments $\bm{y}_L$ act as the state variables while the ``higher'' order moments $\bm{y}_H$ can be viewed as control inputs. Although the infinite dimensional nature of Problem \eqref{eq:OCP} leaves it with little immediate practical relevance, this representation is conceptually informative. In fact, it is not hard to verify that the method proposed by \citet{dowdy2018dynamic} provides a systematic way to construct tractable relaxations of \eqref{eq:OCP} in the form of SDPs. However, Dowdy and Barton's method does in no way reflect the dependence of $\bm{y}(t_f)$ on past values of $\bm{y}(t)$ other than $\bm{y}(0)$ nor the fact that $\bm{y} \in \mathcal{C}^{\infty}(\R_+)$. As we will show in the following, these observations motivate new necessary moment conditions giving rise to a hierarchy of tighter SDP relaxations of Problem \eqref{eq:OCP} than those constructed by Dowdy and Barton's method\cite{dowdy2018dynamic}.

% Temporal Constraints
\subsection{A New Hierarchy of Necessary Moment Conditions}\label{sec:hierarchies}
In this section, we present the key contribution of this article -- a new hierarchy of convex necessary moment conditions that reflect the temporal causality and regularity conditions inherent to the moment trajectories associated with the distribution described by the CME. To provide some intuition for these results, we will first discuss some special cases of the proposed conditions which permit a clear interpretation. To that end, recall that the moment trajectory $y_{\bm{j}}(\cdot)$ must be infinitely differentiable on $\mathbb{R}_+$ as all moment trajectories remain bounded by Assumption \ref{asmpt:boundedness} and obey the linear time-invariant dynamics \eqref{eq:mCME}. As a consequence, the Taylor polynomial
\begin{align*}
	\bm{T}_l(y_{\bm{j}};t_1,t_2) &= \sum_{k=0}^l \frac{(t_2 - t_1)^k}{k!} \frac{d^ky_{\bm{j}}}{dt^k}(t_1) 
\end{align*}
and remainder
\begin{align*}
	\bm{R}_l(y_{\bm{j}};t_1,t_2) & = \frac{1}{l!} \int^{t_2}_{t_1} (t_2-t)^{l} \frac{d^{l+1} y_{\bm{j}}}{dt^{l+1}}(t) \, dt
\end{align*}
are well-defined for any $0 \leq t_1 \leq t_2 < +\infty$ and order $l \geq 0$. A key observation here is that if $|\bm{j}|$ and $l$ are sufficiently small\footnote{if $|\bm{j}|$ or $l$ grow too large, the Taylor polynomial or remainder may depend on moments of higher order than $m+q$ but the dependence will still be linear}, then $\bm{T}_l(y_{\bm{j}};t_1,t_2)$ and $\bm{R}_l(y_{\bm{j}};t_1,t_2)$ depend linearly on $\bm{y}(t_1)$ and $\bm{z}(g_l;t_2)$ with $g_l(t)=\mathds{1}_{[t_1,t_2]}(t)(t_2-t)^l$, respectively. Formally, we can write  
\begin{align}
    \begin{array}{l}
         \bm{T}_l(y_{\bm{j}};t_1,t_2)  =\bm{c}_{l,\bm{j}}(t_1,t_2)^\top \bm{y}(t_1) \\[1em]
         \bm{R}_l(y_{\bm{j}};t_1,t_2) = \bm{d}_{l,\bm{j}}(t_1,t_2)^\top\bm{z}(g_l;t_2) 
    \end{array} \label{eq:taylor_coeff}
\end{align}
for an appropriate choice of the coefficient vectors.

Overall, this observation suggests to employ conditions of the form
\begin{align*}
	y_{\bm{j}}(t_2) = \bm{T}_l(y_{\bm{j}};t_1,t_2) + \bm{R}_l(y_{\bm{j}};t_1,t_2) 
\end{align*}
at different time points along the trajectory as necessary moment conditions. In fact these conditions achieve exactly what we set out to do: they establish a connection between $y_{\bm{j}}(t_2)$ and its past using the smoothness properties of the trajectory $y_{\bm{j}}$. Further, it is straightforward to see that analogous conditions are readily obtained for any generalized moment generated by a sufficiently smooth test function. The above conditions hence appear to be a promising starting point. From a practical perspective, however, they merely suggest a particular choice of test functions as revealed by the following proposition.
	\begin{proposition}\label{prop:increment_conditions}
		Let $0 \leq t_1 \leq t_2 < + \infty$ and $n_I \leq \floor{\frac{m}{q}}$. Further, consider test functions of the form $g_l(t) = \mathds{1}_{[t_1,t_2]}(t)(t_2-t)^l$. If $\bm{y}_{t_1}, \bm{y}_{t_2} \in \R^{n_L+n_H}$ and $\bm{z}_{g_l,t_2} \in \R^{n_L+n_H}$ satisfy 
		\begin{align}
		\bm{K}\left(g_l(t_2) \bm{y}_{t_2} - g_l(t_1)\bm{y}_{t_1} \right) = \bm{A} \bm{z}_{g_l,t_2} - l\bm{K}\bm{z}_{g_{l-1},t_2} \label{eq:taylor_refined}
		\end{align}
		for $l = 0,\dots,n_I$, then $\bm{y}_{t_1}, \bm{y}_{t_2}$ and $\bm{z}_{g_l,t_2}$ also satisfy 
		\begin{align}
		    y_{\bm{j},t_2} = \bm{c}_{l,\bm{j}}(t_1,t_2)^\top\bm{y}_{t_1} + \bm{d}_{l,\bm{j}}(t_1,t_2)^\top\bm{z}_{g_l,t_2} \label{eq:taylor}
		\end{align} 
		for $l=0,\dots,n_I$ and $\bm{j}$ such that $|\bm{j}|\leq m-lq$, where $\bm{c}_{l,\bm{j}}$ and $\bm{d}_{l,\bm{j}}$ are defined as in \eqref{eq:taylor_coeff}. 
	\end{proposition}
\begin{proof}
	The proof is deferred to Supplementary Information (SI). 
\end{proof}
\begin{remark}
	Note that Condition \eqref{eq:taylor_refined} is analogous to Condition \eqref{eq:affine_dynamics} as obtained for the test function $g_l$ with a shifted origin, hence it is a necessary moment condition. Further, we wish to emphasize that Condition \eqref{eq:taylor_refined} is in general more stringent than Condition \eqref{eq:taylor} as is made clear in the proof.
\end{remark}
Beyond a specific choice of test functions, the above considerations motivate a broader strategy to generate necessary moment conditions that reflect causality. This strategy can be summarized as ``discretize and constrain''. Instead of imposing Condition \eqref{eq:affine_dynamics} on the entire time horizon $[0,t_f]$ as proposed by \citet{dowdy2018dynamic}, the time horizon can be partitioned into $n_{\mathsf{T}}$ subintervals $[t_{i-1}, t_i]$ with $0= t_0< t_1 < \cdots < t_{n_{\mathsf{T}}} = t_f$ on which analogous conditions obtained from integrating $\int_{t_{i-1}}^{t_i} g(\tau) \bm{K} \frac{d\bm{y}}{dt}(\tau) \, d\tau$ by parts can be imposed:
\begin{multline}
	\bm{K}\left(g(t_i) \bm{y}(t_i) - g(t_{i-1})\bm{y}(t_{i-1}) \right) = \\ \bm{A} (\bm{z}(g;t_i) - \bm{z}(g;t_{i-1}) )+ \bm{K}(\bm{z}(g';t_i) - \bm{z}(g';t_{i-1})). \label{eq:increment_affine}
\end{multline} 
While by itself this does not provide any restriction over Condition \eqref{eq:affine_dynamics}, the following observation makes it worthwhile: the generalized moments generated by a non-negative test function $g$ form a monotonically increasing sequence with respect to the convex cone $\G$. This follows immediately from the definition of $\bm{z}(g;t)$ and Jensen's inequality as described in Section \ref{sec:NMC}; more formally, 
\begin{align}
	\bm{z}(g;t_i) - \bm{z}(g;t_{i-1}) \in \G, \quad i = 1,\dots,n_{\mathsf{T}} \label{eq:increment_conic}
\end{align}
are necessary moment conditions. The Conditions \eqref{eq:increment_affine} \& \eqref{eq:increment_conic} are generally a non-trivial restriction of the Conditions \eqref{eq:affine_dynamics} \& $\bm{z}(g;t_f) \in \G$ as employed by \citet{dowdy2018dynamic}. To see this, simply observe that we recover Equation \eqref{eq:affine_dynamics} by summing the Equations \eqref{eq:increment_affine} over $i=1,\dots,n_{\mathsf{T}}$ and likewise obtain 
\begin{align*}
    \bm{z}(g;t_f) = \sum_{i=1}^{n_{\mathsf{T}}}\bm{z}(g;t_i) - \bm{z}(g;t_{i-1}) \in \G,
\end{align*}
using that $\G$ is a convex cone and $\bm{z}(g;0) = \bm{0}$ by definition. 

The above described strategies lend themselves to generalization in terms of a hierarchy of necessary moment conditions. This generalization can be performed in several equivalent ways. Next we will present one such generalization utilizing a concept which we refer to as iterated generalized moments: 
\begin{definition} \label{def:iterated_generalized_moments}
	Let $z_{\bm{j}}(g;t)$ be the $\bm{j}$\textsuperscript{th} generalized moment as per Definition \ref{def:generalized_moment}. Then, the iterated generalized moment of Level $l \geq 0$ is defined by
	\begin{align}
		z^l_{\bm{j}}(g;t) = \begin{cases}
						\int_{0}^t z_{\bm{j}}^{l-1}(g;\tau)d\tau, &\quad l \geq 1 \\
						g(t) y_{\bm{j}}(t), &\quad l=0
				   \end{cases}.\nonumber 
	\end{align}
\end{definition}
For the sake of simplified notation and analysis, it will further prove useful to introduce the following left and right integral operators $I_{L},I_{R}: \mathcal{C}(\R^2) \maps \mathcal{C}(\R^2)$ given by
\begin{align*}
	(I_L f)(t_1,t_2) = \int_{t_1}^{t_2} f(t_1, t) \, dt\\
	(I_R f)(t_1,t_2) = \int_{t_1}^{t_2} f(t, t_2) \, dt.
\end{align*}
For vector-valued functions, $I_L$ and $I_R$ shall be understand as being applied componentwise.  

With these two concepts in hand, the following proposition formalizes the proposed hierarchy of necessary moment conditions.
\begin{proposition}\label{prop:integral_hierarchy}
	Let $ t_T > 0$ and consider a non-negative test function $g \in \mathcal{AC}([0,t_T])$. Further, let $\bm{y}$ be the truncated sequence of moment trajectories associated with the solution of \eqref{eq:CME}, and $\bm{z}^l$ the corresponding iterated generalized moments. Then, the following conditions hold for any $l \geq 1$:
	\begin{enumerate}[(i)]
		\item For any $t\in [0,t_T]$
		\begin{multline*}
			\bm{A}\bm{z}^l(g;t) + \bm{K} \bm{z}^l(g';t) =\\ \bm{K} \left(\bm{z}^{l-1}(g;t) - \frac{t^{l-1}}{(l-1)!} g(0) \bm{y}(0)\right)
		\end{multline*} 
		\item Let $\bm{f}(x,y) = \bm{z}^1(g;y) - \bm{z}^1(g;x)$. Then, for any $0 \leq t_1 \leq t_2 \leq t_T$ and $k \in \St{0,\dots,l-1}$ 
		\begin{align*}
			(I_L^{l-1-k}I^{k}_R \bm{f})(t_1,t_2) \in \G.
		\end{align*}
	\end{enumerate}
\end{proposition}
\begin{proof}
	It is easily verified that Condition (i) is obtained from integrating Equation \eqref{eq:affine_dynamics} $l-1$ times. Validity of Condition (ii) follows by a similar inductive argument: Since $\bm{y}(t) \in \G$ for all $t \in [0,t_T]$, it follows by non-negativity of $g$ on $[0,t_T]$ and Jensen's inequality that
	\begin{align*}
		\bm{z}^1(g;t_2) - \bm{z}^1(g;t_1) = \int_{t_1}^{t_2} g(t) \bm{y}(t)\, dt  \in \G
	\end{align*}
	for any $0 \leq t_1 \leq t_2 \leq t_T$. Now suppose Condition (ii) is satisfied for $l-1$. Then, it follows by Jensen's inequality that for any $0 \leq t_1 \leq t_2 \leq t_T$ and $k=0,\dots,l-2$
	\begin{align*}
		(I_L^{l-1-k} I_R^{k}\bm{f})(t_1,t_2) = \int_{t_1}^{t_2} (I_L^{l-2-k} I_R^{k}\bm{f})(t_1,t) \, dt \in \G.
	\end{align*}	
	For $k = l-1$, an analogous argument applies.% This concludes the proof.
\end{proof}
Before we proceed, a few remarks are in order to contextualize this result. 
\begin{remark}\label{rmk:dowdy_special_case}
	Choosing $l=1$, $t_1=0$ and $t_2 = t_f$ reproduces the necessary moment conditions proposed by \citet{dowdy2018dynamic}.
\end{remark}
\begin{remark}
     Regarding Condition (ii), one might be tempted to argue that any permutation of the operator products $I_L$ and $I_R$ of length $l-1$ applied to $\bm{f}(x,y) = \bm{z}^1(g;y) - \bm{z}^1(g;x)$ gives rise to a new valid necessary moment condition. It can be confirmed, however, that $I_L$ and $I_R$ commute such that Condition (ii) is invariant under permutation of $I_L$ and $I_R$; a more detailed discussion of this claim can be found in the SI.
\end{remark}
\begin{remark}
	We wish to emphasize that Conditions (i) and (ii) depend affinely on the iterated generalized moments up to Level $l$ evaluated at $t_1$ and $t_2$, respectively. Accordingly, they preserve the computational advantages of the original necessary moment conditions. To avoid notational clutter in the remainder of this article, however, we will disguise this fact and concisely denote the left-hand-side of Condition (ii) by 
	\begin{align*}
		\Omega_{l,k}\left(\Set{\bm{z}^i(g;t_1)}_{i=1}^{l}, \Set{\bm{z}^i(g;t_2)}_{i=1}^{l}, t_1, t_2\right). 
	\end{align*}
	 An explicit algebraic expression for $\Omega_{l,k}$ is provided in the SI.
\end{remark}
\begin{remark}
	For the $\bm{0}$-th order moments, additional constraints arise from the definition as 
	\begin{align*}
		z_{\bm{0}}^l(g;t) = \begin{cases}
			\int_0^t z_{\bm{0}}^{l-1}(g;\tau)d\tau, \ &l \geq 1\\
			g(t), \ &l=0
		\end{cases} 
	\end{align*}
	can be evaluated explicitly. 
\end{remark}
It is crucial to mention that Condition (i) in Proposition \ref{prop:increment_conditions} is effectively unrestrictive unless $\bm{z}^l(g';t)$ can be further constrained. The following proposition provides a concrete guideline which test functions allow to circumvent this issue.
\begin{proposition}\label{prop:test_functions}
	Let $F$ be a finite set of test functions such that $\spn{F}$ is closed under differentiation. Then, for any $f \in F$, there exists a linear map $\Gamma_f$ such that Condition (i) in Proposition \ref{prop:integral_hierarchy} is equivalent to
	\begin{align*}
	    \Gamma_f\left(\St{\bm{z}^l(g;t)}_{g\in F} \right) = \bm{K} \left(\bm{z}^{l-1}(f;t) -\frac{t^{l-1}}{(l-1)!} f(0) \bm{y}(0) \right).
	\end{align*}
\end{proposition}
We omit the elementary proof of Proposition \ref{prop:test_functions} and instead provide a concrete example that shows how to construct the maps $\Gamma_f$ for a given set of exponential test functions. 
\begin{example}\label{ex:test_functions}
    Let $F = \Set{e^{\rho_i t}}_{i=1}^{n_f}$ for some fixed $\rho_i \in \mathbb{R}$. Clearly $\spn{F}$ is closed under differentiation as for any $c_i \in \R^{n_f}$, we have
    \begin{align*}
        \frac{d}{dt}\left( \sum_{i=1}^{n_f} c_i e^{\rho_i t} \right) = \sum_{i=1}^{n_f} c_i \rho_i e^{\rho_i t}.
    \end{align*}
    Now let $f(t) = e^{\rho_i t}$. Comparing with Condition (i) in Proposition \ref{prop:integral_hierarchy} shows that the map $\Gamma_f$ is defined by
    \begin{align*}
	    \Gamma_f\left(\St{\bm{z}^l(g;t)}_{g\in F} \right) = \bm{A}\bm{z}^l(f;t) + \bm{K} \bm{z}^l(f';t).
    \end{align*}
    By definition of the iterated generalized moments and the fact that $f'(t) = \rho_i f(t)$, it follows further that $\bm{z}^l(f';t) = \rho_i \bm{z}^l(f;t)$. Thus, 
    \begin{align*}
        \Gamma_f\left(\St{\bm{z}^l(g;t)}_{g\in F} \right)  = \left(\bm{A} - \rho_i \bm{K}\right) \bm{z}^l(f;t). 
    \end{align*}
%Example \ref{ex:test_functions} shows that the significance of the hypothesis that $F$ is closed under differentiation merely lies in the consequence that $\bm{z}(f';t)$ can be represented as a linear combination of other test functions.
\end{example} 
Example \ref{ex:test_functions} indicates the significance of the hypotheses of Proposition \ref{prop:test_functions}. In particular, it emphasizes that the closedness of $\spn{F}$ under differentiation is precisely what is needed in order to guarantee that the associated necessary moment conditions described in Proposition \ref{prop:integral_hierarchy} are ``self-contained''; that is, the conditions only depend on generalized moments as generated by test functions in $F$. It is further noteworthy that there exist rich function classes beyond exponentials that can be used to assemble test function sets that satisfy the hypotheses of Proposition \ref{prop:test_functions}, for example polynomials and trignometric functions. 

Another issue is that we require non-negativity of the test functions in Proposition \ref{prop:integral_hierarchy}. This problem can be alleviated by a simple reformulation and shift of the time horizon in Proposition \ref{prop:integral_hierarchy}. For example, if a test function $g$ is non-negative on $[0,t_+]$ and non-positive on $[t_+,t_T]$, we can simply consider the two test functions $g_+(t) = \mathds{1}_{[0,t_+]}(t) g(t)$ and $g_-(t) = -\mathds{1}_{[t_+,t_T]}(t) g(t)$ in place of $g$ and impose the necessary moment conditions on the intervals $[0,t_+]$ and $[t_+,t_T]$, respectively. This construction naturally extends to test functions with any finite number of sign changes.

We conclude this section by establishing some compelling properties of the hierarchy of necessary moment conditions put forward in Proposition \ref{prop:integral_hierarchy}. On the one hand, Conditions (i) and (ii) in Proposition \ref{prop:integral_hierarchy} include the conditions considered in Proposition \ref{prop:increment_conditions} as special cases. So in particular, they enforce consistency with higher-order Taylor expansions of the true moment trajectories as discussed in the beginning of this section. The following corollary to Proposition \ref{prop:integral_hierarchy} formalizes this claim.
\begin{corollary}\label{cor:implications} 
	Let $n_I \in\mathbb{Z}_+$ and $t_T > 0$ be fixed. Further, suppose $g \in \mathcal{AC}([0,t_T])$ is non-negative, and let $\bm{y}$ and $\bm{z}^l(g;\cdot)$ be arbitrary functions such that $\bm{z}^l(g;\cdot)$ is linear in the first argument and $\bm{z}^0(g;t)= g(t)\bm{y}(t)$ holds. Fix $0 \leq t_1 \leq t_2 \leq  t_T$ and define $h_l(t) = \mathds{1}_{[t_1,t_2]}(t)(t_2-t)^l$ for $l=0,1,\dots,n_I$. If Conditions (i) and (ii) of Proposition \ref{prop:integral_hierarchy} are satisfied by $\St{\bm{z}^l(g;t_i)}_{l=0}^{n_I+1}$ for $i=1,2$, then there exist functions $\bm{z}(h_lg;\cdot)$ that are linear in the first argument, and satisfy
		\begin{multline*}
			\bm{K}\left(h_l(t_2)g(t_2)\bm{y}(t_2) - h_l(t_1)g(t_1)\bm{y}(t_1) \right) = \\
		    \bm{A} \bm{z}(h_l g;t_2)  + \bm{K}\bm{z}((h_lg)';t_2)
		\end{multline*}
		and
		\begin{align*}
			\bm{z}(h_l g; t_2) \in \G 
		\end{align*}
		for all $l \in \St{0,\dots,n_I}$. 
\end{corollary} 
\begin{proof}
	The proof is deferred to the SI.
\end{proof}
\begin{remark}
    To see the connection to Condition \eqref{eq:taylor_refined} in Proposition \ref{prop:increment_conditions}, simply consider the case where $g(t) = 1$. Moreover, note that Corollary \ref{cor:implications} also shows that necessary moment conditions of the form of \eqref{eq:increment_affine} \& \eqref{eq:increment_conic} are implied as they are recovered for $l=0$ since we can simply identify $\bm{z}(h_0g;t_2)$ with $\bm{z}^1(g;t_2) - \bm{z}^1(g;t_1)$.  
\end{remark}

On the other hand, the proposed necessary moment conditions display benign scaling behavior in the following sense:
\begin{corollary}\label{cor:linear_scaling}
	Let $0 \leq t_1 \leq t_2 \leq t_3 < +\infty$ and $n_I$ be a fixed positive integer. Suppose $\St{\bm{z}^s}_{s=1}^{n_I}$ is a set of functions such that
	\begin{align*}
	    \Omega_{l,k}(\Set{\bm{z}^s(t_{i})}_{s=1}^{l}, \Set{\bm{z}^s(t_{i+1})}_{s=1}^{l}, t_i, t_{i+1}) \in \G
	\end{align*}
	for all $i \in \{1,2\}$ and $k,l \in \Z_+$ such that $k < l \leq n_I$. Then, 
	\begin{align*}
	\Omega_{l,k}(\Set{\bm{z}^s(t_{1})}_{s=1}^{l}, \Set{\bm{z}^s(t_{3})}_{s=1}^{l}, t_1, t_{3}) \in \G
	\end{align*}
	holds for all $k,l \in \Z_+$ such that $k < l \leq n_I$. 
\end{corollary}
\begin{proof}
	The proof is deferred to the SI.
\end{proof}
In essence, Corollary \ref{cor:linear_scaling} shows that imposing the proposed necessary moment conditions at multiple time points along a time horizon scales linearly with the number of time points considered. 

\subsection{An Augmented Semidefinite Program} \label{sec:SDP}
In this section, we construct an SDP based on Proposition \ref{prop:integral_hierarchy} whose optimal value furnishes bounds on the moment solutions of \eqref{eq:CME} at a given time point $t_f \in [0,t_T]$. To that end, we consider the truncation order $m$ to be fixed and the following user choices as known:
\begin{enumerate}[(i)]
	\item $\mathsf{T} = \St{t_1,\dots,t_{n_{\mathsf{T}}}}$ -- A finite, ordered set of time points such that $0 < t_1 < t_2 < \dots < t_{n_{\mathsf{T}}} \leq t_T$ and $t_f \in \mathsf{T}$. %For convenience we define $\mathsf{T}_0 = \mathsf{T} \cup \St{0}$. 
	\item $\mathsf{F} = \St{g_1,\dots,g_{n_{\mathsf{F}}}}$ -- A finite set of test functions that satisfies the hypotheses of Propositions \ref{prop:integral_hierarchy} and \ref{prop:test_functions}. 
	\item $n_I$ -- A non-negative integer controlling the hierarchy level in Proposition \ref{prop:integral_hierarchy}.
\end{enumerate}
These quantities parametrize a spectrahedron $\mathsf{S}(\mathsf{F},\mathsf{T}, n_I)$ described by the necessary moment conditions of Proposition \ref{prop:integral_hierarchy} as imposed for all test functions in $\mathsf{F}$, at all time points in $\mathsf{T}$ and for all hierarchy Levels up to $n_I$. In the formulation of $\mathsf{S}(\mathsf{F},\mathsf{T},n_I)$, however, we use a slightly different but equivalent formulation of Condition (i) of Proposition \ref{prop:integral_hierarchy}. The reason for this modification is that it results in weakly coupled conditions that allow the resultant SDPs to be decomposed in a natural way as we will discuss in Section \ref{sec:forgetting}. Details on this reformulation can be found in the SI. 
$\mathsf{S}(\mathsf{F},\mathsf{T},n_I)$ is explicitly stated below; for the sake of concise notation we introduced the shorthand $n(t)$ for the left neighboring point of any $t \in \mathsf{T}$, i.e., $n(t_i) = t_{i-1}$ for $i = 2,\dots,n_{\mathsf{T}}$ and $n(t_1) = 0$. 
\begin{widetext}
    \begin{align*}
		&\mathsf{S}(\mathsf{F},\mathsf{T},n_I) = &\St{  \{\bm{y}_t\}, \{\bm{z}^{l}_{g,t}\} 
         \left\rvert	
		\begin{array}{l}
			\bm{z}^0_{g,t} = g(t) \bm{y}_t, \quad \forall (g,t) \in \mathsf{F}\times \mathsf{T},\\[1em]
			\bm{y}_t \in \G, \quad \forall t \in \mathsf{T},\\[1em] 
			\Gamma_g\left(\{\bm{z}_{f,t}^{l}\}_{f \in \mathsf{F}}\right) = \begin{cases}
		        \bm{K} \left(\bm{z}^{l-1}_{g,t} - \frac{t^{l-1}}{(l-1)!}g(0) \bm{y}_0\right), &\text{if } t = t_1 \\
			    \left(\frac{t}{n(t)}\right)^{l-1}\Gamma_g \left(\{\bm{z}_{f,n(t)}^{l}\}_{f \in \mathsf{F}}\right) + \bm{K} \left(\bm{z}^{l-1}_{g,t} - \left(\frac{t}{n(t)}\right)^{l-1} \bm{z}_{g,n(t)}^{l-1} \right), &\text{if } t\neq t_1 
			\end{cases}, \\
			\qquad  \qquad  \qquad \qquad  \forall (g,t,l) \in \mathsf{F} \times \mathsf{T} \times \St{1,\dots,n_I},\\[1em]
		    \Omega_{l,k}\left(\{\bm{z}^s_{g,n(t)}\}_{s=1}^{l},\{\bm{z}^s_{g,t}\}_{s=1}^{l},n(t),t\right) \in \G, \\
			\qquad  \qquad  \qquad \qquad   \forall (g,t) \in \mathsf{F} \times \mathsf{T} \text{ and } \forall k,l \in \Z_+ \text{ such that } k < l \leq n_I			
		\end{array}\right.} 
    \end{align*}
    \vspace*{\fill}
\end{widetext}

By construction, the set $\mathsf{S}(\mathsf{F},\mathsf{T},n_I)$ contains the sequences $\{\bm{y}(t):{t\in\mathsf{T}}\}$ and $\{ \bm{z}^l(g;t) : (g,t,l) \in \mathsf{F} \times \mathsf{T} \times \{0,\dots,n_I\}\}$ as generated by the true moment trajectories associated with the solution of \eqref{eq:CME}. Another piece of information that can be used to further restrict the set of candidates for the true moment solutions to \eqref{eq:CME} is information about the moments of the initial distribution. We know for example from the definition that any iterated generalized moment $\bm{z}^l_g(t)$ for $l\geq 1$ must vanish at $t=0$. Moreover, one usually has specific information about the initial distribution of the system state, hence also about $\bm{y}(0)$. Here, we assume that the initial moments and iterated generalized moments are confined to a spectrahedral set denoted by $\mathsf{S}_0(\mathsf{F},n_I)$. In the common setting in which the moments of the initial distribution $\bm{y}_0$ are known exactly, $\mathsf{S}_0(\mathsf{F},n_I)$ would be given by
\begin{multline*}
    \mathsf{S}_0(\mathsf{F},n_I) = \\\St{\bm{y}_0, \{\bm{z}^{l}_{g,0}\} \left \rvert 
	\begin{array}{l}
	\bm{y}_0 = \bm{y}(0), \\
	\bm{z}^l_{g,0} = \bm{0}, \ \forall (g,l) \in \mathsf{F}\times \Set{1,\dots,n_I}
	%\bm{z}^0_g = g(0) \bm{y}_0, \quad \forall g \in \mathsf{F}
	\end{array}\right.
}. 
\end{multline*}
Albeit adding the corresponding constraints to the description of $\mathsf{S}(\mathsf{T},\mathsf{F},n_I)$ appears natural, we deliberately choose to reflect this piece of information separately. Our motivation for this distinction is twofold: On the one hand, we want to emphasize that the presented approach naturally extends to the setting of uncertain or imperfect knowledge of the moments of the initial distribution of the system. Specifically, if the moments of the initial distribution are not known exactly, however, known to be confined to a spectrahedral set, the proposed bounding procedure applies without modification. 
On the other hand, we will argue in Section \ref{sec:forgetting} that, based on this specific feature, the arising optimization problems lend themselves to be decomposed. The distinction in notation made here will simplify our exposition there.

The following Theorems finally summarize the key feature of the proposed methodology -- the ability to generate a sequence of monotonically improving bounds on the moment trajectories associated with the solution of \eqref{eq:CME}. Theorem \ref{thm:sdp} shows that these bounds can be practically obtained via solution of a hierarchy of SDPs. 
\begin{theorem}\label{thm:sdp}
	Let $\bm{y}(t)$ be as in Definition \ref{def:moment} and $t_f \in \mathsf{T}$. Consider a multi-index $|\bm{i}| \leq m$ and define   
	\begin{align}
		y_{\bm{i},t_f}^* = \inf_{\{\bm{y}_t\}, \{\bm{z}^{l}_{g,t}\}} \qquad  &y_{\bm{i},t_f} \tag{SDP} \label{eq:SDP}\\
		\text{s.t.}\qquad & (\{\bm{y}_t\}, \{\bm{z}^{l}_{g,t}\}) \in \mathsf{S}(\mathsf{F},\mathsf{T},n_I) \nonumber \\
		&(\bm{y}_0,\{\bm{z}^{l}_{g,0}\}) \in \mathsf{S}_0(\mathsf{F},n_I) \nonumber
	\end{align} 
	Then, $y_{\bm{i},t_f}^* \leq y_{\bm{i}}(t_f)$.
\end{theorem}
\begin{proof}
	Set $\bm{y}_t = \bm{y}(t)$ for all $t \in \mathsf{T}$ and $\bm{z}^{l}_{g,t} = \bm{z}^{l}(g;t)$ for all $(g,t,l) \in \mathsf{F}\times\mathsf{T}\times \St{0,\dots, n_I}$ with $\bm{z}^l(g;t)$ as in Definition \ref{def:iterated_generalized_moments}.
	 By Proposition \ref{prop:integral_hierarchy} $(\{\bm{y}_t\}, \{\bm{z}^{l}_{g,t}\}) \in \mathsf{S}(\mathsf{F},\mathsf{T},n_I)$, and the result follows.% and obviously $y_{\bm{i},T} = y$ 
\end{proof}
\begin{remark}
	For equal truncation orders and choice of $\G$, Remark \ref{rmk:dowdy_special_case} implies that the bounds obtained from \eqref{eq:SDP} are at least as tight as those obtained by the approach of \citet{dowdy2018dynamic}.
\end{remark}
\begin{remark}
	The lower bound $\bm{y}_{\bm{i},t_f}^*$ can be evaluated using off-the-shelf solvers for SDPs such as MOSEK \cite{andersen2000mosek}, SeDuMi \cite{sturm1999using}, or SDPT3 \cite{toh1999sdpt3}.
\end{remark}
\begin{remark}
	Similar problems as \eqref{eq:SDP} can be formulated to bound properties that can be described in terms of moments of non-negative measures on the reachable set; examples include variances \cite{dowdy2018bounds}, the volume of a confidence ellipsoids \cite{sakurai2018optimization}, and the value that the probability measure assigns to a semialgebraic set \cite{dowdy2018bounds}. 
\end{remark}
The formulation of \eqref{eq:SDP} provides several mechanisms to improve the bounds by adjusting the parameters $\mathsf{T}$, $\mathsf{F}$ and $n_I$. Theorem \ref{thm:mono} shows that appropriate adjustments lead to a sequence of monotonically improving bounds. 
\begin{theorem}\label{thm:mono}
	Let $y_{\bm{i},t_f}^*$ be defined as in Theorem \ref{thm:sdp}. Let $\tilde{t} \in [0,t_T]$ and define $\tilde{\mathsf{T}} = \mathsf{T}\cup \St{\tilde{t}}$. Further, let $\tilde{g}$ be an absolutely continuous function that is non-negative on $[0,t_T]$ and define $\tilde{\mathsf{F}} = \mathsf{F} \cup \St{\tilde{g}}$. Then, 
	\begin{align}
		y_{\bm{i},t_f}^* \leq\inf_{\{\bm{y}_t\}, \{\bm{z}^{l}_{g,t}\}} \qquad  &y_{\bm{i},t_f}\label{eq:SDP_refined} \\
		\text{s.t.}\qquad & (\{\bm{y}_t\}, \{\bm{z}^{l}_{g,t}\}) \in \mathsf{S}(\tilde{\mathsf{F}},\tilde{\mathsf{T}},n_I) \nonumber \\
		&(\bm{y}_0,\{\bm{z}^{l}_{g,0}\}) \in \mathsf{S}_0(\tilde{\mathsf{F}},n_I). \nonumber 
	\end{align}
	Likewise,
	\begin{align}
		y_{\bm{i},t_f}^* \leq\inf_{\{\bm{y}_t\}, \{\bm{z}^{l}_{g,t}\}} \qquad  &y_{\bm{i},t_f}\label{eq:SDP_refined_int_order} \\
		\text{s.t.}\qquad & (\{\bm{y}_t\}, \{\bm{z}^{l}_{g,t}\}) \in \mathsf{S}(\mathsf{F},\mathsf{T},n_I+1) \nonumber \\
		&(\bm{y}_0,\{\bm{z}^{l}_{g,0}\}) \in \mathsf{S}_0(\mathsf{F},n_I+1). \nonumber
	\end{align}
\end{theorem}
\begin{proof}
	\eqref{eq:SDP_refined} is obvious if $\tilde{t} \in \mathsf{T}$ and $\tilde{g} \in \mathsf{F}$. If $\tilde{t} \notin \mathsf{T}$ and/or $\tilde{g} \notin \mathsf{F}$, any feasible point of the right-hand-side of \eqref{eq:SDP_refined} can be used to construct a feasible point of \eqref{eq:SDP}; simply remove the decision variables that correspond to time point $\tilde{t}$ and/or test function $\tilde{g}$. Similarly, removing the iterated generalized moments of Level $n_I+1$ of the right-hand-side of \eqref{eq:SDP_refined_int_order} yields a feasible point for \eqref{eq:SDP}.  
\end{proof}
\begin{remark}
	Increasing the truncation order also gives rise to montonically improving bounds. For the sake of brevity, we omit a formal statement and proof here as many easily adapted results of this type exist; see for example Corollary 6 in \citet{kuntz2019bounding}.
\end{remark}

We conclude this section with a brief discussion of the scalability of \eqref{eq:SDP}. Table \ref{tab:scaling} summarizes how the number of variables, affine constraints and LMIs as well as their dimension scales with $\mathsf{F}$, $\mathsf{T}$, $n_I$ and the truncation order $m$. The results demonstrate the value of the proposed formulation if the number of species $n$ in the system under investigation is large. In that case, the bound tightening mechanisms offered by adjusting $\mathsf{F}$, $\mathsf{T}$ and $n_I$ scale much more moderately than increasing the truncation order. Furthermore, it should be emphasized that the invariance of LMI size with respect to $\mathsf{F}$, $\mathsf{T}$ and $n_I$ is a very desirable property to achieve scalability of SDP hierarchies in practice \cite{ahmadi2017optimization,ahmadi2019construction}. Lastly, it is worth noting that moment-based SDPs are notorious for becoming numerically ill-conditioned as the truncation order increases. Thus, the presented hierarchy provides a mechanism to circumvent this issue to some extent. 
\begin{table*}
\centering
\caption{Scaling of \eqref{eq:SDP}}\label{tab:scaling}
\begin{tabularx}{0.9\textwidth}{p{0.5cm} X X X X}%
	& \textbf{\#variables} & \textbf{\#affine \newline constraints} & \textbf{\#LMI} & \textbf{LMI size} \\[0.25em]
	\hline
	$\mathsf{F}$ & \order{|\mathsf{F}|} & \order{|\mathsf{F}|} & \order{|\mathsf{F}|} & \order{1} \\[0.5em]
	$\mathsf{T}$ &  \order{|\mathsf{T}|} & \order{|\mathsf{T}|} & \order{|\mathsf{T}|} &  \order{1}\\[0.5em]
	$n_I$ & \order{n_I}& \order{n_I} & \order{n_I^2}&  \order{1}\\[0.5em]
	$m$ & \order{{m+q+n \choose n}} & \order{{m+q+n \choose n}} & \order{1} & \order{{\floor{(m+q)/2}+n \choose n}}\\[0.5em]
	\hline
\end{tabularx}
\end{table*}

\section{Practical Considerations}\label{sec:practicalities}
\subsection{Leveraging Causality for Decomposition} \label{sec:forgetting}
Techniques for the efficient numerical integration of ODEs hinge fundamentally on the causality that is inherent to the solution of ODEs. Specifically, causality enables the original problem, namely integration over a long time horizon, to be decomposed into a sequence of simpler, more tractable subproblems, each corresponding to integration over only a small fraction of the horizon. In this section, we discuss how the structure of the presented optimization problems can be exploited in a similar spirit. Additionally, we show that such exploitation of structure gives rise to a mechanism for trading off tractability and scalability. 

Suppose we are interested in computing moment bounds at the end of a {\em long} time horizon $[0,t_f]$. In light of the arguments made in Section \ref{sec:hierarchies}, it is reasonable to expect that the set $\mathsf{T}$ should ideally be populated with a large number of time points in this setting. Accordingly, solving the resultant optimization problem in one go may become prohibitively costly, even despite the benign scaling of the SDP size with respect to $|\mathsf{T}|$. As alluded to in the beginning of this section, this limitation may be circumvented by decomposing the problem into a sequence of simpler subproblems each of which cover only a fraction of the time horizon. To that end, suppose that $\mathsf{T} = \St{t_1, \dots, t_{n_{\mathsf{T}}}}$ is ordered with $t_{n_{\mathsf{T}}}=t_f$, and let $t_0 = 0$. Further consider the subsets $\mathsf{T}_1, \dots, \mathsf{T}_{n_{\mathsf{T}}}$ of $\mathsf{T}$ such that $\mathsf{T}_k = \St{t_k}$. We now define 
\begin{widetext}
\begin{align}
	\mathsf{S}_{k} = \begin{cases}
	\mathsf{S}_0(\mathsf{F},n_I), &\text{if } k = 0,\\[1em]
	\left\lbrace\bm{y}_{t_k}, \{\bm{z}^{l}_{g,t_k}\} \left\rvert \begin{array}{l} 
	                                                                   \exists (\bm{y}_{t_{k-1}}, \{\bm{z}^{l}_{g,t_{k-1}}\}) \in \mathsf{S}_{k-1}  \text{ such that }  \\
	                                                                   (\{\bm{y}_{t_{k-1}}, \bm{y}_{t_k}\}, \{\bm{z}^{l}_{g,t_{k-1}}, \bm{z}^{l}_{g,t_k}\}) \in \mathsf{S}(\mathsf{F},\mathsf{T}_k, n_I) 
                                                                 \end{array}
                                                     \right\rbrace\right., &\text{if } k \geq 1.
	\end{cases} \nonumber 
\end{align}
\end{widetext}
At this point, it is worth emphasizing the meaning of each $\mathsf{S}_k$ and how its construction directly exploits the way we formulated the necessary moment conditions in $\mathsf{S}(\mathsf{F},\mathsf{T},n_I)$. To that end, note that each condition in $\mathsf{S}(\mathsf{F},\mathsf{T},n_I)$ only links variables corresponding to adjacent time points. As a consequence, the set $\mathsf{S}(\mathsf{F},\mathsf{T}_k,n_I)$ constrains only the variables $(\{\bm{y}_{t_{k-1}}, \bm{y}_{t_k}\}, \{\bm{z}^{l}_{g,t_{k-1}}, \bm{z}^{l}_{g,t_k}\})$. By construction of $\mathsf{S}_k$, we project out the variables $(\bm{y}_{t_{k-1}}, \{\bm{z}^{l}_{g,t_{k-1}}\})$ while imposing their membership in $\mathsf{S}_{k-1}$. It follows by induction that $\mathsf{S}_k$ precisely describes the projection of $\mathsf{S}(\mathsf{F}, \cup_{i=1}^k \mathsf{T}_{i}, n_I)$ onto the variables $(\bm{y}_{t_{k}}, \{\bm{z}^{l}_{g,t_{k}}\})$ under the condition that $(\bm{y}_0, \{\bm{z}_{g,0}\}) \in \mathsf{S}_0(\mathsf{F},n_I)$. By this argument, it follows that the original problem \eqref{eq:SDP} is equivalent to the following reduced space formulation:
\begin{align}
\inf_{\bm{y}_{t_f}, \{\bm{z}^{l}_{g,t_f}\}} \qquad &y_{\bm{i},t_f} \nonumber \\
\text{s.t.} \qquad  &(\bm{y}_{t_f}, \{\bm{z}^{l}_{g,t_f}\}) \in \mathsf{S}_{n_{\mathsf{T}}}. \nonumber 
\end{align}
where all decision variables that correspond to time points before $t_f$ have been projected out. It should be clear that the above optimization problem provides a computational advantage over the original problem only if the set $\mathsf{S}_{n_{\mathsf{T}}}$ can be represented, or at least tightly approximated, in a ``simple'' way. To that end, we suggest to successively compute conic outer approximations of the projections $\mathsf{S}_k$ according to Algorithm \ref{alg:successive_oa}.
\begin{algorithm}[H]
	\caption{Successive Overapproximation}\label{alg:successive_oa}
	\begin{algorithmic}[1]
		\Procedure{sO}{$\mathsf{S}_0,\mathsf{S}_1,\dots,\mathsf{S}_{n_{\mathsf{T}}}$}
		\State Set $\tilde{\mathsf{S}}_0 = \mathsf{S}_0$
		\For {$k = 1,\dots,n_{\mathsf{T}}$}
			\State Compute conic overapproximation
			\begin{align}
				\tilde{\mathsf{S}}_{k} \supset \left\lbrace\bm{y}_{t_k}, \{\bm{z}^l_{g,t_k}\} \left\lvert \begin{array}{l} \exists (\bm{y}_{t_{k-1}}, \{\bm{z}^{l}_{g,t_{k-1}}\}) \in \tilde{\mathsf{S}}_{k-1} \text{ such that} \\
					\begin{pmatrix} \St{\bm{y}_{t_{k-1}}, \bm{y}_{t_{k}}, }\\ \St{\bm{z}^{l}_{g,t_{k-1}},\bm{z}^{l}_{g,t_{k}}}\end{pmatrix} \in \mathsf{S}_k
				\end{array}\right\rbrace\right. \nonumber 
			\end{align}
		\EndFor
		\State \Return $\tilde{\mathsf{S}}_{n_{\mathsf{T}}}$
		\EndProcedure
	\end{algorithmic}
\end{algorithm}

Note that Algorithm \ref{alg:successive_oa} parallels the decomposition approach taken in classical numerical integration of ODEs: the task of finding moment bounds over the entire time horizon $[0,t_f]$ is decomposed into a sequence of smaller subproblems corresponding to finding moment bounds over smaller subintervals of the horizon and each subproblem is solved by using the solution of the previous subproblem as input data. In other words, Algorithm \ref{alg:successive_oa} propagates the moment bounds forward in time, successively subinterval by subinterval, in the same way as a numerical integrator propagates values of the state of a dynamical system forward in time.

We conclude this section with some final remarks. First, we would like to emphasize that the specific choices of $\mathsf{T}$ and $\mathsf{T}_k$ made in this section are made purely for clarity of exposition. In general, $t_f$ need not be the last element in $\mathsf{T}$ and the partition can be chosen as coarse as desired, i.e., each $\mathsf{T}_k$ can comprise multiple time points. In that case, however, Algorithm \ref{alg:successive_oa} needs to be adjusted accordingly. Second, computing and representing the conic overapproximations in Algorithm \ref{alg:successive_oa} may be expensive, in particular if many moments are considered. For example, computing a polyhedral outer approximation of the positive semidefinite cone is known to converge exponentially slowly in the worst-case \cite{braun2015approximation}. Second-order cone approximations perform better empirically \cite{ahmadi2019dsos,ahmadi2017optimization} and theoretically \cite{bertsimas2020polyhedral}, however, are more expensive to compute and represent. On the other hand, it may not be necessary to find overapproximations that are globally tight but only near the optimal solution of the original problem. Finally, with decisions on accuracy of the overapproximation and the coarseness of the partition of $\mathsf{T}$ required in Algorithm \ref{alg:successive_oa}, one is left with mechanisms to trade-off accuracy and computational cost.

\subsection{Quantifying Approximation Quality} \label{sec:poly_approx}
A natural question that arises from the formulation of problem \eqref{eq:SDP} is how to choose the parameters required for its construction, i.e., the sets $\mathsf{F}$ and $\mathsf{T}$, and the level $n_I$ of the proposed constraint hierarchy. We will show that an approximation of Problem \eqref{eq:OCP} can provide useful guidance for these choices. Specifically, using \eqref{eq:OCP} as a baseline, we show that an approximation of problem \eqref{eq:OCP} provides rigorous information on the best attainable bounds given the truncation order $m$ is fixed. To that end, recall that problem \eqref{eq:OCP} requires optimization over an infinite dimensional vector space, namely $\mathcal{C}^\infty(\R_+)$. To overcome this challenge, we will make two restrictions. On the one hand, we will restrict our considerations to a compact interval $[0,t_T]$ and, on the other hand, we will restrict the search space to the set of univariate polynomials up to a fixed but arbitrary maximum degree $d \in \mathbb{Z}_+$. Note that the latter restriction is in some sense arbitrarily weak as $\mathbb{R}[t]$ is dense in $\mathcal{C}^\infty([0,t_T])$ \cite{rudin1964principles}.

The above discussed restrictions enable the use of the following result to construct a tractable approximation of \eqref{eq:OCP}.
\begin{proposition}[Proposition 2 in \citet{ahmadi2018time}] \label{prop:sos_matrix}
	Let $m,d$ be positive integers. If $d$ is odd, let $r=k=\frac{d+1}{2}m$. Otherwise, let $k=(\frac{d}{2}+1)m$, and $r=\frac{dm}{2}$. Then, there exist two linear maps $\alpha:\Sym^k \maps \Sym^{m}[t]$ and $\beta:\Sym^r \maps \Sym^m[t]$ such that the matrix polynomial $\bm{X} \in \Sym^m[t]$ satisfies $\bm{X}(t) \succeq \bm{0}$ on $[0,1]$ if and only if there exist two matrices $\bm{Q}_\alpha \in \Sym^k_+$ and $\bm{Q}_\beta \in \Sym^r_+$ such that $\bm{X} = \alpha(\bm{Q}_\alpha) + \beta(\bm{Q}_\beta)$.
\end{proposition}
The maps $\alpha$ and $\beta$ in Proposition \ref{prop:sos_matrix} are remarkably simple and freely available software tools for sum-of-squares programming allow for simple, concise implementation. The interested reader is referred to \citet{ahmadi2018time} for an explicit description of $\alpha$ and $\beta$ alongside a simple proof of Proposition \ref{prop:sos_matrix}. 

Proposition \ref{prop:sos_matrix} allows to construct a tractable restriction of \eqref{eq:OCP} on a compact horizon. The following theorem which may be regarded as a special case of the results of \citet{ahmadi2018time} formalizes this claim. 
\begin{theorem}\label{thm:tvsdp}
	Let $d \in \mathbb{Z}_+$. Then, the following semi-infinite optimization problem
	\begin{align}
	\inf_{\bm{y} \in \R_d^{n_L+n_H}[t]} \quad &y_{\bm{j}}(t_f) \tag{pOCP} \label{eq:pOCP}\\
	\text{s.t.} \quad &\frac{d\bm{y}_L}{dt}(t) = \bm{A}_L\bm{y}_L(t) + \bm{A}_H\bm{y}_H(t), \quad \forall t \in [0, t_T], \nonumber \\
	&\bm{y}(0) = \bm{y}_0, \nonumber \\
	&\bm{y}(t) \in \G, \quad \forall t \in [0, t_T]. \nonumber 
	\end{align}
	is equivalent to a finite SDP. 
\end{theorem}
\begin{proof}
	First, note that all equality constraints in the above optimization problem require equality of polynomials of fixed maximum degree. Accordingly, equality can be enforced by matching the coefficients of the polynomials when expressed in a common basis which in turn can be done via finitely many affine equality constraints. Additionally, recall that $\G$ is described in terms of finitely many LMIs. Thus, the constraint $\bm{y}(t) \in \G$, $\forall t \in [0, t_T]$ is as well by Proposition \ref{prop:sos_matrix}.
\end{proof}
Unfortunately, \eqref{eq:pOCP} may be a strong restriction and often even infeasible. However, the formulation of \eqref{eq:pOCP} can be further relaxed without giving up too much relevant information. Specifically, we propose to restrict the solution space to piecewise polynomial functions in analogy to the collocation approach to optimal control \cite{cuthrell1987optimization}. The following corollary to Theorem \ref{thm:tvsdp} formalizes this approach.
\begin{corollary}
	Let $d,n_{\mathsf{T}} \in \mathbb{Z}_+$ and consider $n_{\mathsf{T}}+1$ time points $t_0,\dots,t_{n_{\mathsf{T}}}$ such that $0=t_0<t_1<\dots<t_{n_{\mathsf{T}}} \leq t_T$. Further suppose $t_f \in [t_k,t_{k-1}]$ for some $k$. Then, the following semi-infinite optimization problem
	\begin{align}
	\inf_{\bm{y}^i \in \R_d^{n_L+n_H}[t]} \quad &y^k_{\bm{j}}(t_f) \tag{pwpOCP} \label{eq:pwpOCP}\\
	\text{s.t.} \quad & \frac{d\bm{y}^i_L}{dt}(t) = \bm{A}_L\bm{y}^i_L(t) + \bm{A}_H\bm{y}^i_H(t),  \nonumber \\
	&\qquad \qquad \forall t \in [t_{i-1}, t_{i}], \ \forall i \in \St{1,\dots,n_{\mathsf{T}}}, \nonumber \\
	&\bm{y}^i(t_{i}) = \bm{y}^{i+1}(t_{i}), \quad \forall i \in \St{1,\dots,n_{\mathsf{T}}-1}, \nonumber \\
	&\bm{y}^1(0) = \bm{y}_0, \nonumber \\
	&\bm{y}^i(t) \in \G, \quad \forall t \in [t_{i-1}, t_{i}],\ \forall i \in  \St{1,\dots,n_{\mathsf{T}}}  \nonumber 
	\end{align}
	is equivalent to a finite SDP. Further, \eqref{eq:pwpOCP} is a valid restriction of \eqref{eq:SDP}.
\end{corollary}
\begin{proof}
	That \eqref{eq:pwpOCP} is equivalent to a finite SDP follows immediately from Theorem \ref{thm:tvsdp}. Further, let $\St{\bm{y}^i}$ be feasible for \eqref{eq:pwpOCP} and consider the piecewise polynomial obtained by parsing the $\bm{y}^i$ together like
	\begin{align*}
		\tilde{\bm{y}}(t) = \bm{y}^i(t), \ \forall t \in (t_{i-1},t_i] \text{ and } \forall i \in \St{1,\dots,n_{\mathsf{T}}}.
 	\end{align*}
 	By construction $\tilde{\bm{y}}$ satisfies \eqref{eq:mCME} and $\tilde{\bm{y}}(t) \in \G$, $\forall t \in [0,t_f]$. Accordingly, the iterated generalized moments obtained from $\tilde{\bm{y}}$ satisfy Conditions (i) and (ii) in Proposition \ref{prop:integral_hierarchy}. Thus, it is straightforward to generate a feasible point for \eqref{eq:SDP} from $\tilde{\bm{y}}$.
\end{proof}
Since \eqref{eq:pwpOCP} is fully independent of the choice of $\mathsf{F}$, $\mathsf{T}$ and $n_I$, it provides a way to check rigorously the approximation quality of \eqref{eq:SDP} against the baseline of \eqref{eq:OCP}. This can guide the user choice of the truncation order $m$ and the parameters $\mathsf{F}$, $\mathsf{T}$ and $n_I$. Specifically, the difference of optimal values of \eqref{eq:pwpOCP} and \eqref{eq:SDP} quantifies the potential for improvements by adding elements to $\mathsf{F}$ and $\mathsf{T}$ versus moving to a higher level in the proposed hierarchy. 

% Examples
\section{Examples}\label{sec:examples}
In this section, we present several case studies that demonstrate the effectiveness of the proposed bounding hierarchy. We put special emphasis on showcasing that the proposed method enables the computation of substantially tighter bounds than can be obtained by the method of \citet{dowdy2018dynamic}. Throughout, we use the subscripts $DB$ to indicate any results obtained with Dowdy and Barton's method \cite{dowdy2018dynamic} and the subscript $HB$ for those generated with the method presented in this paper.

\subsection{Preliminaries} \label{sec:test_functions}
\subsubsection{Reaction Kinetics} 
The reaction networks in all considered examples are assumed to follow mass action kinetics. 

\subsubsection{State Space \& LMIs}
Following \citet{dowdy2018dynamic}, we reduce the state space of every reaction network explicitly to the minimum number of independent species by eliminating reaction invariants. Further, we employ the set of LMIs suggested by \citet{dowdy2018dynamic}. These comprise LMIs of the form \eqref{eq:LMI} which reflect non-negativity of the underlying probability measure as well as non-negativity of molecular counts of all species including those eliminated via reaction invariants.

\subsubsection{Hierarchy Parameters}
Applying the proposed bounding scheme requires the user to specify a range of parameters, namely the truncation order $m$, the hierarchy level $n_I$, the test function set $\mathsf{F}$ and the set of time points $\mathsf{T}$ used to discretize the time domain. While all these hierarchy parameters can in principle be chosen arbitrarily (assuming the test functions satisfy the hypotheses of Propositions \ref{prop:integral_hierarchy} and \ref{prop:test_functions}) and independently, a careful choice is essential to achieve a good trade-off between bound quality and computational cost. We discuss this issue in greater detail in Section \ref{sec:bounding_mechanisms}. While we are at present not aware of a systematic way of choosing the hierarchy parameters optimally in the aforementioned sense, we found that the following set of simple heuristics performs well in practice:
\begin{itemize}
    \item The set of time points $\mathsf{T}$ is chosen by equidistantly discretizing the entire time horizon $[0,t_f]$, where $t_f$ denotes the time point at which the bounds are to be evaluated, into $n_{\mathsf{T}}$ intervals. 
    \item In line with Dowdy and Barton's original work \cite{dowdy2018dynamic}, we employ exponential test functions of the form $g(t) = e^{\rho(t_T-t)}$. As argued in Example \ref{ex:test_functions}, any set of test functions of this form satisfies the hypotheses of Propositions \ref{prop:integral_hierarchy} and \ref{prop:test_functions}. Throughout, we choose $t_T$ to coincide with the end of the time horizon on which the bounds are to be evaluated. As $t_T$ merely controls the scale of the generalized moments generated by $g$, this choice is somewhat arbitrary but contributes in our experience to improved numerical conditioning of \eqref{eq:SDP}. For the choice of the parameters $\rho$, we draw motivation from linear time-invariant systems theory and choose $\rho$ based on the singular values of the coefficient matrix $\bm{A}$ of the moment dynamics \eqref{eq:mCME}. Concretely, we choose the test function set $\mathsf{F} = \St{e^{-\sigma_i(t_T-t)}}_{i=1}^{n_{\mathsf{F}}}$ assembled from the smallest $n_{\mathsf{F}}$ unique singular values $\sigma_1,\dots,\sigma_{n_{\mathsf{F}}}$ of $\bm{A}$. 
    \item Motivated by the scaling of the size of the bounding SDP (see Table \ref{tab:scaling}), we use the following greedy procedure to ultimately choose $m$, $n_{I}$, $n_{\mathsf{T}}$ and $n_{\mathsf{F}}$
    \begin{enumerate}
        \item Fix $m = 2$, $n_{I} = 2$, $n_{\mathsf{F}} = 1$ and successively increase $n_{\mathsf{T}}$ until no significant bound tightening effect is observed. 
        \item Increase $n_{\mathsf{F}}$ successively until no significant bound tightening effect is observed.
        \item Increase $m$ until bounds are sufficiently tight or computational cost exceeds a tolerable amount.
    \end{enumerate}
    Note that the above procedure fixes the hierarchy Level $n_I$ at 2. While increasing $n_I$ generally also has a bound tightening effect, in our experience it promotes numerical ill-conditioning and is rarely significantly more efficient than the other bound tightening mechanisms.
\end{itemize}
In all following case studies, we employ the above heuristics to choose the hierarchy parameters. As the time point and test function sets are systematically generated once the parameters $n_{\mathsf{T}}$ and $n_{\mathsf{F}}$ are chosen, we instead report these parameters in place of $\mathsf{T}$ and $\mathsf{F}$.

\subsubsection{Numerical Considerations}
Sum-of-squares and moment problems are notorious for poor numerical conditioning and the problems \eqref{eq:SDP} and \eqref{eq:pwpOCP} are no exception to this issue. While the specific reasons for this problem remain largely unclear, it is widely suspected to originate from the fact that moments and coefficients of polynomials expressed in the monomial basis often vary over several orders of magnitude. To circumvent this deficiency, we employ a simple scaling strategy for the decision variables in the bounding problems. This strategy is applicable whenever bounds are computed at multiple time points $t_1 < t_2 < \cdots$ along a trajectory and can be summarized as follows: we solve the SDPs in chronological order and scale the decision variables in the bounding problem corresponding to the time point $t_k$ by the values attained at the solution of the bounding problem associated with the previous time point $t_{k-1}$. For the initial problem, we perform scaling based on the moments of the distribution of the initial state of the system.
While an appropriate scaling of the decision variables is crucial to avoid numerical issues, it is not always sufficient to achieve convergence of the solver to a desired degree of accuracy with respect to optimality. We hedge against potentially inaccurate, suboptimal solutions and ensure validity of the computed bounds by verifying that the solver converged to a dual feasible point and reporting the associated dual objective value.

\subsubsection{Implementation}
All semidefinite programs solved for the case studies presented in this section were assembled using JuMP \cite{dunning2017jump} and solved with MOSEK v9.0.97 \cite{andersen2000mosek}. Our implementation is openly available at \url{https://github.com/FHoltorf/StochMP}. 

\subsection{Generic Examples}
To contrast the performance of the proposed methodology with its predecessor, we first consider three generic reaction networks that were studied by \citet{dowdy2018dynamic}. 

\subsubsection{Simple Reaction Network}
First, we study the bound quality for means and variances of the molecular counts of the species $A$ and $C$ following the simple reaction network
\begin{align}
	\ce{A} + \ce{B} \stackrel{c_1}{\rightarrow} \ce{C} \stackrel[c_2]{c_3}{\leftrightharpoons}\ce{D}.\label{sys:toy}
\end{align} 
Figure \ref{fig:toy_system} shows a comparison between the bounds obtained by both methods. For reference, also the trajectories obtained with Gillespie's Stochastic Simulation Algorithm are provided. The results showcase that the presented necessary moment conditions have the potential to tighten the obtained bounds significantly. In particular bounds on the variance of both species are dramatically improved at the relatively low truncation order of $m=4$. 
\begin{figure}
	\centering
	\begin{subfigure}{0.49\textwidth}
		\includegraphics[width=1.0\textwidth]{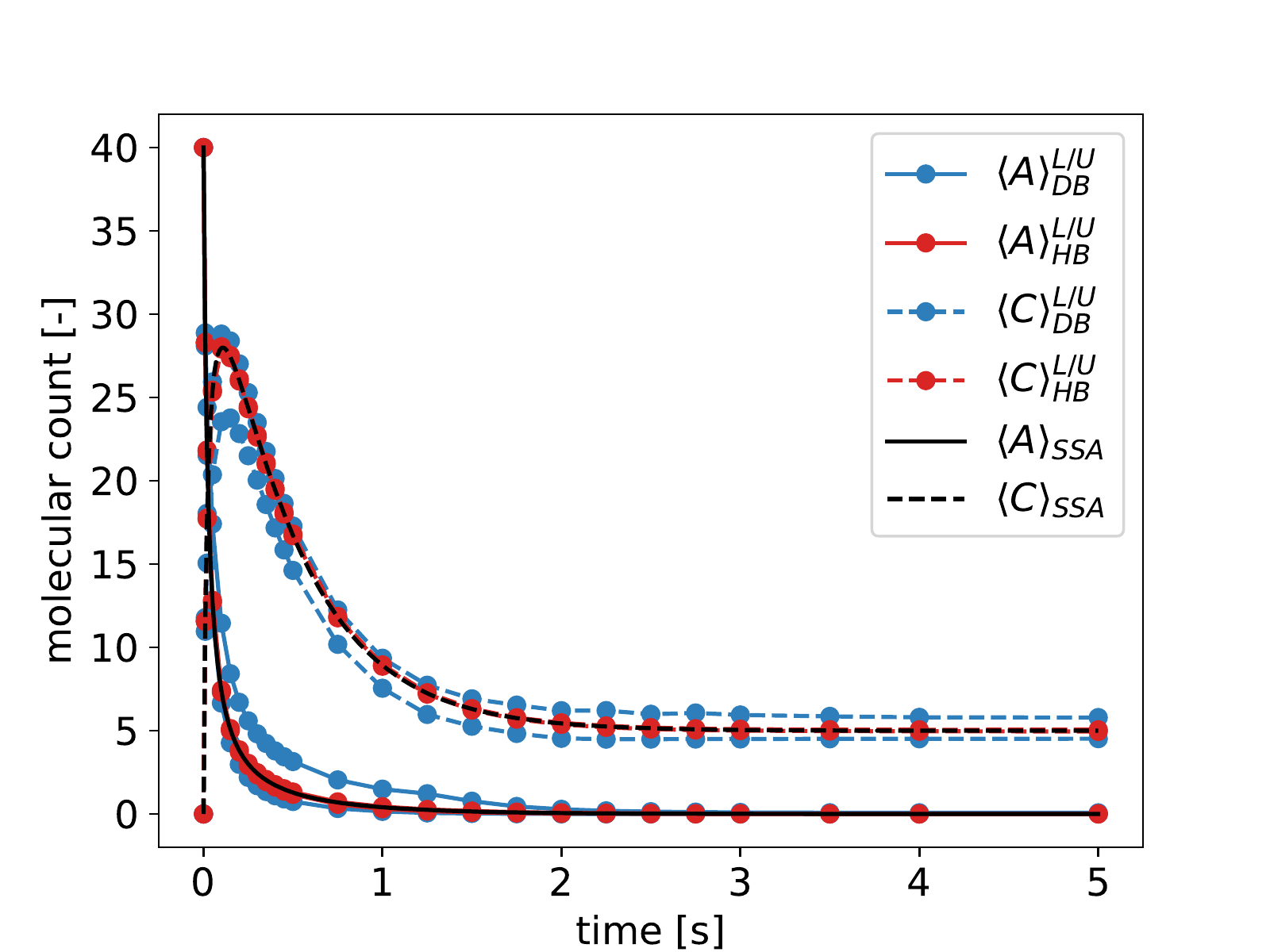}
		\caption{}
	\end{subfigure}
	\begin{subfigure}{0.49\textwidth}
		\includegraphics[width=1.0\textwidth]{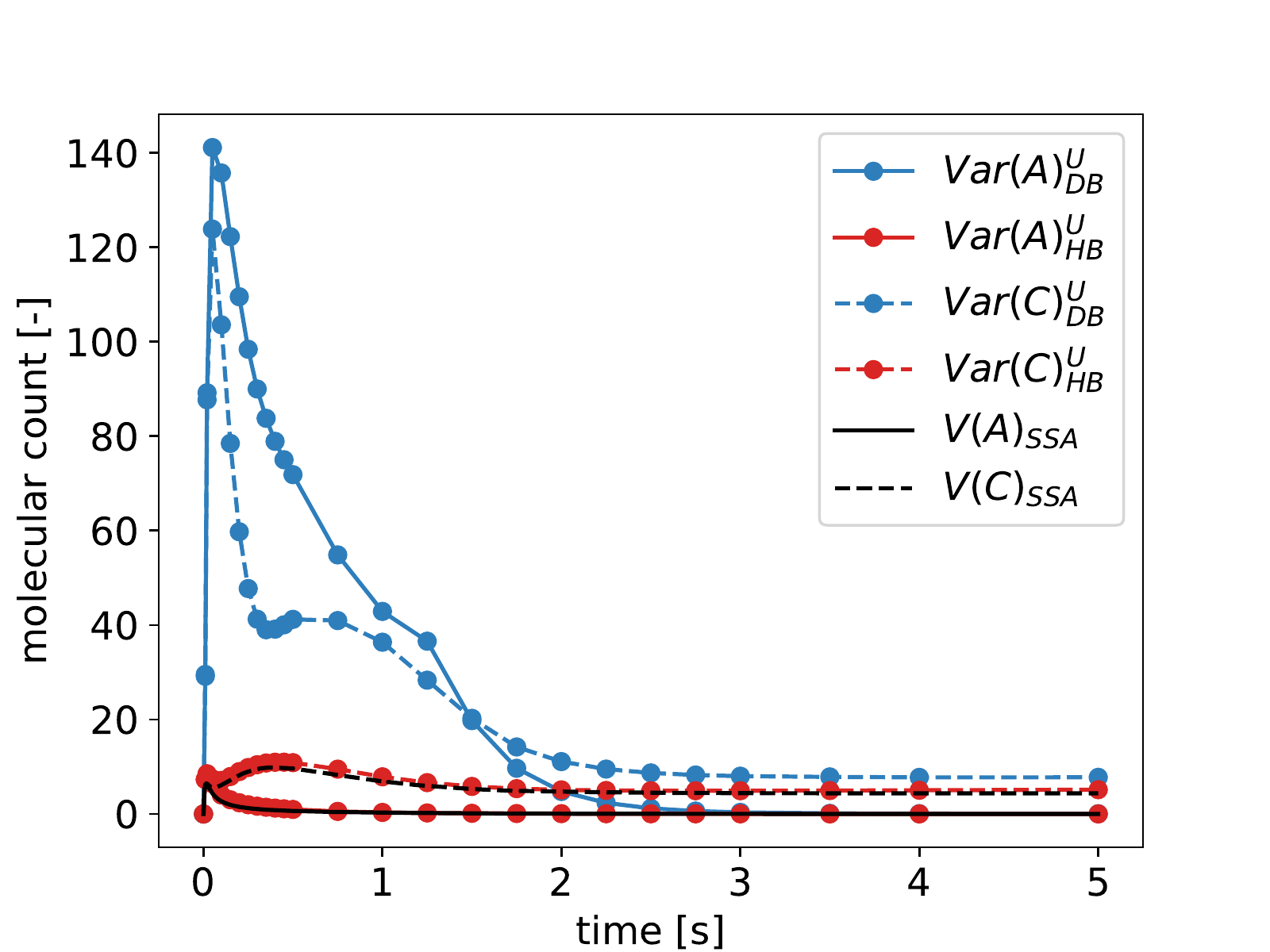}
		\caption{}
	\end{subfigure}
	\begin{subfigure}{0.49\textwidth}
		\includegraphics[width=1.0\textwidth]{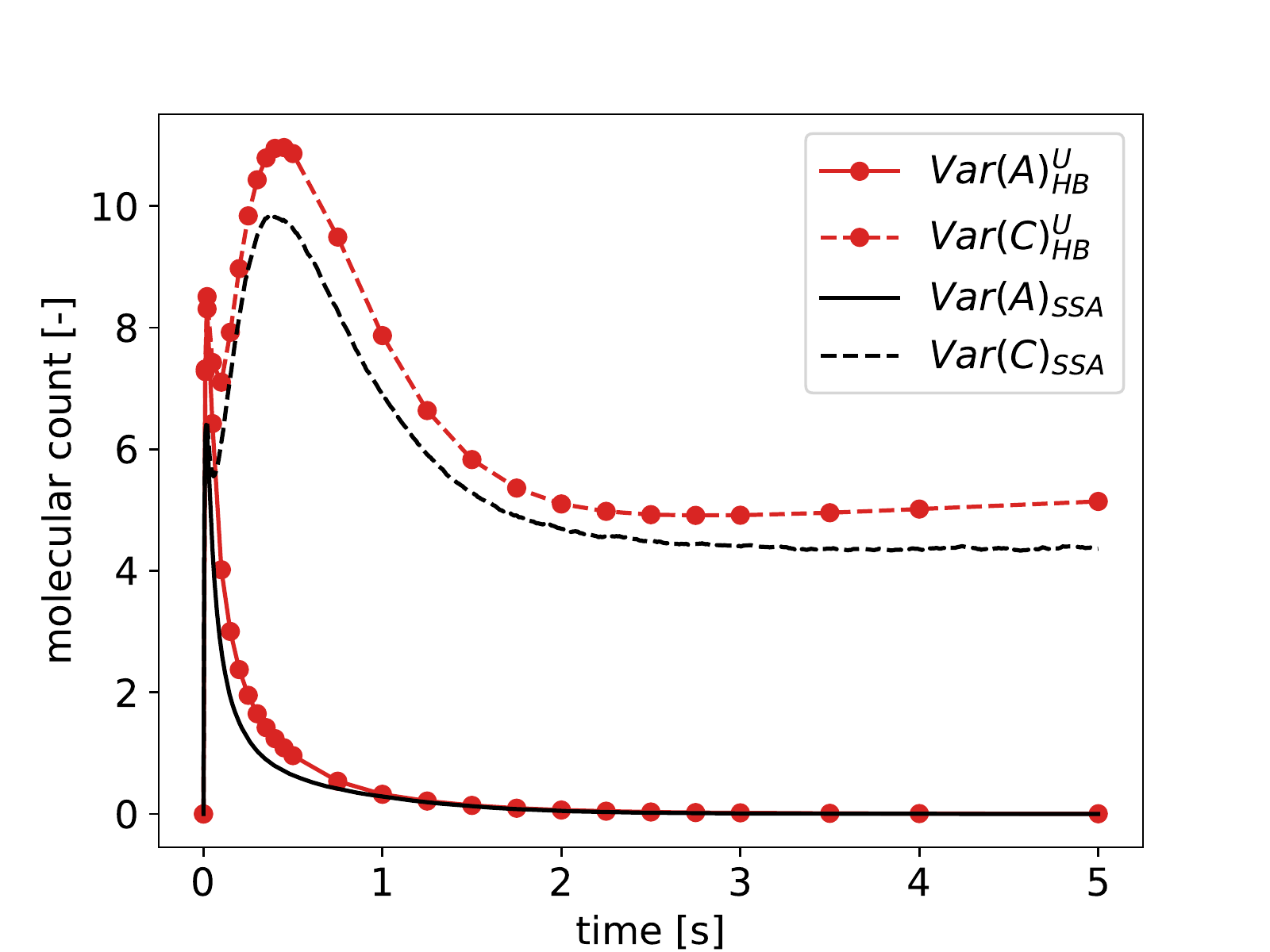}
		\caption{}
	\end{subfigure}
	\caption[Bounds on means and variances of species \ce{A} and \ce{B} in reaction network \eqref{sys:toy}]{Bounds on (a) means and (b,c) variances of molecular counts of species \ce{A} and \ce{C} in reaction network \eqref{sys:toy}; initial state: $x_{A,0} = 40$, $x_{B,0} = 41$, and $x_{C,0}=x_{D,0}=0$; kinetic parameters: $\bm{c} = (1, 2.1, 0.3) \, \si{\per \second}$; hierarchy parameters: $m = 4$, $n_{\mathsf{F}}=3$, $n_{\mathsf{T}} = 10$.} \label{fig:toy_system}
\end{figure}

\subsubsection{Cyclic Reaction Network}
Second, we investigate the cyclic reaction network illustrated in Figure \ref{sys:cyclic_network}. 
\begin{figure}%[h!]
	\centering
	\begin{tikzpicture}[scale=0.5]
	\draw[-Latex] (-30+25:1) arc (-30+25:90-50:1) node[midway, right]{$c_3$};
	\draw[-Latex] (90+50:1) arc (90+50:210-25:1) node[midway, left]{$c_1$};
	\draw[-Latex] (210+25:1) arc (210+25:330-25:1) node[midway, below]{$c_2$};
	
	\node (A+B) at (90:1) {$\ce{A}+\ce{B}$};
	\node (C) at (-30:1) {$\ce{D}$};
	\node (D) at (210:1) {$\ce{C}$};
	\end{tikzpicture}
	\caption{Cyclic reaction network from \citet{dowdy2018dynamic}}\label{sys:cyclic_network}
\end{figure}
This reaction network exhibits similar characteristics as the simple network studied in the previous section, i.e., a bounded state space and two independent species. In contrast to the simple network, however, the situation is a bit more complex since the reaction is fully reversible. Therefore, the reachable set of the cyclic reaction network is at least as large as that of the simple network given an identical initial state. 

Figure \ref{sys:cyclic_network} shows a comparison between the bounds obtained from \eqref{eq:SDP} and Dowdy and Barton's method. Although the bounds are slightly looser than before, the results demonstrate again the potential of the proposed methodology. 
\begin{figure}
	\centering
	\begin{subfigure}{0.49\textwidth}
		\includegraphics[width=0.99\textwidth]{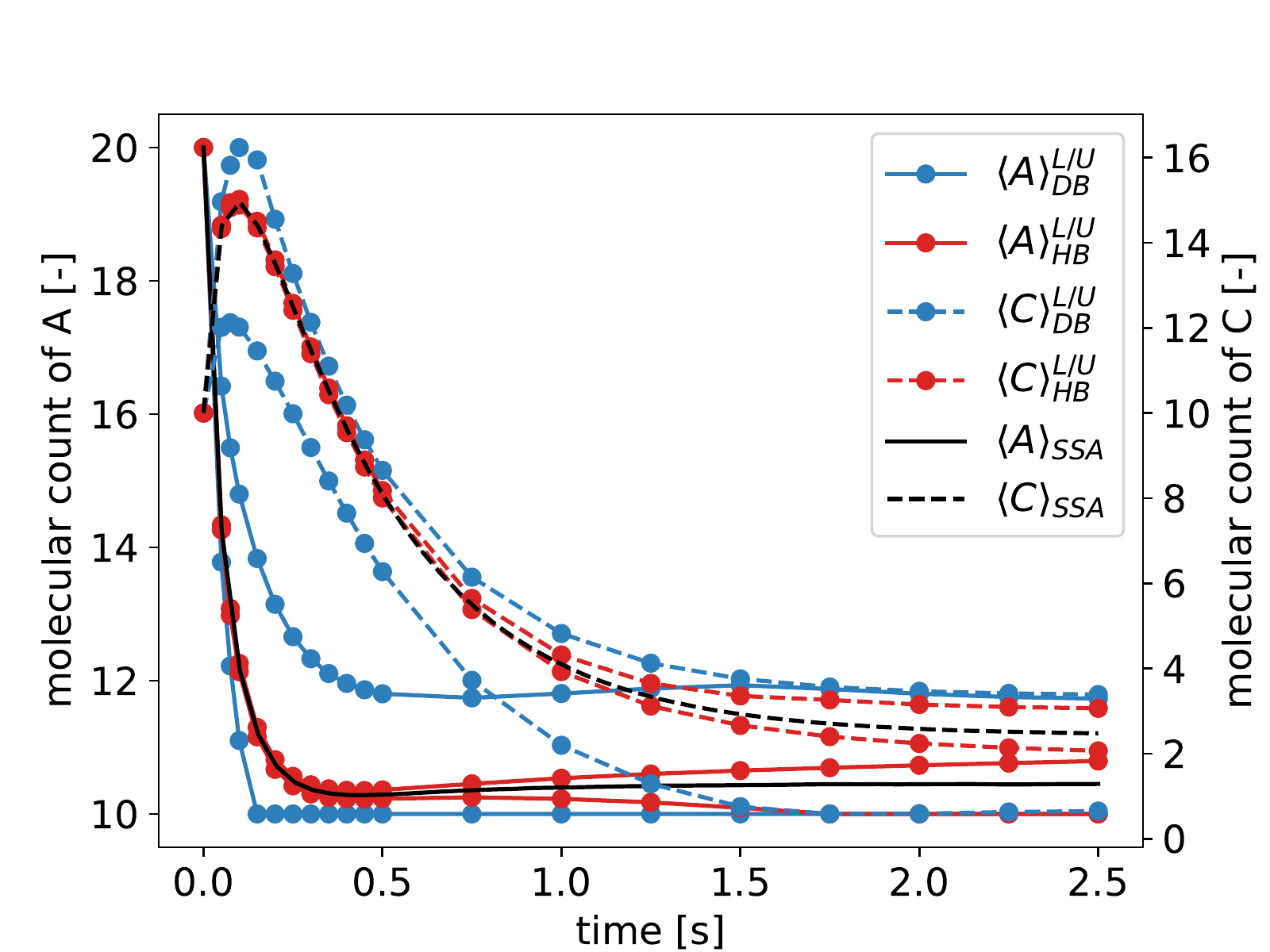}
		\caption{}
	\end{subfigure}
	\begin{subfigure}{0.49\textwidth}
		\includegraphics[width=0.99\textwidth]{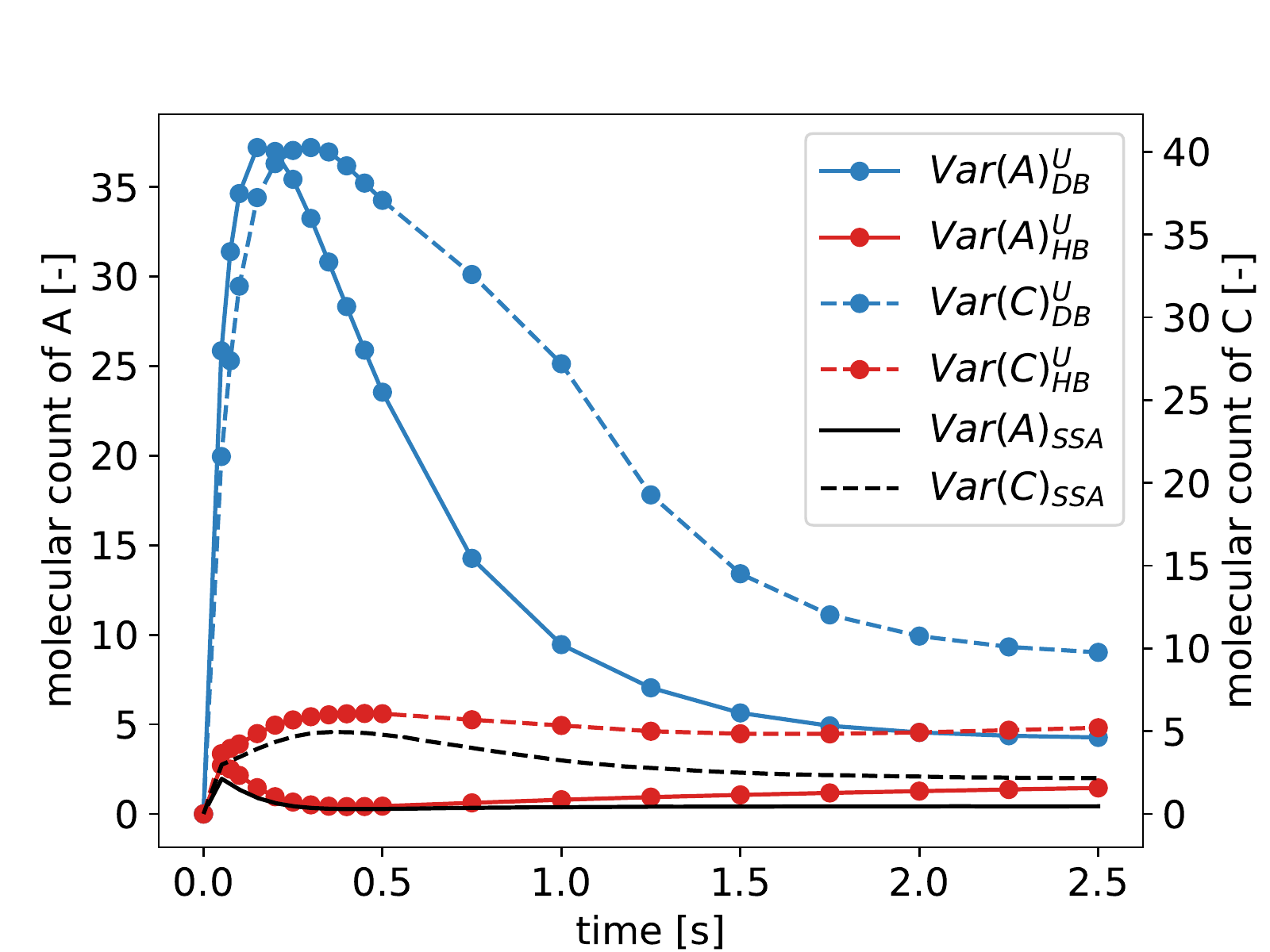}
		\caption{}
	\end{subfigure}
	\caption[Bounds on means and variances of the molecular counts of species \ce{A} and \ce{C} in the cyclic reaction network shown in Figure \ref{sys:cyclic_network}]{Bounds on means (a) and variances (b) of the molecular counts of species \ce{A} and \ce{C} in the cyclic reaction network shown in Figure \ref{sys:cyclic_network}; initial state: $x_{A,0} = 20$, $x_{B,0} = 10$, and $x_{C,0}=10$, $x_{D,0}=0$; kinetic parameters: $\bm{c} = (1, 2.1, 0.3) \, \si{\per \second}$; hierarchy parameters: $m = 2$, $n_{\mathsf{F}} = 3$, $n_{\mathsf{T}} = 10$.} \label{fig:cyclic_network}
\end{figure}

\subsubsection{Large Reaction Network}
Last, we investigate the reaction network illustrated in Figure \ref{sys:large_network}. In contrast to the previous networks, this large reaction network poses a challenge for sampling based analysis techniques. The underlying reason for that is two-fold. On the one hand, the network is characterized by a large, 7-dimensional state space. On the other hand, the system is extremely stiff. These properties frustrate sampling based techniques as they exacerbate the need for large sample sizes and render each sample evaluation expensive.

Figure \ref{fig:large_network} shows bounds on the mean molecular counts of species $A$ and $H$. In line with the results of the previous sections, the bounds obtained by the proposed method are again considerably tighter. In this example, however, this result carries more weight as increasing the truncation order leads to a prohibitive increase in problem size for the method of \citet{dowdy2018dynamic}. Accordingly, the proposed method offers bounds at a quality that was previously not attainable for problems of such complexity.  
\begin{figure}%%[h!]
	\centering
	\includegraphics[scale=0.9]{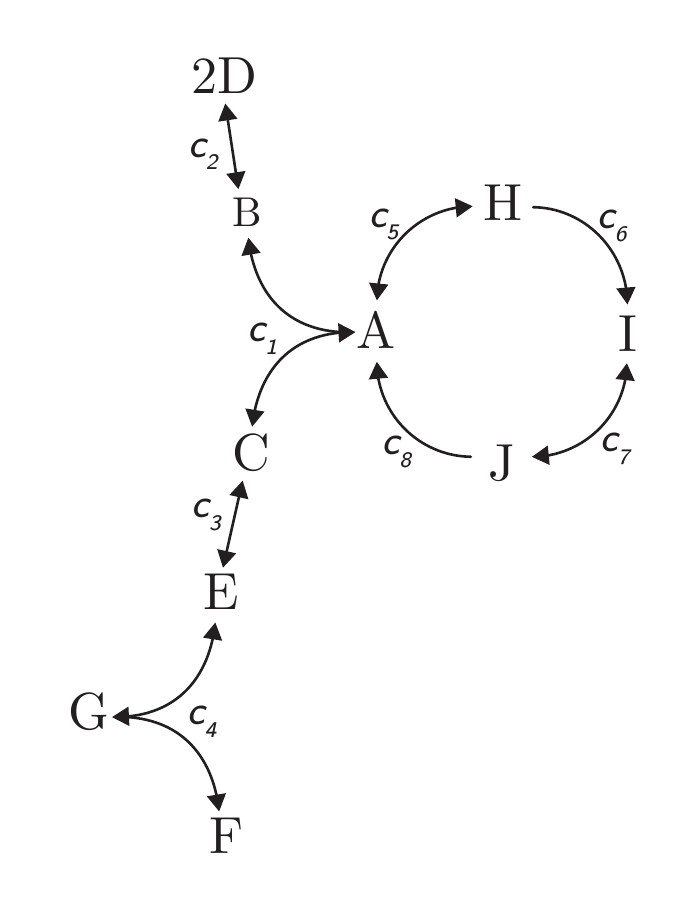}
	\caption{Large reaction network from \citet{dowdy2018dynamic}} \label{sys:large_network}
\end{figure}
\begin{figure}
	\centering
	\includegraphics[scale=0.5]{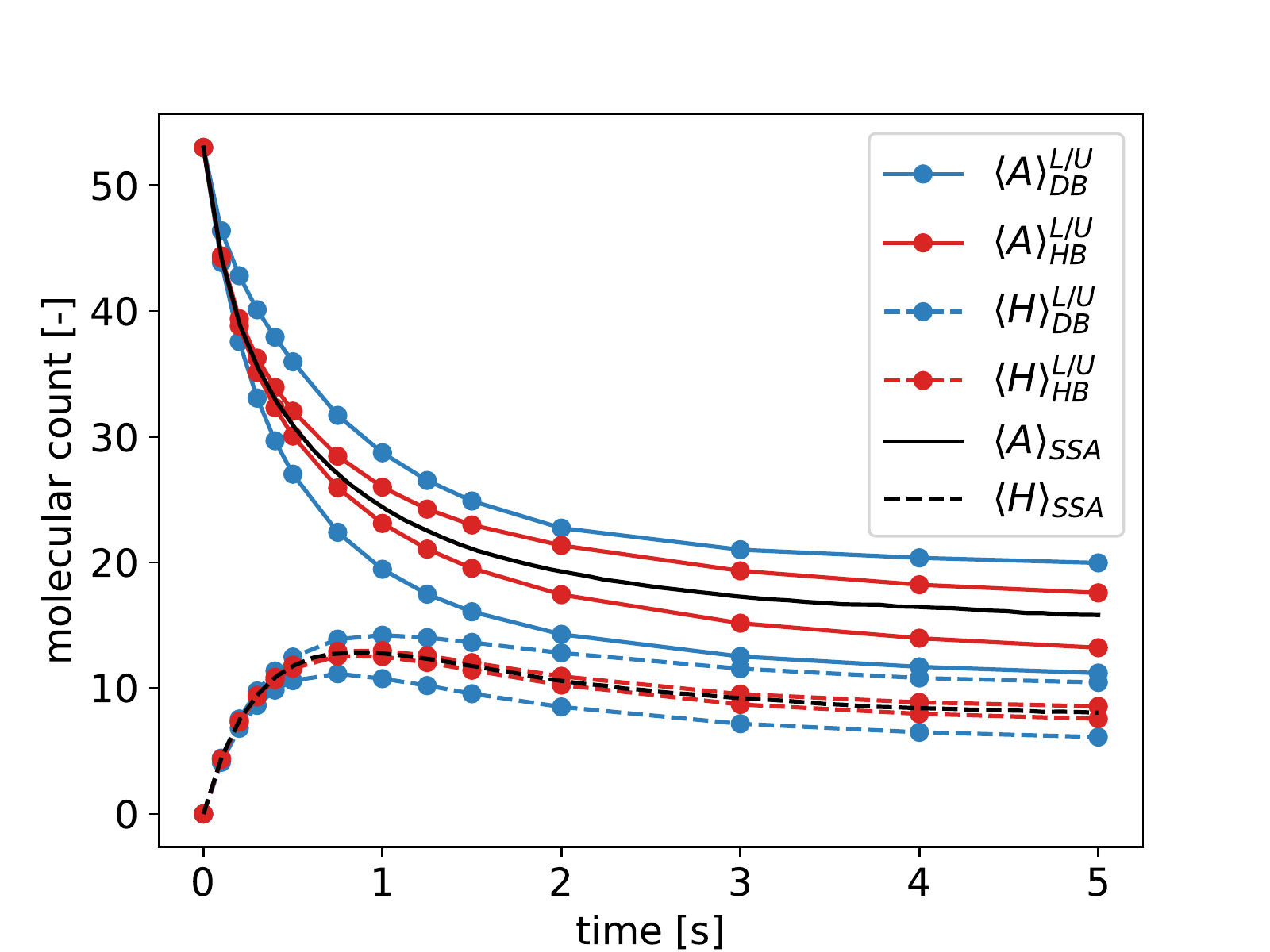}
	\caption[Bounds on the mean molecular counts of species \ce{A} and \ce{H} in the large reaction network shown in Figure \ref{sys:large_network}]{Bounds on the mean molecular counts of species \ce{A} and \ce{H} in the large reaction network shown in Figure \ref{sys:large_network}; initial state: $x_{\ce{A},0}=x_{\ce{F},0} = 53$, $x_{\ce{B},0} = x_{\ce{C},0} = x_{\ce{D},0} = x_{\ce{E},0} = x_{\ce{G},0} = x_{\ce{H},0} = x_{\ce{I},0} = x_{\ce{J},0} = 0$; kinetic parameters: $\bm{c} = (1, 1, 1, 10^4, 1, 1, 10^5, 1) \, \si{\per \second}$; hierarchy parameters: $m=2$, $n_{\mathsf{F}}=5$, $n_{\mathsf{T}} = 5$.} \label{fig:large_network}
\end{figure}

\subsection{Open Systems}
We now depart from systems that have a bounded state space and turn to open systems. Such systems are of particular interest as their moments can rarely be found analytically and the corresponding CME gives rise to a countably infinite system of coupled ODEs precluding a direct numerical integration. To showcase the ability to compute tight moment bounds for systems of this type, we study two birth-death processes with different ``degrees of nonlinearity''.

The first system we investigate is the following simple nonlinear birth-death process:
\begin{align}
	\emptyset \stackrel{c_1}{\rightarrow}\ce{A}, \qquad 2\ce{A} \stackrel{c_2}{\rightarrow} \emptyset \label{sys:birth_death} 
\end{align}
Figure \ref{fig:birth_death} draws a comparison between the bounds for the mean and variance of the molecular count of species \ce{A}. The proposed method again yields substantially tighter bounds than its predecessor. In fact, the bounds on the mean are even tight enough to essentially recover the true solution.
\begin{figure}
	\centering
	\begin{subfigure}{0.49\textwidth}
		\includegraphics[width=1.0\textwidth]{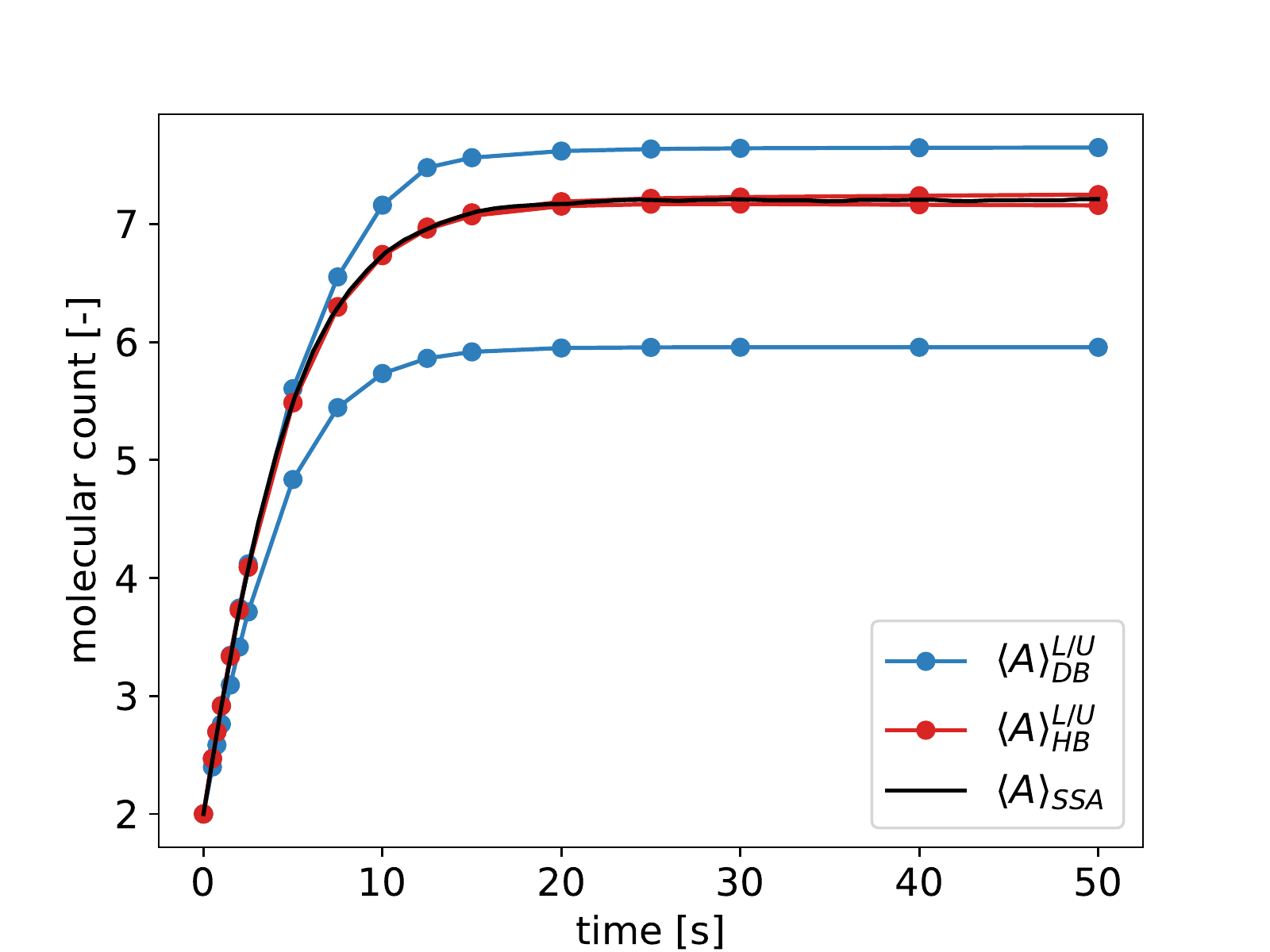}
		\caption{}
	\end{subfigure}
	\begin{subfigure}{0.49\textwidth}
		\includegraphics[width=1.0\textwidth]{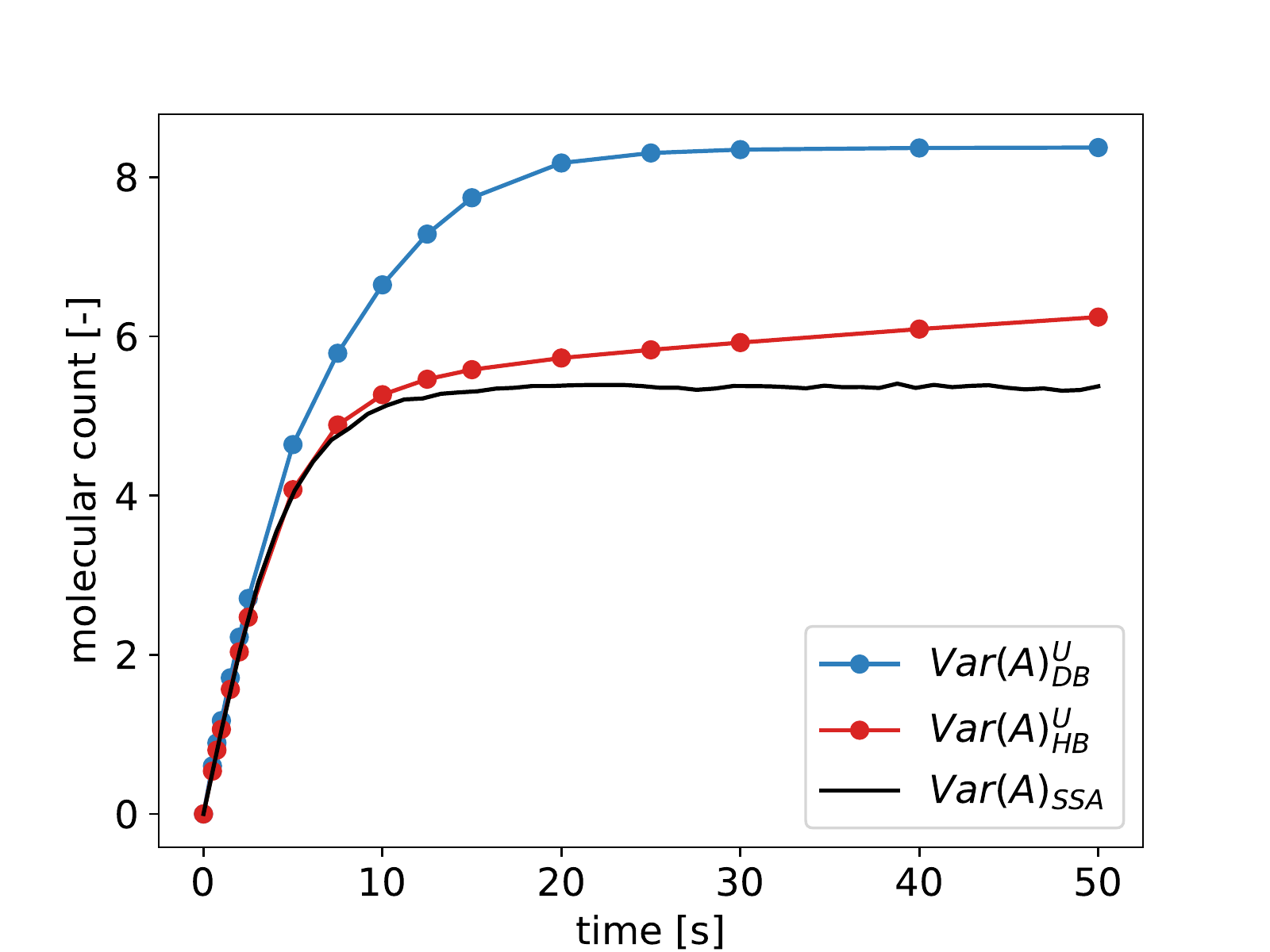}
		\caption{}
	\end{subfigure}
	\caption[Bounds on means and variances of the molecular count of species \ce{A} undergoing the birth-death process \eqref{sys:birth_death}]{Bounds on means (a) and variances (b) of the molecular count of species \ce{A} undergoing the birth-death process \eqref{sys:birth_death}; initial state: $x_{A,0} = 2$; kinetic parameters: $\bm{c} = (1, 0.01) \, \si{\per \second}$; hierarchy parameters: $m = 8$, $n_{\mathsf{F}}=3$, $n_{\mathsf{T}} = 20$.} \label{fig:birth_death}
\end{figure}

Next, we will consider Schl{\"o}gl's system \cite{schlogl1972chemical}: 
\begin{align}
 2\ce{X} \stackrel[c_1]{c_2}{\leftrightharpoons} 3\ce{X}, \qquad \ce{X} \stackrel[c_3]{c_4}{\leftrightharpoons} \emptyset \label{sys:schloegl}.
\end{align}
As a canonical example for a chemical bifurcation, Schl{\"o}gel's system was previously studied by different authors \cite{dowdy2018bounds,kuntz2019bounding} to illustrate bounding methods for the moments of stationary solutions of the CME. We provide the first analysis for the dynamic case here. 

Figure \ref{fig:schloegl} illustrates the results for Schl{\"o}gl's system. 
Although the proposed methodology again strongly outperforms its predecessor, the obtained bounds are rather loose. We emphasize, however, that we could not reproduce bounds of similar quality by increasing the truncation order in Dowdy and Barton's \cite{dowdy2018dynamic} method. The bounds started stalling before eventually poor numerical conditioning prohibited solution of the SDPs altogether.
\begin{figure}
	\centering
	\begin{subfigure}{0.49\textwidth}
		\includegraphics[width=\textwidth]{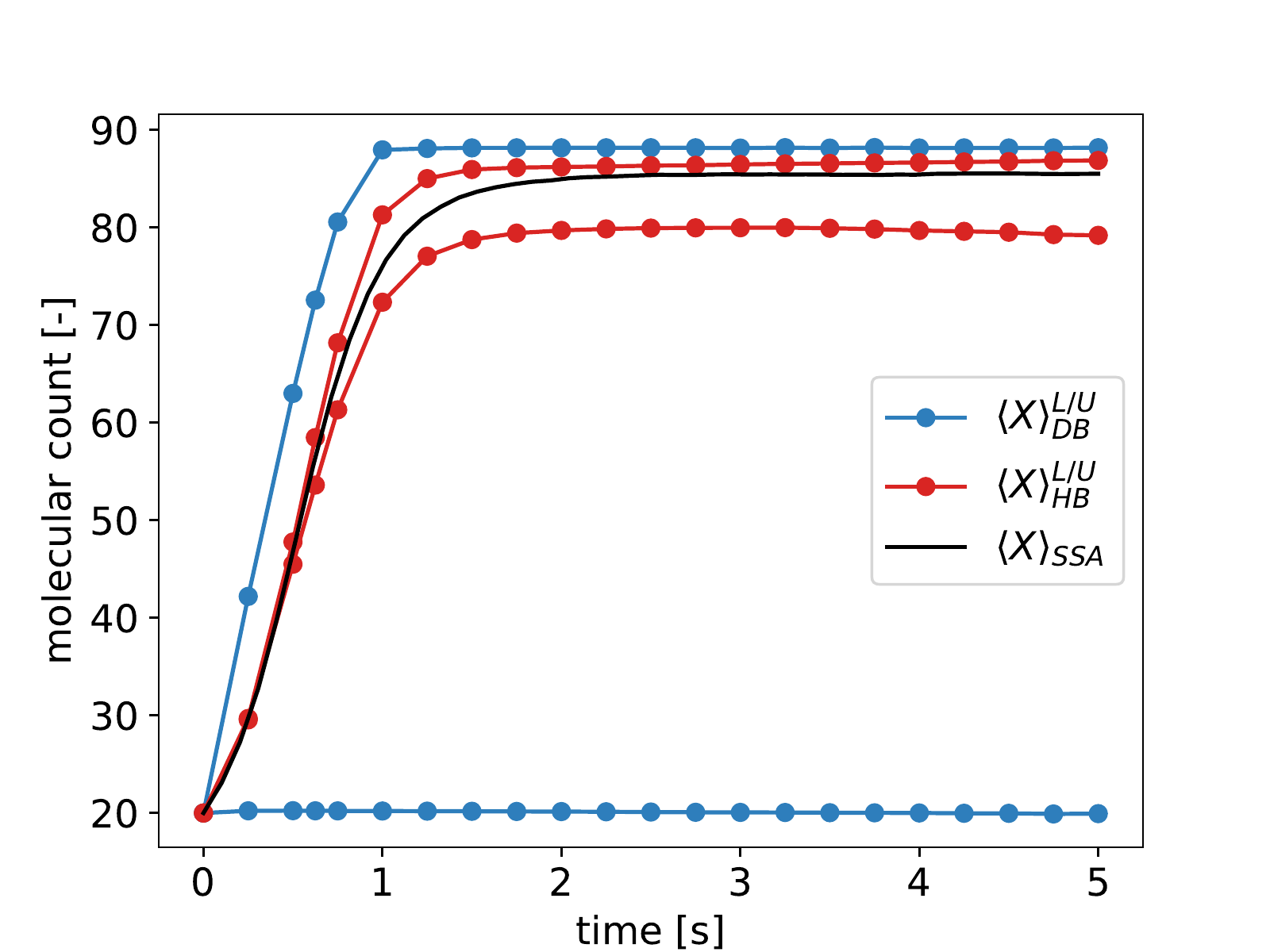}
		\caption{}
	\end{subfigure}
	\begin{subfigure}{0.49\textwidth}
		\includegraphics[width=\textwidth]{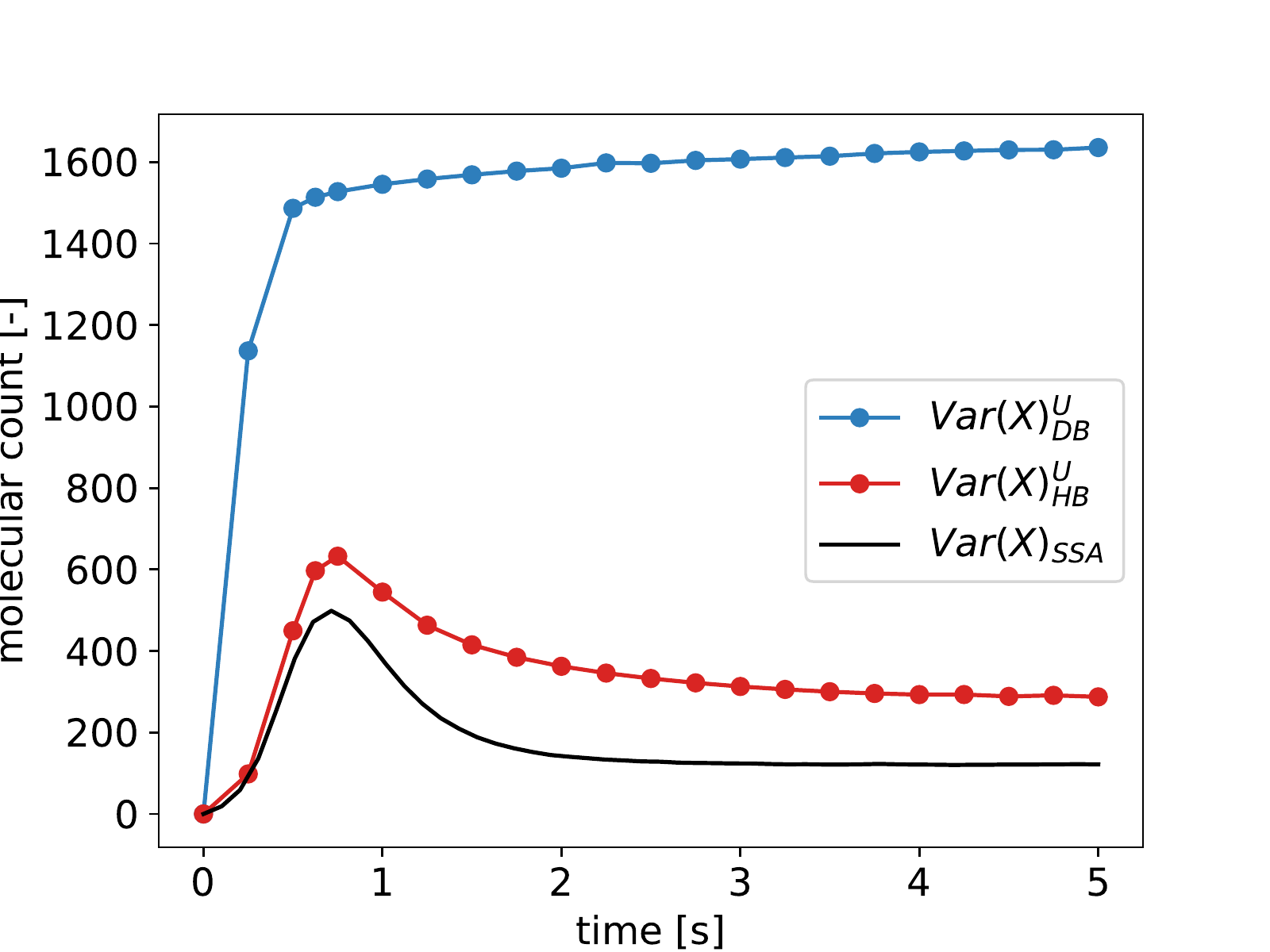}
		\caption{}
	\end{subfigure}
	\caption[Bounds on means and variances of the molecular counts of species \ce{X} in Schl{\"o}gl's system \eqref{sys:schloegl}]{Bounds on means (a) and variances (b) of molecular counts of species \ce{X} in Schl{\"o}gl's system \eqref{sys:schloegl}; initial state: $x_{X,0} = 20$; kinetic parameters: $\bm{c} = (0.15, \SI{1.5e-3}, 20, 2) \, \si{\per \second}$; hierarchy parameters: $m = 6$, $n_{\mathsf{F}} = 3$, $n_{\mathsf{T}} = 20$.} \label{fig:schloegl}
\end{figure}

\subsection{Biochemical Reaction Networks}
To finally demonstrate that the proposed methodology may be useful in practice, we examine three reaction networks drawn from biochemical application. In these applications, molecular counts are often present in the order of 10s to 100s necessitating the consideration of stochasticity.

\subsubsection{Michaelis-Menten Kinetics}
Michaelis-Menten kinetics underpin a vast range of metabolic processes. Understanding the behavior and noise present in the associated reaction networks is of particular value for the investigation of the metabolic degradation of trace substances in biological organisms. We examine the basic Michaelis-Menten reaction network:
\begin{align}
	&\ce{S} + \ce{E} \stackrel[c_1]{c_2}{\leftrightharpoons} \ce{S}:\ce{E}\stackrel{c_3}{\rightarrow} \ce{P} + \ce{E} \nonumber\\
	&\ce{S}:\ce{E} \stackrel{c_3}{\rightarrow} \ce{P} + \ce{E} \label{sys:michaelis_menten}\\
	&\ce{P} \stackrel{c_4}{\rightarrow} \ce{S} \nonumber
\end{align}
The reaction network features a two-dimensional state space. Accordingly, we bound the means and variances of the molecular counts of the product \ce{P} and substrate \ce{S}. The results are illustrated in Figure \ref{fig:michaelis_menten}. For the sake of completeness, Figure \ref{fig:michaelis_menten} also features a comparison with the bounds obtained by Dowdy and Barton's \cite{dowdy2018dynamic} method. The proposed method reproduces essentially the exact solution for the means while providing reasonably tight bounds for the variances. Further, it again outperforms its predecessor, especially for bounds on the variances.  
\begin{figure}
	\centering
	\begin{subfigure}{0.49\textwidth}
		\includegraphics[width=\textwidth]{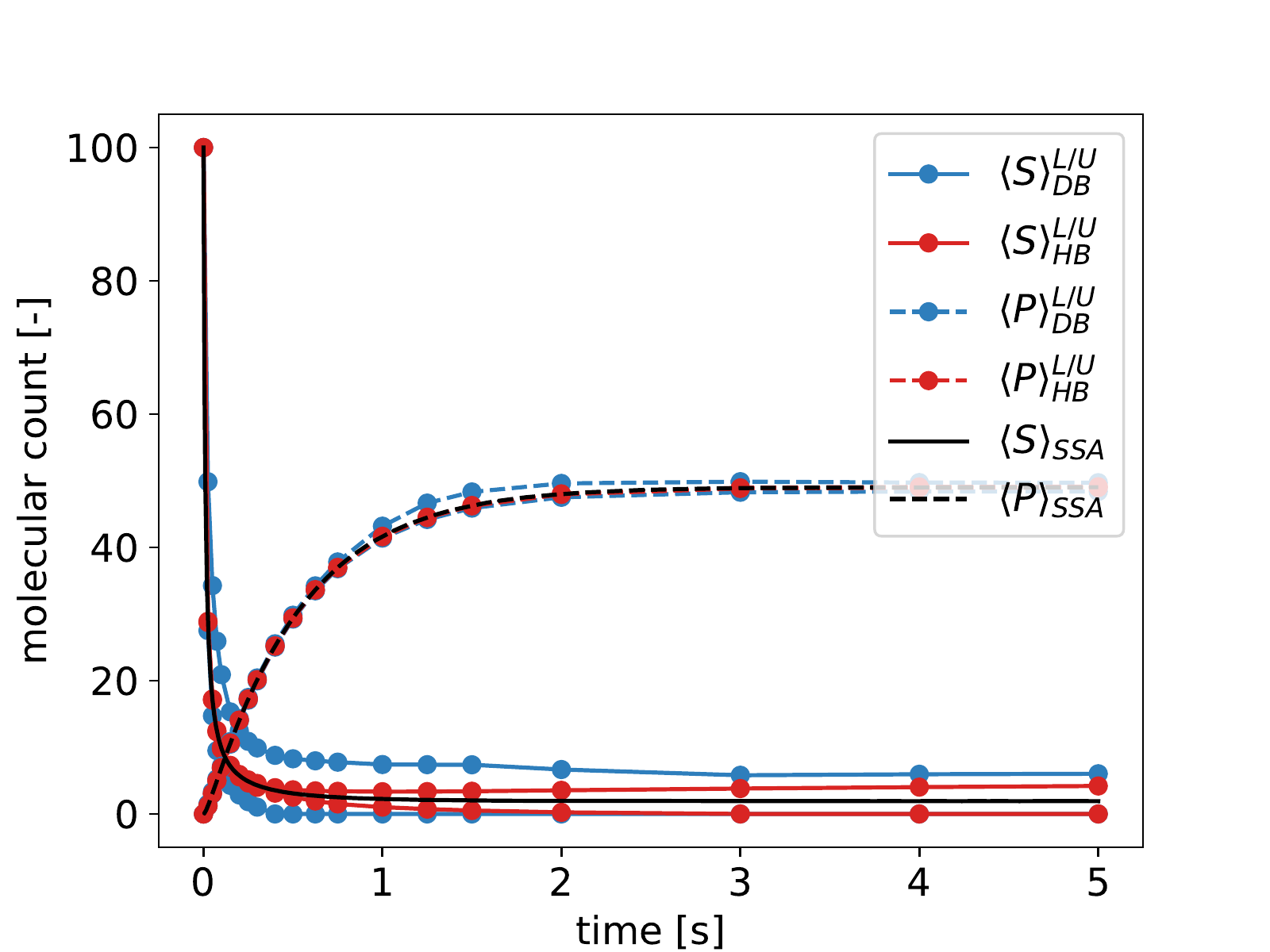}
		\caption{}
	\end{subfigure}
	\begin{subfigure}{0.49\textwidth}
		\includegraphics[width=\textwidth]{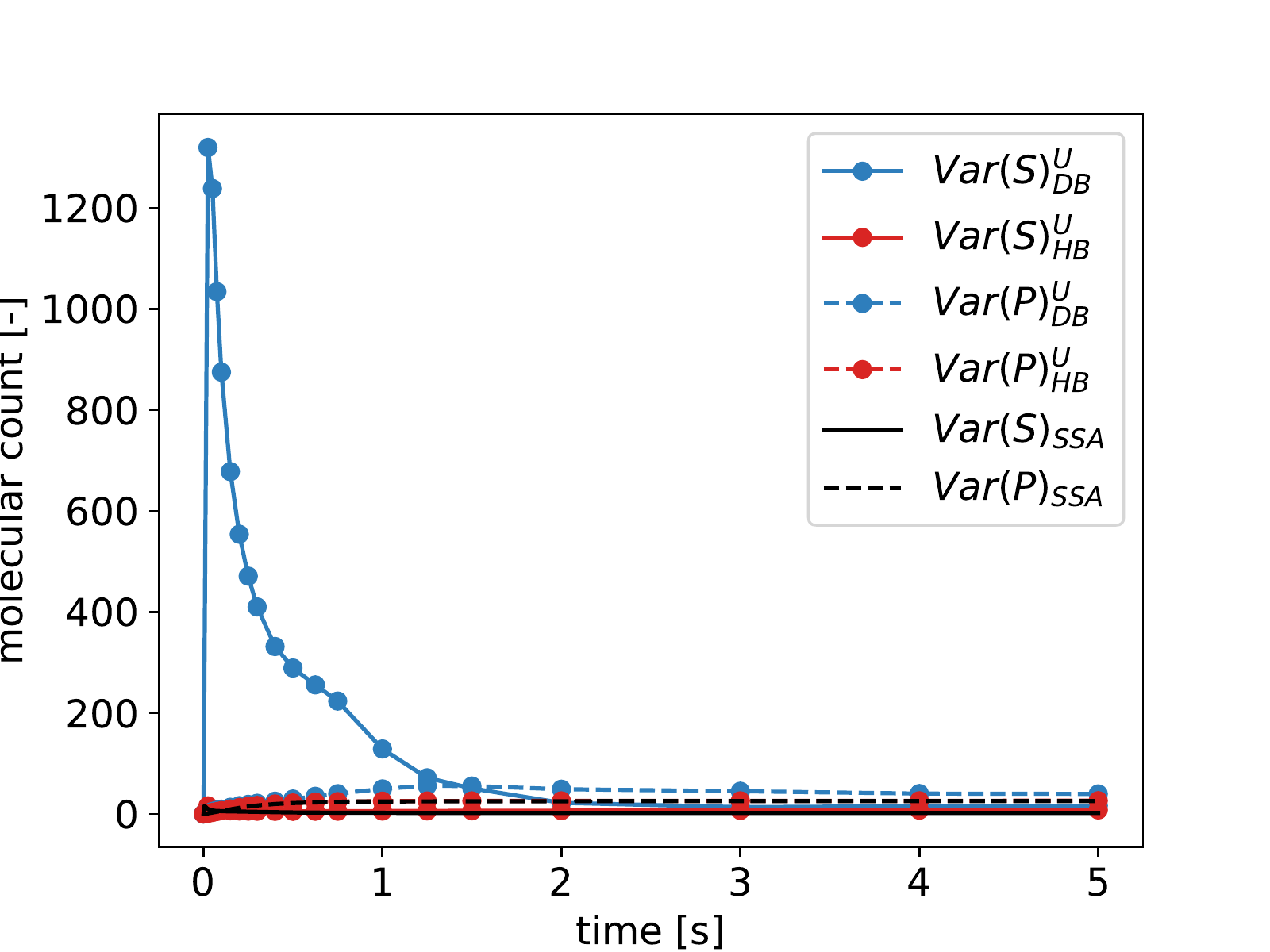}
		\caption{}
	\end{subfigure}
	\begin{subfigure}{0.49\textwidth}
		\includegraphics[width=\textwidth]{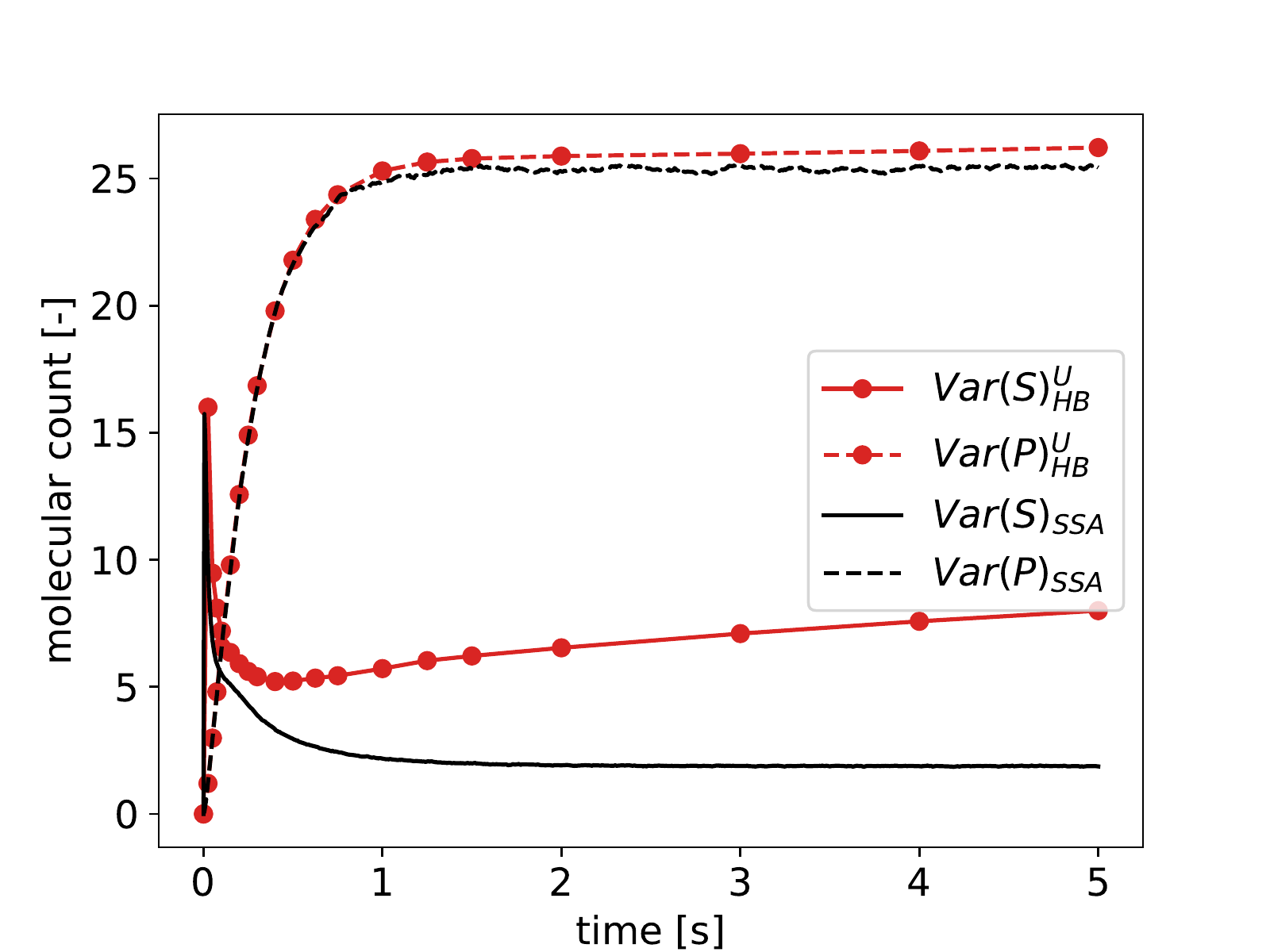}
		\caption{}
	\end{subfigure}
	\caption[Bounds on means and variances of the molecular counts of species \ce{S} and \ce{P} in the metabolic reaction network \eqref{sys:michaelis_menten}]{Bounds on means (a) and variances (b,c) of the molecular counts of species \ce{S} and \ce{P} in the metabolic reaction network \eqref{sys:michaelis_menten}; initial state: $x_{\ce{S},0} = x_{\ce{E},0} = 100$, $x_{\ce{P},0} = x_{\ce{S}:\ce{E},0}=0$; kinetic parameters: $\bm{c} = (1, 1, 1, 1) \, \si{\per \second}$; hierarchy parameters: $m = 4$, $n_{\mathsf{F}} =3$, $n_{\mathsf{T}} = 20$.} \label{fig:michaelis_menten}
\end{figure}

\subsubsection{Negative Feedback Biocircuit}
Many efforts of modern synthetic biology culminate in the design of biocircuits subject to stringent constraints on robustness and performance. Upon successful design, the implications of such tailored biocircuits are often far reaching, even addressing global challenges such as water pollution \cite{sinha2010reprogramming} and energy \cite{peralta2012microbial}. Accordingly, in recent years the use of systems theoretic techniques has received considerable attention to conceptualize, better understand and speed up the design process of biocircuits  \cite{del2016control}. In this context, \citet{sakurai2018optimization} demonstrated the utility of stationary moment bounds for the design of biocircuits subject to robustness constraints. Here, we demonstrate that the proposed methodology could enable an extension of their analysis to the dynamic case. 

We examine the negative feedback biocircuit illustrated in Figure \ref{sys:neg_feedback_biocircuit} studied by \citet{sakurai2018optimization}. The corresponding reaction network is given by 
\begin{align}
\begin{array}{l}
	\ce{DNA} \stackrel{c_1}{\rightarrow} \ce{DNA} + \ce{mRNA} \\
	\ce{mRNA} \stackrel{c_2}{\rightarrow} \emptyset \\
	\ce{mRNA} \stackrel{c_3}{\rightarrow} \ce{mRNA} + \ce{P} \\
	\ce{P} \stackrel{c_4}{\rightarrow} \emptyset \\
	\ce{P} + \ce{DNA}  \stackrel[c_6]{c_5}{\rightleftharpoons} \ce{P}:\ce{DNA} 
\end{array}. \label{eq:neg_feedback_biocircuit}
\end{align}
Figure \ref{fig:neg_feedback_biocircuit} illustrates the obtained bounds on means and variance of the molecular counts of the species \ce{mRNA} and \ce{P}. The bounds are of high quality and may provide useful information for robustness analysis as the noise level measured by the variance changes significantly over the time horizon until the steady-state value is reached.   
\begin{figure}
	\centering
	\includegraphics[scale=0.8]{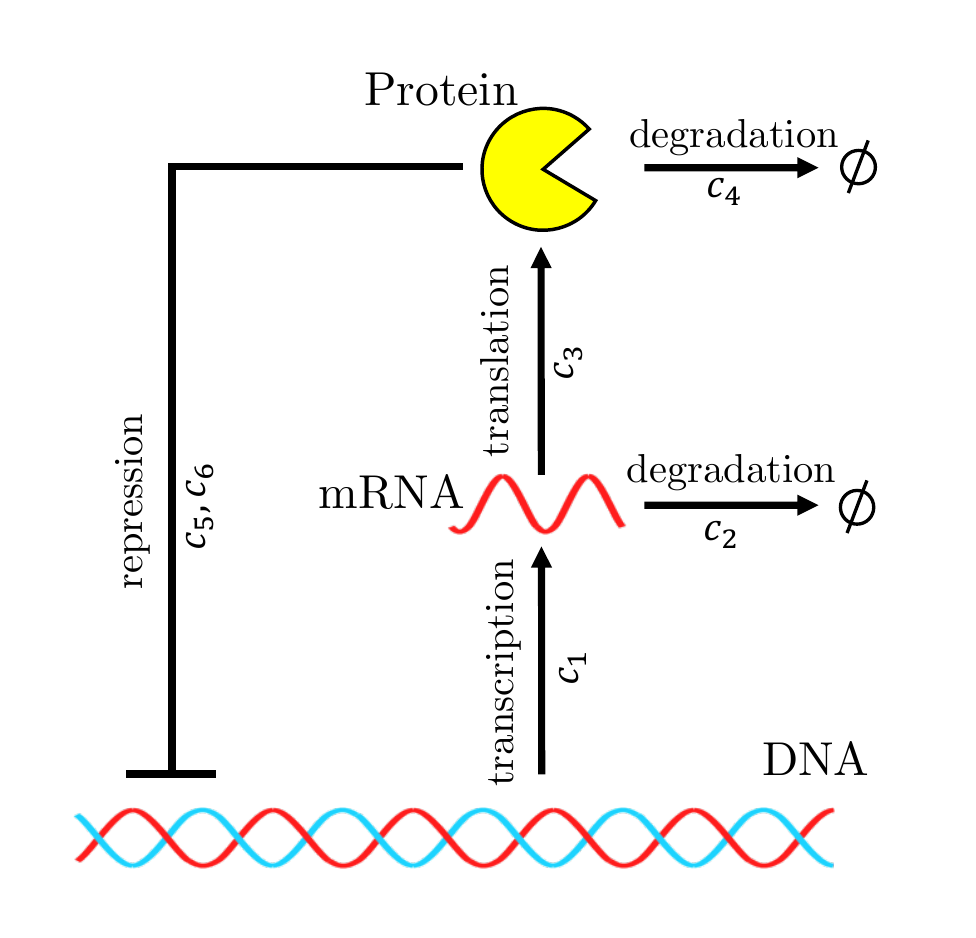}
	\caption{Negative feedback biocircuit from \citet{sakurai2018optimization}}\label{sys:neg_feedback_biocircuit}
\end{figure}
\begin{figure}
	\centering
	\begin{subfigure}{0.49\textwidth}
		\includegraphics[width=0.99\textwidth]{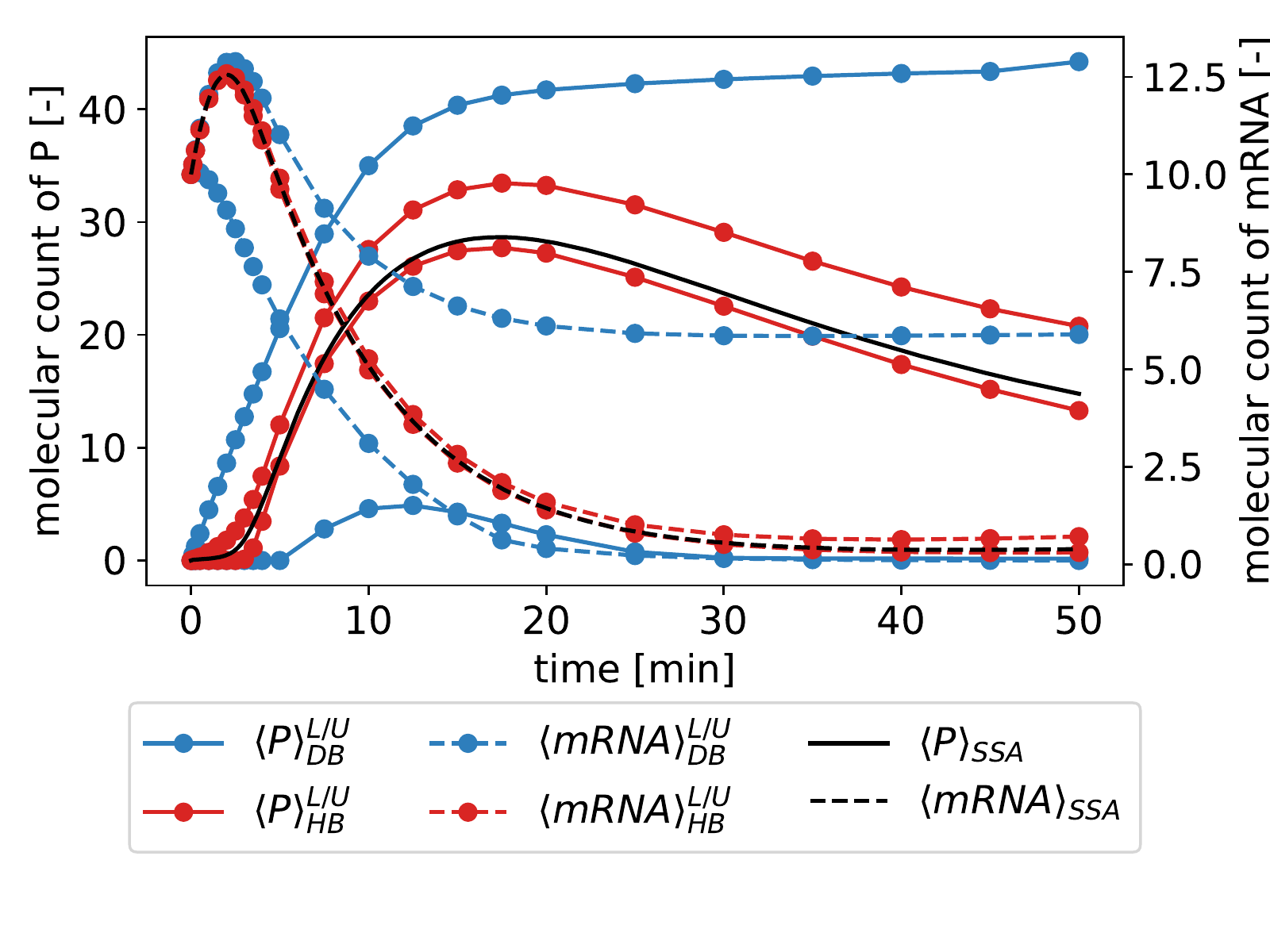}
		\caption{}
	\end{subfigure}
	\begin{subfigure}{0.49\textwidth}
		\includegraphics[width=0.99\textwidth]{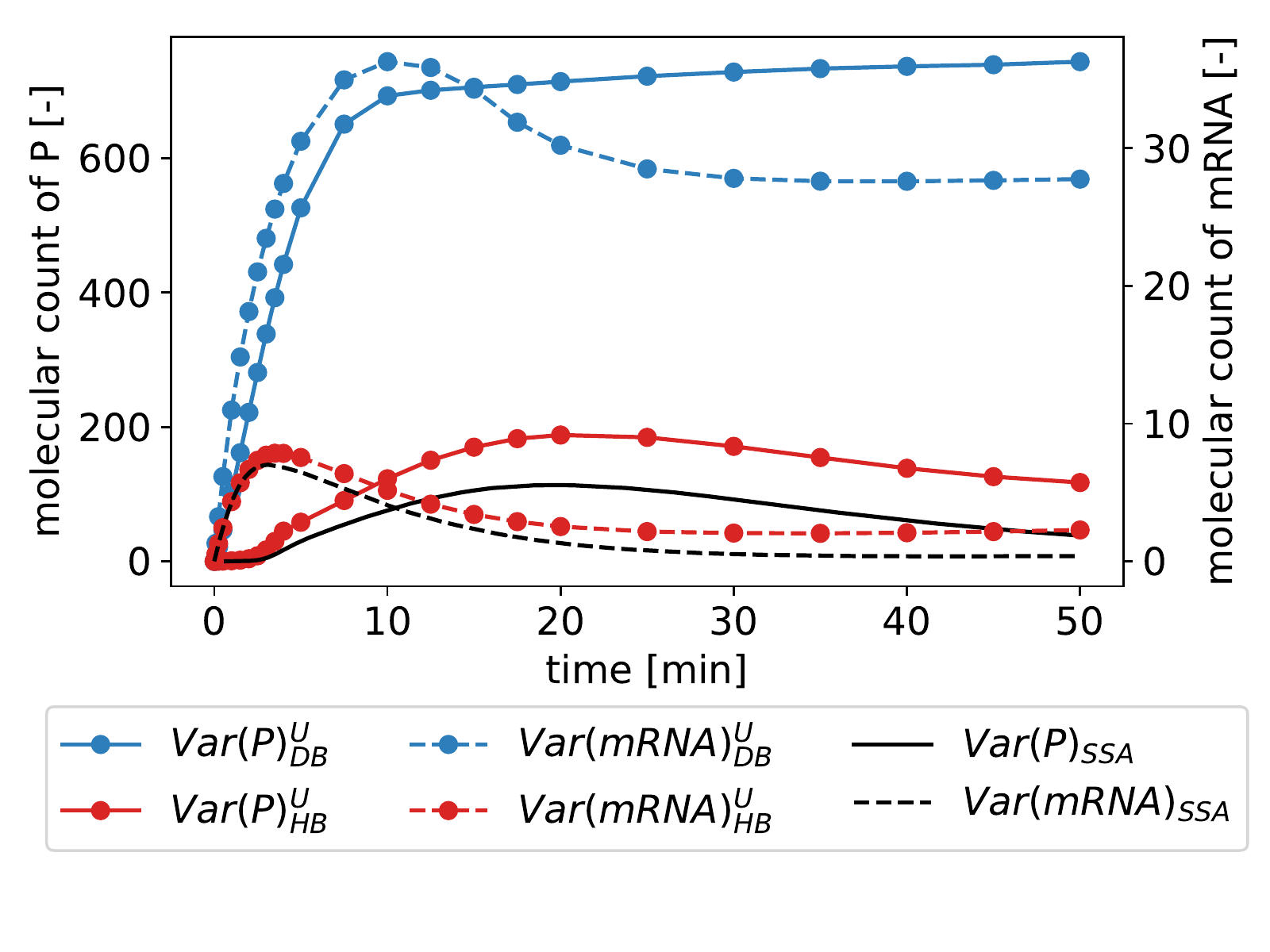}
		\caption{}
	\end{subfigure}
	\caption[Bounds on means and variances of the molecular counts of $\ce{mRNA}$ and $P$ in the negative feedback biocircuit illustrated in Figure \ref{sys:neg_feedback_biocircuit}]{Bounds on means (a) and variances (b) of the molecular counts of $\ce{mRNA}$ and $P$ in the negative feedback biocircuit illustrated in Figure \ref{sys:neg_feedback_biocircuit}; initial state: $x_{\ce{mRNA},0} = 10$, $x_{\ce{P},0} = x_{\ce{P}:\ce{DNA},0} = 0$, $x_{\ce{DNA},0} = 20$; kinetic parameters: $\bm{c} = (0.2, \ln(2)/5, 0.5, \ln(2)/20, 5, 1) \, \si{\per \minute}$; hierarchy parameters: $m = 4$, $n_{\mathsf{F}} = 3$, $n_{\mathsf{T}} = 20$.} \label{fig:neg_feedback_biocircuit}
\end{figure}

\subsubsection{Viral Infection}
As our last example we consider the following reaction network from \citet{srivastava2002stochastic} used to model for the intracellular kinetics of a virus:
\begin{align}
	\begin{array}{lll}
		\ce{G} \stackrel{c_1}{\rightarrow} \ce{T},\quad &\ce{T} \stackrel{c_2}{\rightarrow} \emptyset, \quad &\ce{T} \stackrel{c_3}{\rightarrow} \ce{T} + \ce{G} \\
		\ce{G} + \ce{S} \stackrel{c_4}{\rightarrow} \emptyset, \quad &\ce{T} \stackrel{c_5}{\rightarrow}  \ce{T} + \ce{S},  \quad &\ce{S} \stackrel{c_6}{\rightarrow} \emptyset 
	\end{array}. \label{sys:viral_infection}
\end{align}
In this network, \ce{S} represents the viral structural protein while \ce{T} and \ce{G} represent the viral nucleic acids categorized as template and genomic, respectively. 

Studying a system undergoing the reactions as described by the above network with sampling based approaches is computationally expensive. This is due to two distinct reasons. On the one hand, the state space is infinite. On the other hand, the molecular counts of \ce{S}, \ce{T} and \ce{G} vary over several orders of magnitude so that they evolve at different time scales and noise levels. This characteristic is found in many biochemical reaction networks and has motivated the development of hybrid simulation techniques; see for example \citet{haseltine2002approximate}. Such hybrid approaches combine the use of deterministic models for the dynamics of species present at large counts with stochastic models for rare events and species present at low counts. Although such methods accelerate simulation substantially, they generally introduce a range of assumptions accompanied by an unknown error. The methodology presented in this article provides a framework to quantify this error. Figure \ref{fig:viral_infection} shows the bounds obtained for the mean molecular counts of all three species. Although the bounds are not tight, we argue that they may be informative enough to assess whether approximate solutions are reasonable. This in stark contrast to Dowdy and Barton's method \cite{dowdy2018dynamic} which in this case provides extremely loose bounds that could not be substantially improved due to numerical difficulties and prohibitive computational cost at high truncation orders. 

\begin{figure}
	\centering
	\begin{subfigure}{0.49\textwidth}
		\includegraphics[width=\textwidth]{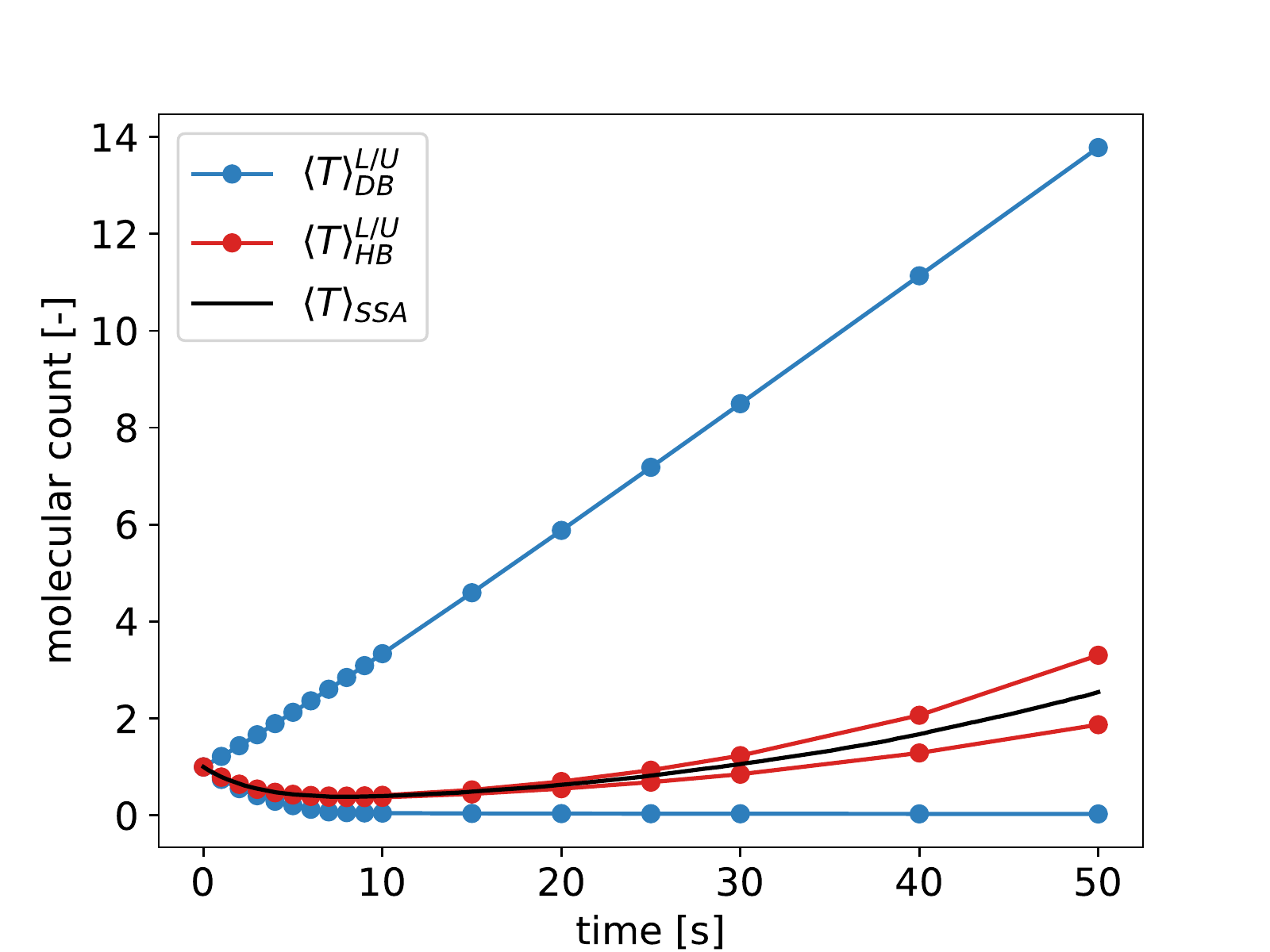}
		\caption{}
	\end{subfigure}
	\begin{subfigure}{0.49\textwidth}
		\includegraphics[width=\textwidth]{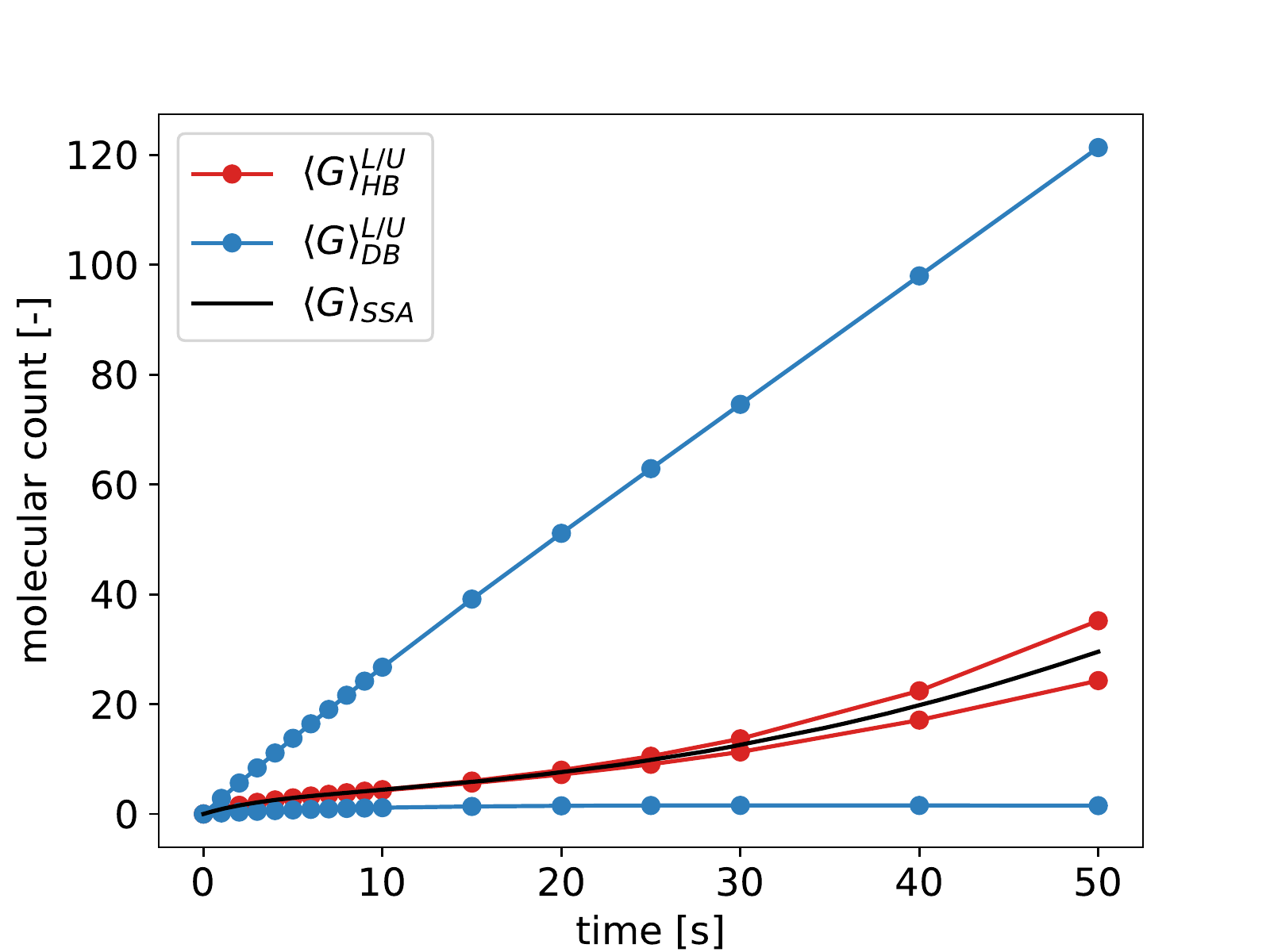}
		\caption{}
	\end{subfigure}
	\begin{subfigure}{0.49\textwidth}
		\includegraphics[width=\textwidth]{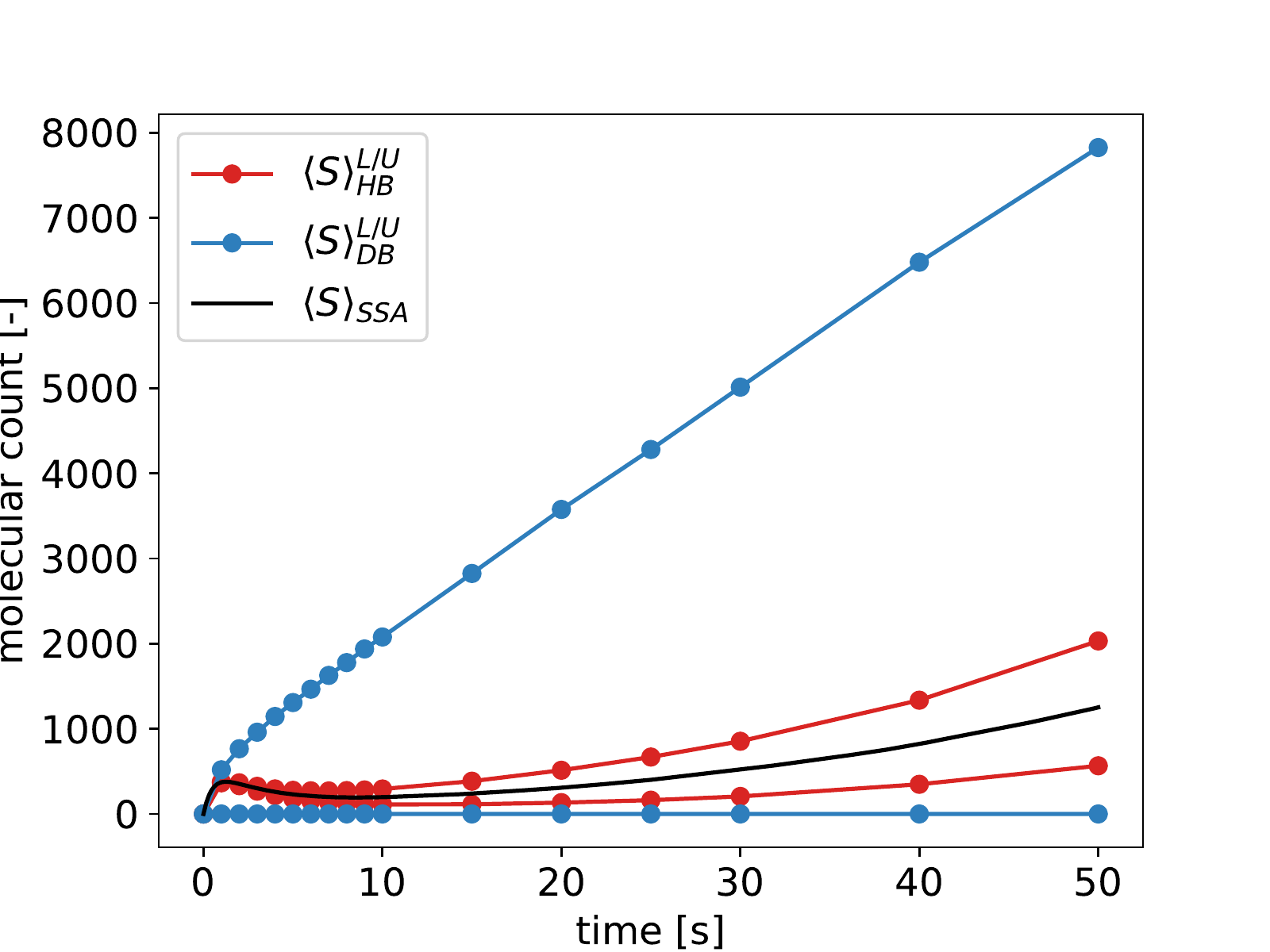}
		\caption{}
	\end{subfigure}
	\caption[Bounds on mean molecular counts of the species \ce{Template} (a), \ce{Genome} (b) and \ce{Struct} (c) in the viral infection network \eqref{sys:viral_infection}]{Bounds on mean molecular counts of \ce{Template} (a), \ce{Genome} (b) and \ce{Struct} (c) species in the viral infection network \eqref{sys:viral_infection}; initial state: $x_{\ce{T},0} = 1$, $x_{\ce{G},0} = x_{\ce{S},0}=0$;	kinetic parameters: $\bm{c} = (0.025, 0.25, 1, \SI{7.5e-6}, \SI{1e3}, 1.99) \, \si{\per \day}$; hierarchy parameters: $m=4$, $n_{\mathsf{F}}=3$, $n_{\mathsf{T}} = 20$.} \label{fig:viral_infection}
\end{figure}

\section{Bound Tightening Mechanisms}\label{sec:bounding_mechanisms}
In this section, we briefly assess the effect of the different bound tightening mechanisms provided by the proposed bounding hierarchy. We conduct this empirical analysis on the basis of the birth-death process \eqref{sys:birth_death}. Furthermore, we restrict our considerations here to studying the effect of increasing the truncation order $m$, the hierarchy level $n_I$ and the number of time points $n_{\mathsf{T}}$ used to discretize the horizon; throughout, we only use the constant test function ($n_{\mathsf{F}} = 1$). 

Figure \ref{fig:isolated_tightening} shows the effect of isolated changes in the different hierarchy parameters on the bounds obtained for the mean molecular count and its variance. The results indicate that all bound tightening mechanisms, when used in isolation, appear to suffer from diminishing returns, eventually causing the bounds to stall. Moreover, solely increasing the truncation order $m$ appears insufficient to provide informative bounds over a long time horizon in this example; increasing either the number of time points or the hierarchy level $n_I$ are significantly more effective in comparison.

Figure \ref{fig:joint_tightening} shows the effect of joint changes in the considered hierarchy parameters on the tightness of bounds on the mean molecular count. The figure indicates that jointly changing the hierarchy parameters effectively mitigates stalling of the bounds in this example such that significantly tighter bounds are obtained overall. While the general trends illustrated by Figures \ref{fig:isolated_tightening} and \ref{fig:joint_tightening} align well with our experiences for a range of other examples, we wish to emphasize that it is in general hard to predict which combination of hierarchy parameters provides the best trade-off between computational cost and bound quality; when the choice of test functions is added to the equation, the situation becomes even more complicated. Moreover, as Figure \ref{fig:projections} illustrates, the feasible region of the bounding SDPs shrinks anisotropically and, more importantly, with different intensity along different directions for different bound tightening mechanisms. As a consequence, the optimal choice of the hierarchy parameters is in general not only dependent on the system under investigation but also on the statistical quantity to be bounded. 

In summary, the results presented in this section underline the value of the additional bound tightening mechanisms offered by the proposed hierarchy; however, they also emphasize the need for better guidelines to enable an effective use of the tightening mechanisms in practice. 

\begin{figure*}
    \centering
	\begin{subfigure}{0.4\textwidth}
		\includegraphics[width=1.0\textwidth]{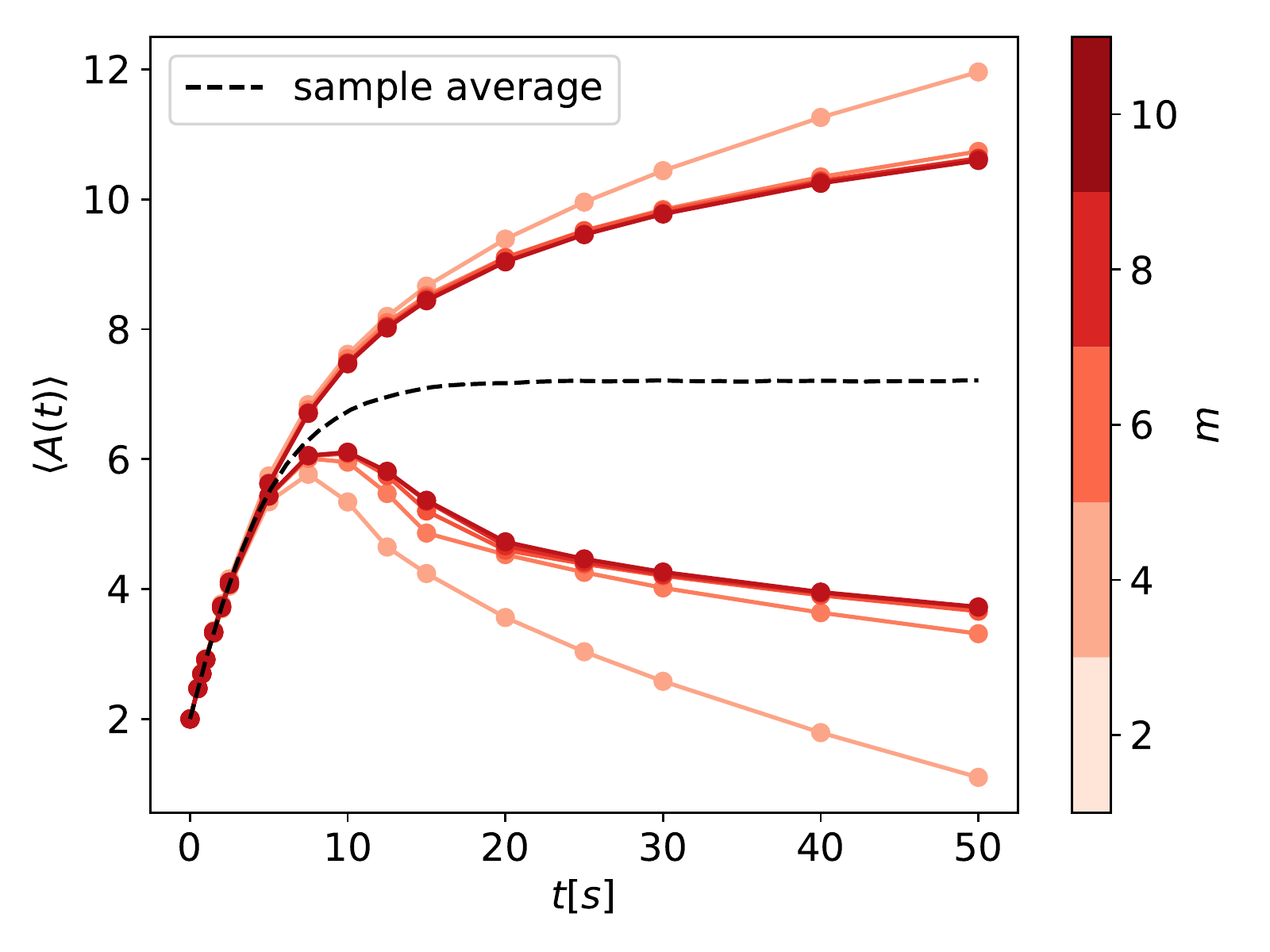}
		\caption{$n_{I} = 2$, $n_{\mathsf{T}} = 2$}
	\end{subfigure}
	\begin{subfigure}{0.4\textwidth}
		\includegraphics[width=1.0\textwidth]{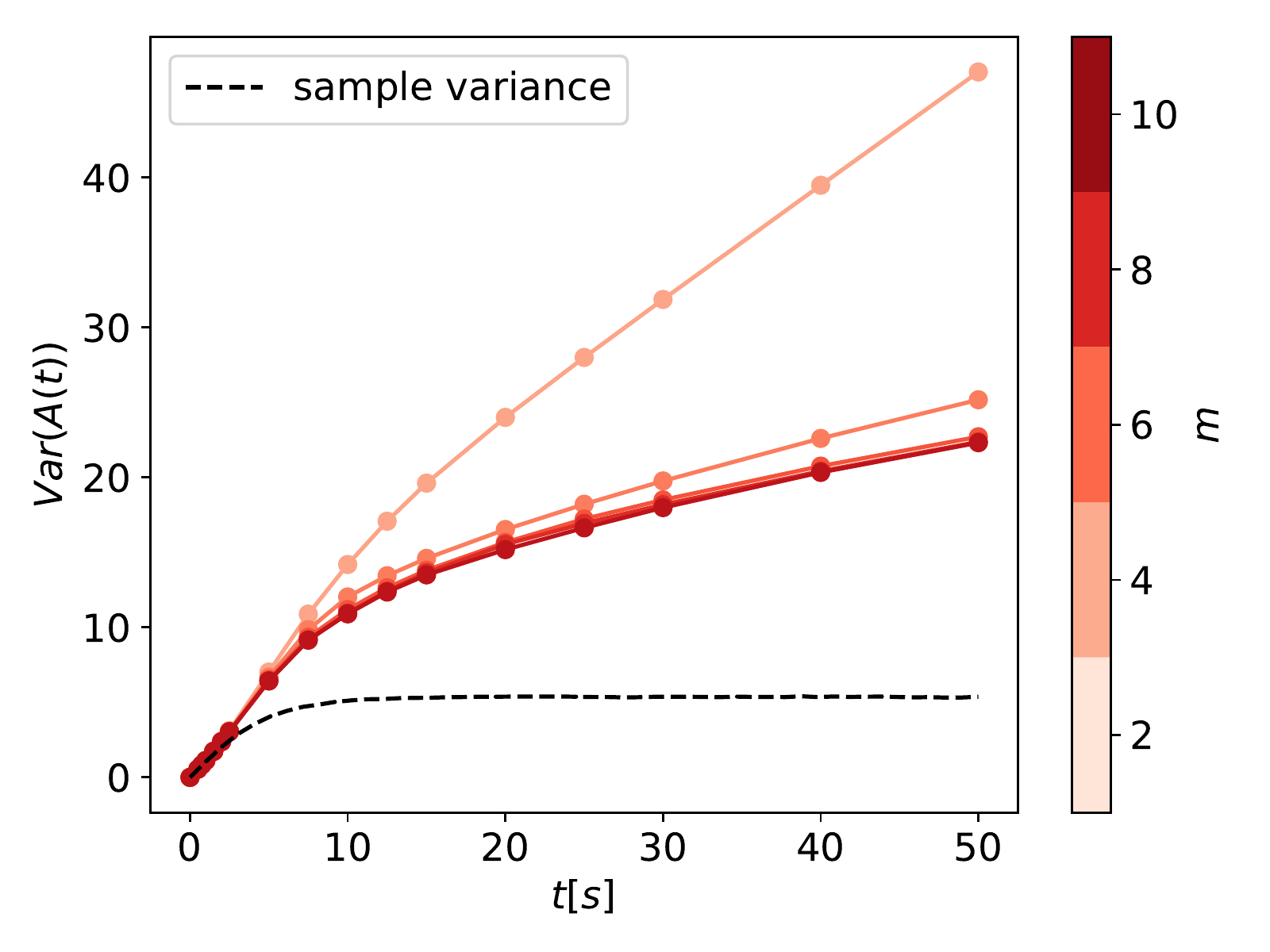}
		\caption{$n_{I} = 2$, $n_{\mathsf{T}} = 2$}
	\end{subfigure}
	\begin{subfigure}{0.4\textwidth}
		\includegraphics[width=1.0\textwidth]{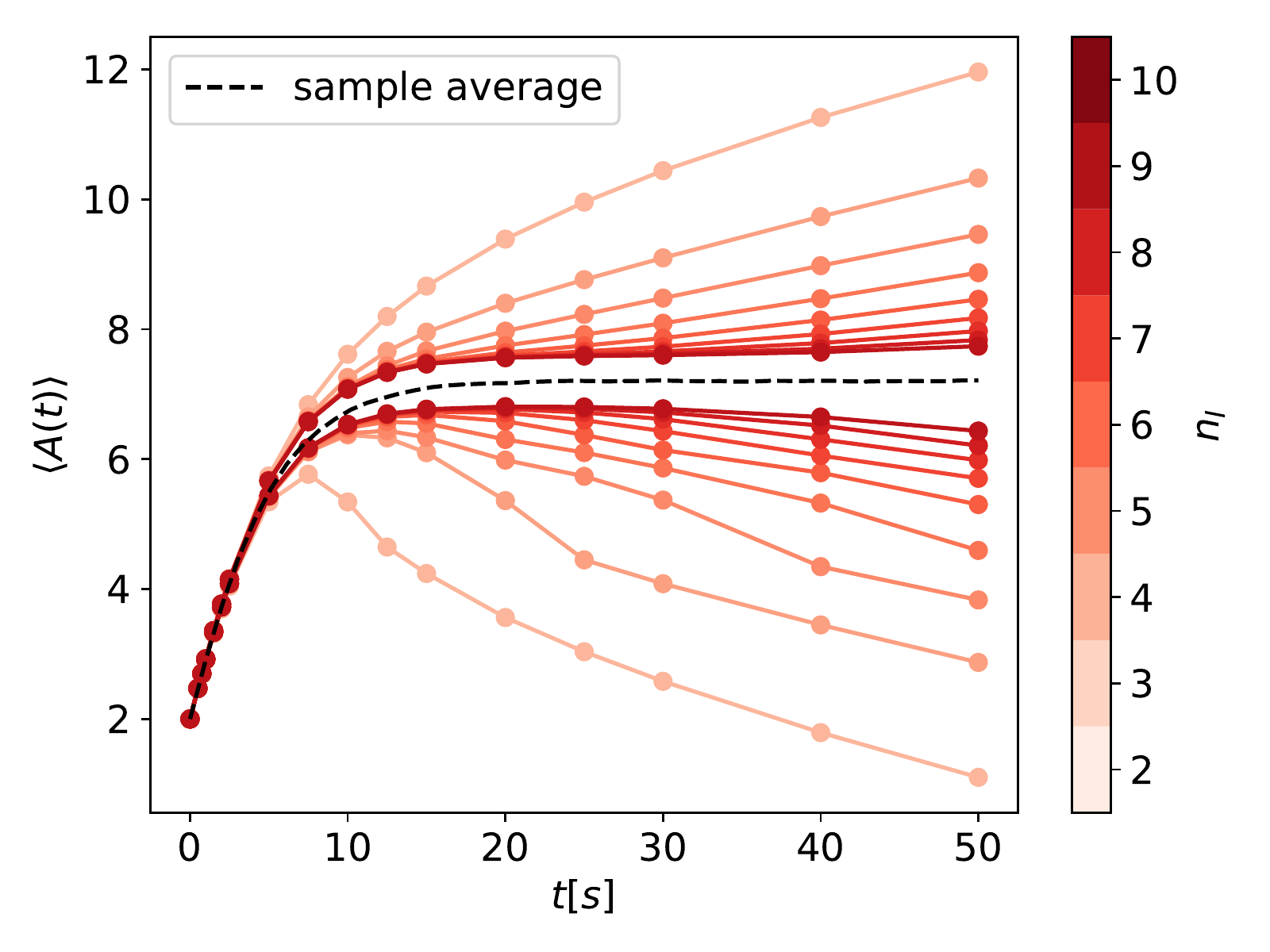}
		\caption{$m = 2$, $n_{\mathsf{T}} = 2$}
	\end{subfigure}
	\begin{subfigure}{0.4\textwidth}
		\includegraphics[width=1.0\textwidth]{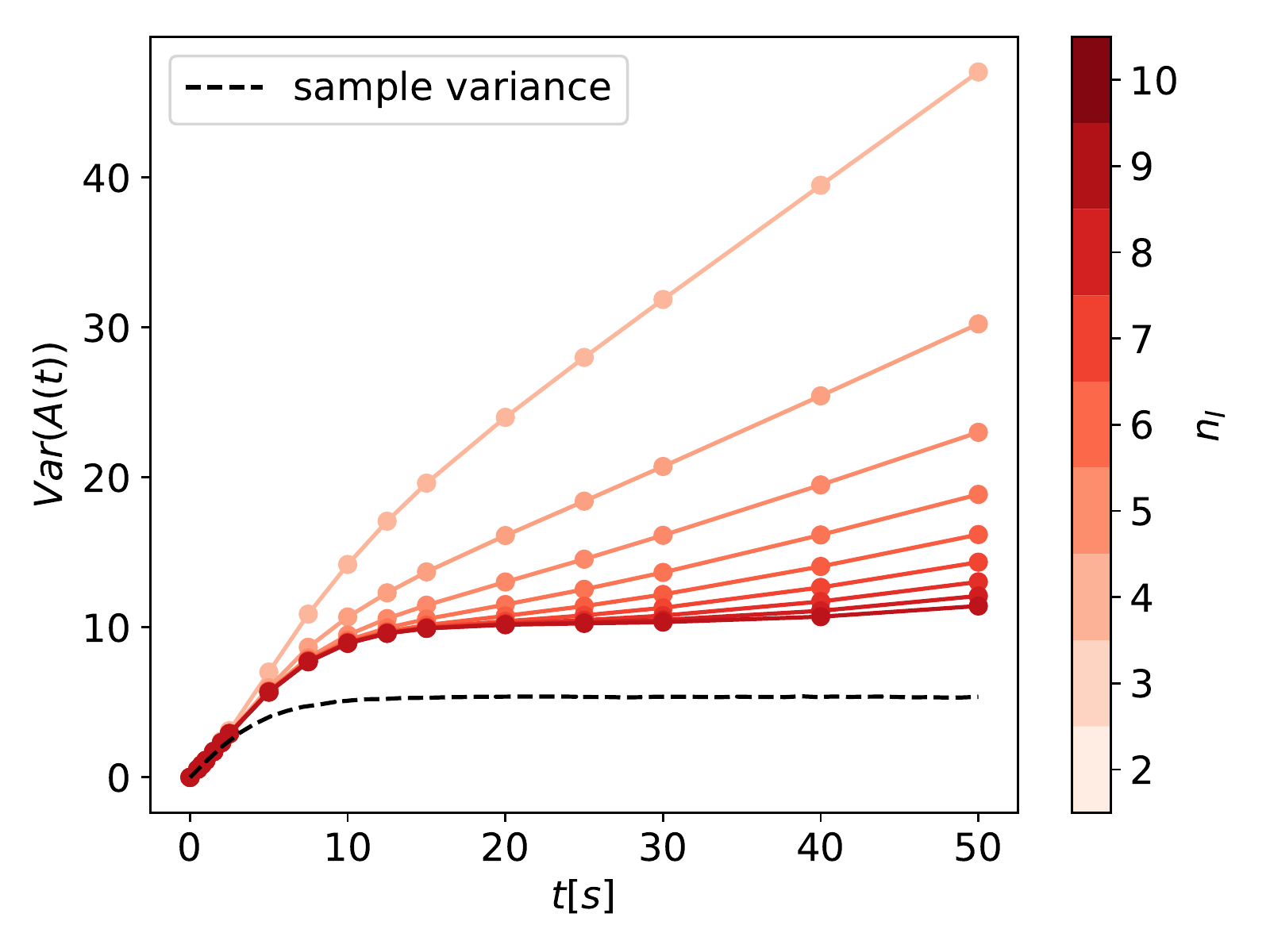}
			\caption{$m = 2$, $n_{\mathsf{T}} = 2$}
	\end{subfigure}
	\begin{subfigure}{0.4\textwidth}
		\includegraphics[width=1.0\textwidth]{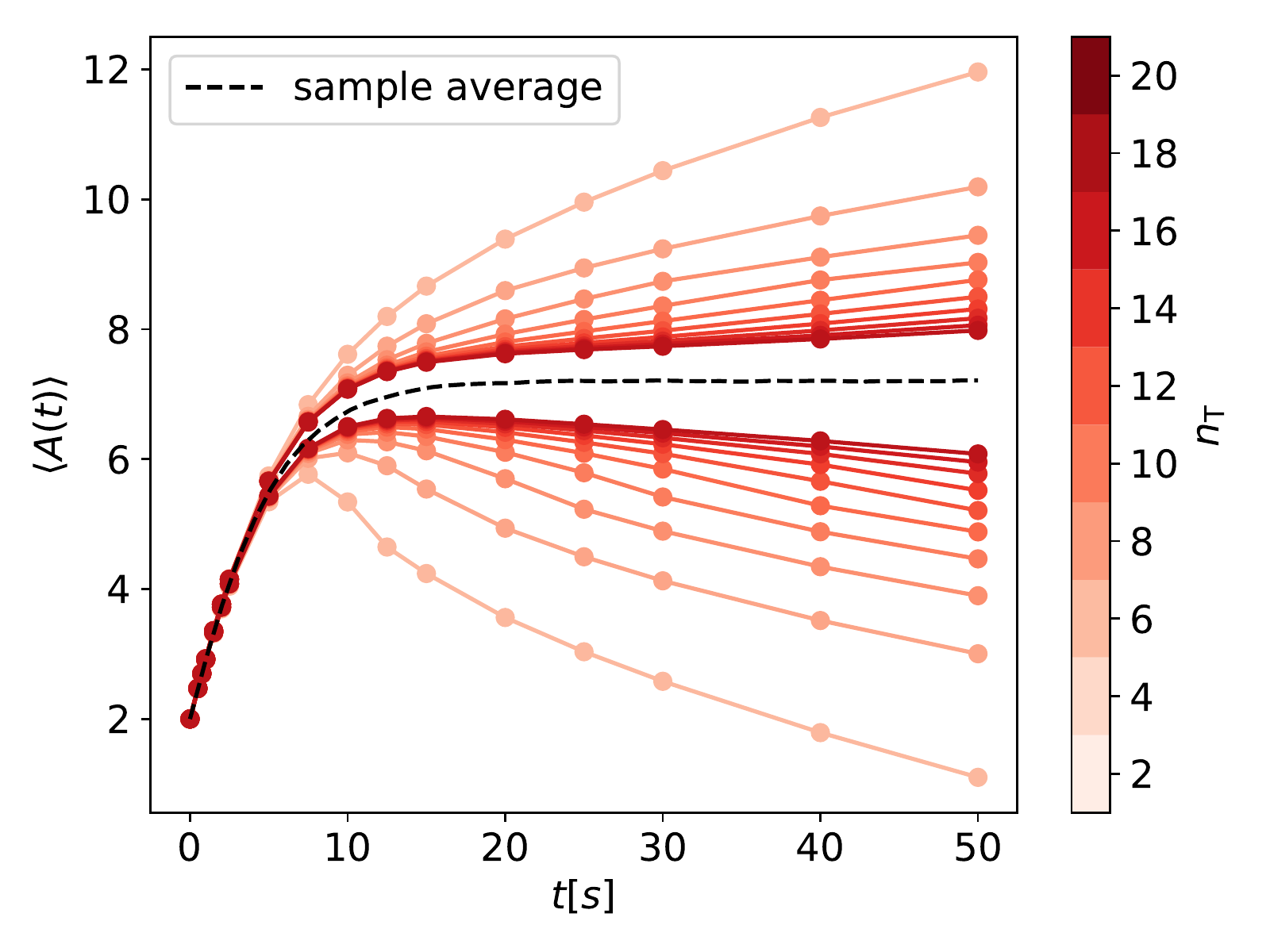}
		\caption{$m=2$, $n_{I} = 2$}
	\end{subfigure}
	\begin{subfigure}{0.4\textwidth}
		\includegraphics[width=1.0\textwidth]{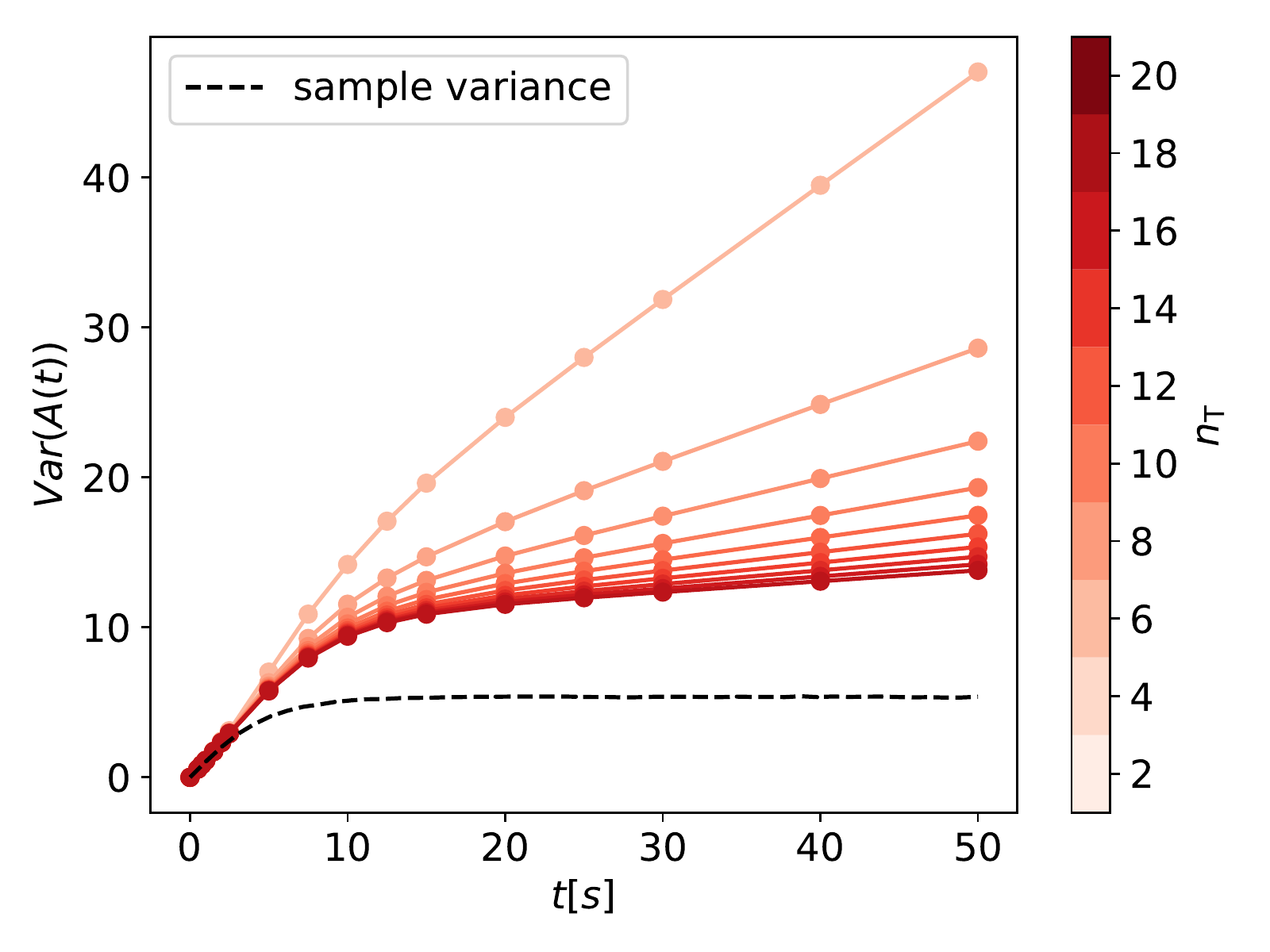}
		\caption{$m=2$, $n_{I} = 2$}
	\end{subfigure}
    \caption{Bounds on the trajectories of the mean molecular count and variance of the birth-death process \eqref{sys:birth_death} for increasing $m$ (a,b), $n_I$ (c,d), and $n_{\mathsf{T}}$ (e,f) compared against the empiric sample mean and variance generated with Gillespie's SSA. In each figure only one parameter is varied while the others are held constant at the level indicated in the subcaptions.}\label{fig:isolated_tightening}
\end{figure*}

\begin{figure*}
    \centering
    \begin{subfigure}{0.49\textwidth}
        \centering
        \includegraphics[width=1.0\textwidth]{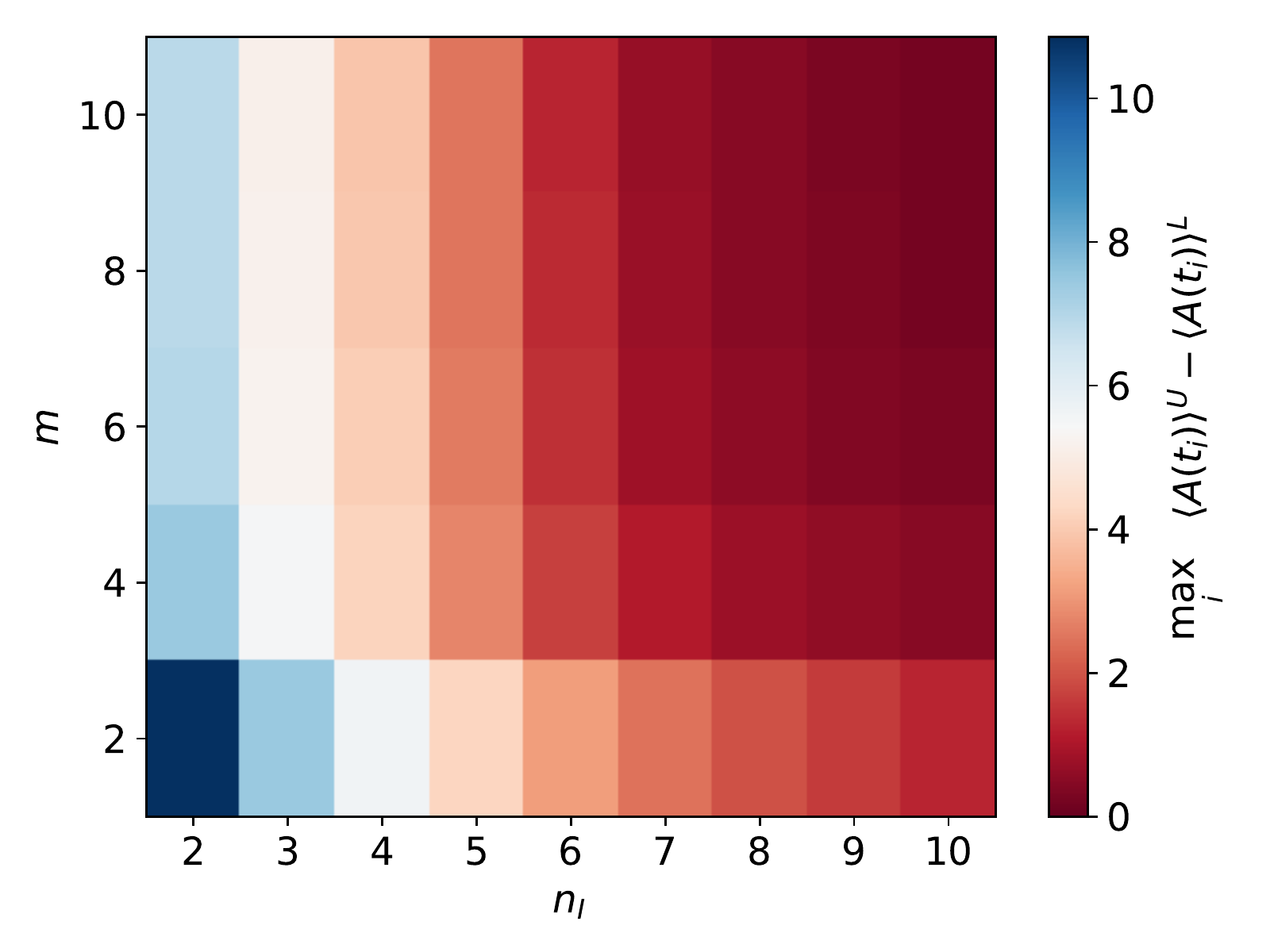}
        \caption{$n_{\mathsf{T}}=2$}    
    \end{subfigure}
    \begin{subfigure}{0.49\textwidth}
        \centering
        \includegraphics[width=1.0\textwidth]{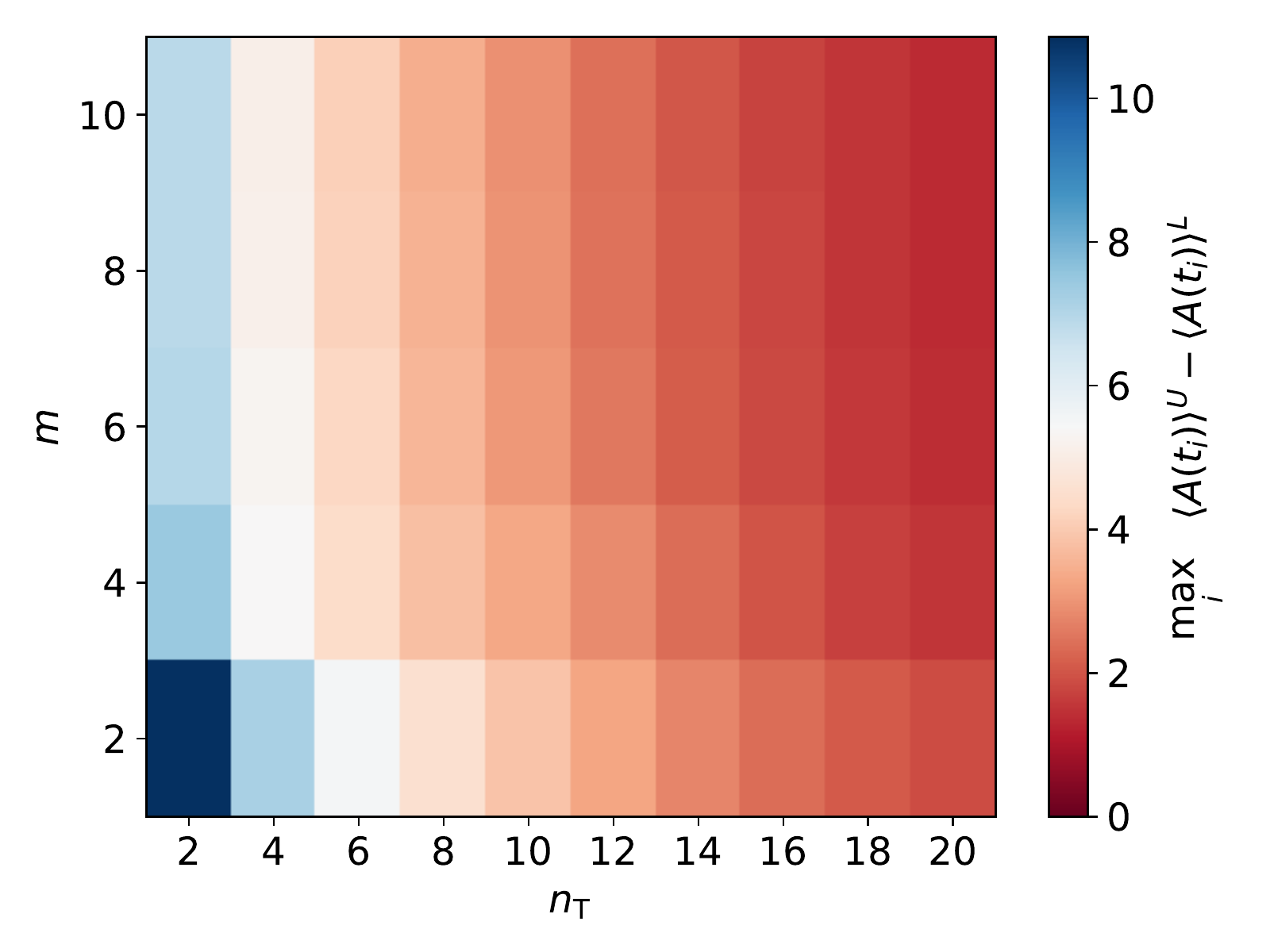}
        \caption{$n = 2$}    
    \end{subfigure}
    \begin{subfigure}{0.49\textwidth}
        \centering
        \includegraphics[width=1.0\textwidth]{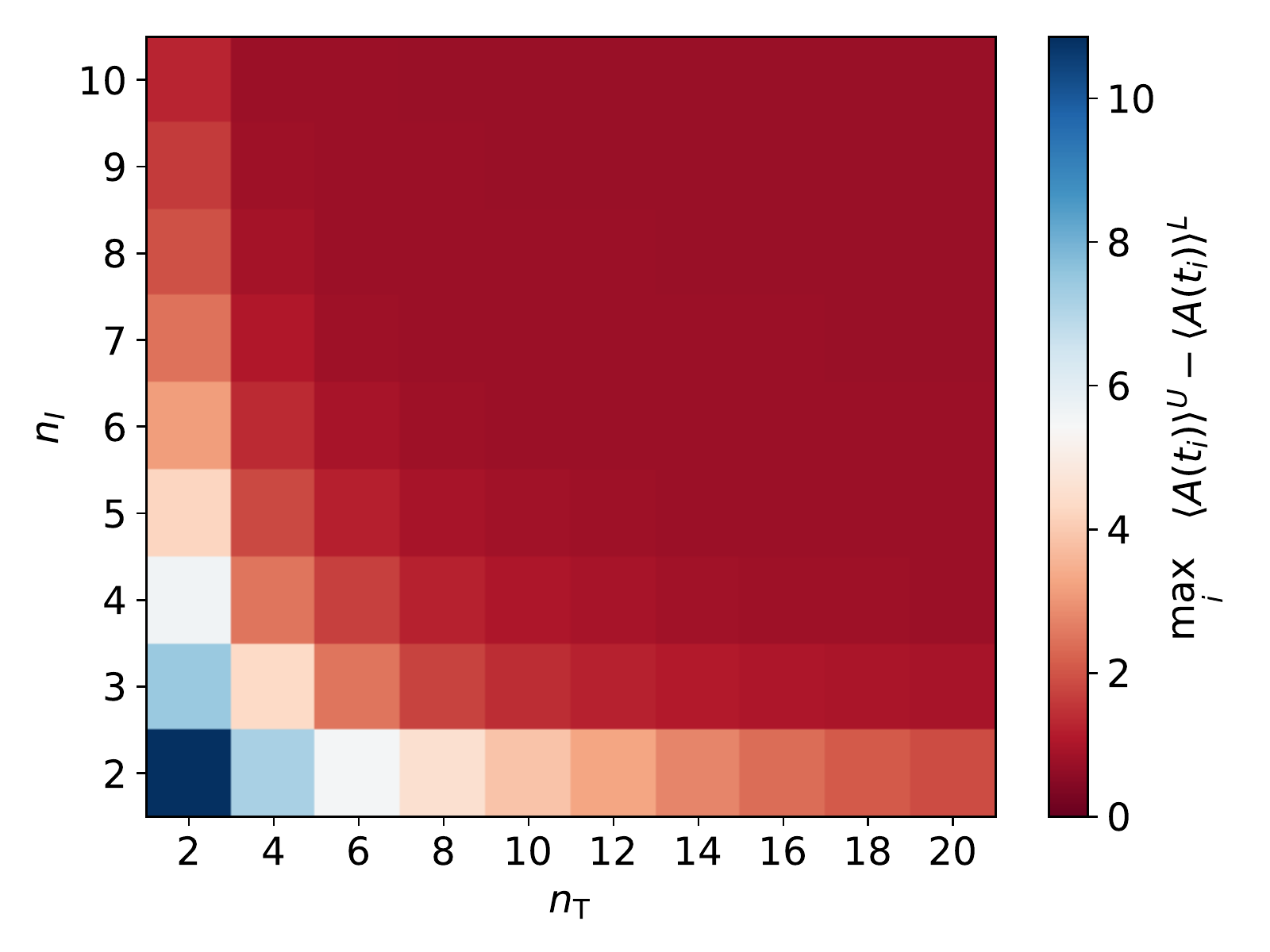}
        \caption{$m = 2$}    
    \end{subfigure}
    \caption[Effect of joint changes in hierarchy parameters $m$, $n_I$ and $n_{\mathsf{T}}$ on tightness of bounds]{Maximum gap between upper and lower bounds on mean molecular counts among the time points probed along the time horizon for joint changes in $m$ and $n_I$ (a), $m$ and $n_{\mathsf{T}}$ (b), and $n_I$ and $n_{\mathsf{T}}$ (c).} \label{fig:joint_tightening}
\end{figure*}

\begin{figure*}
    \centering
    \begin{subfigure}{0.49\textwidth}
        \centering
        \includegraphics[width=1.0\textwidth]{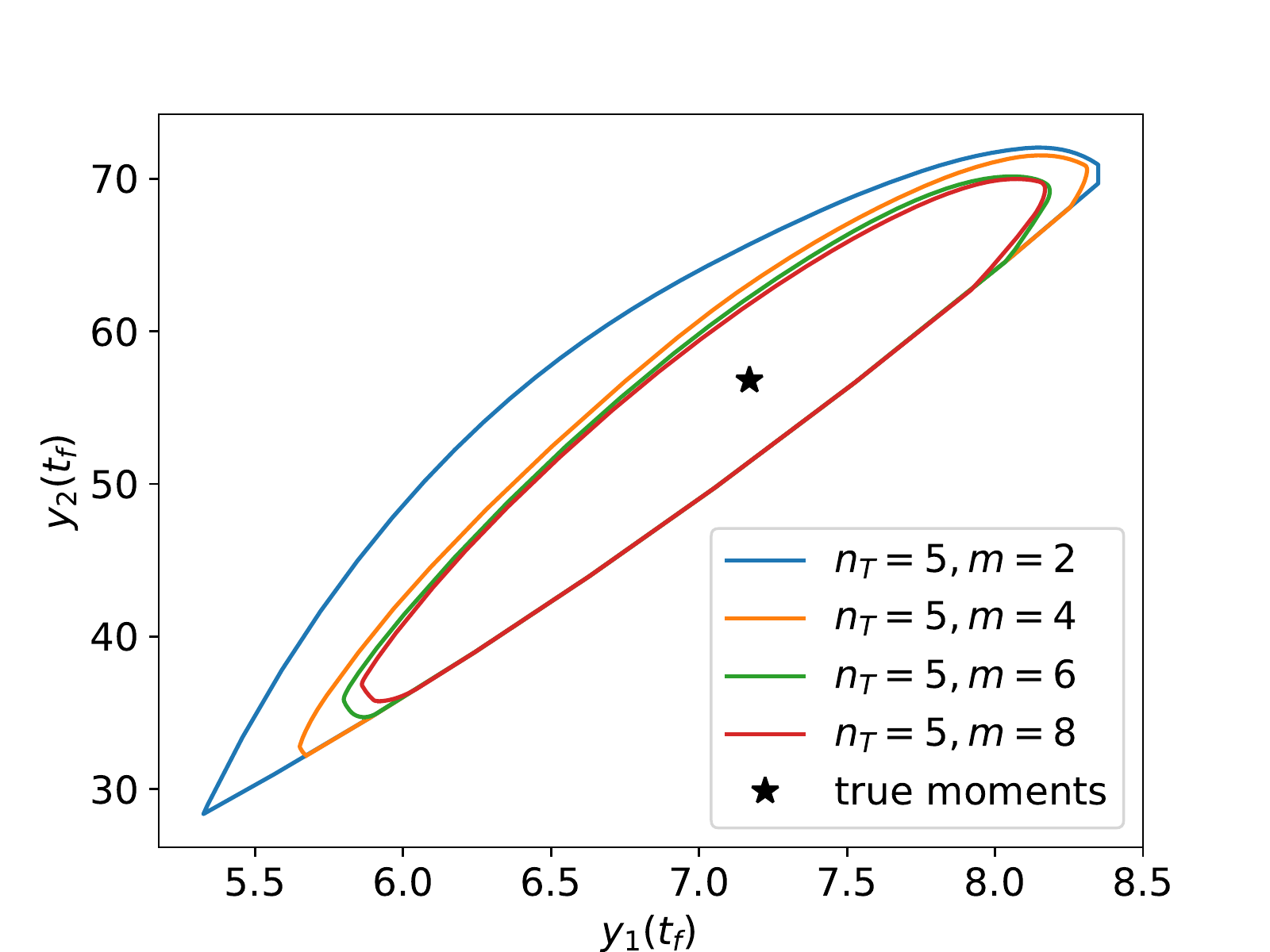}
        \caption{}    
    \end{subfigure}
    \begin{subfigure}{0.49\textwidth}
        \centering
        \includegraphics[width=1.0\textwidth]{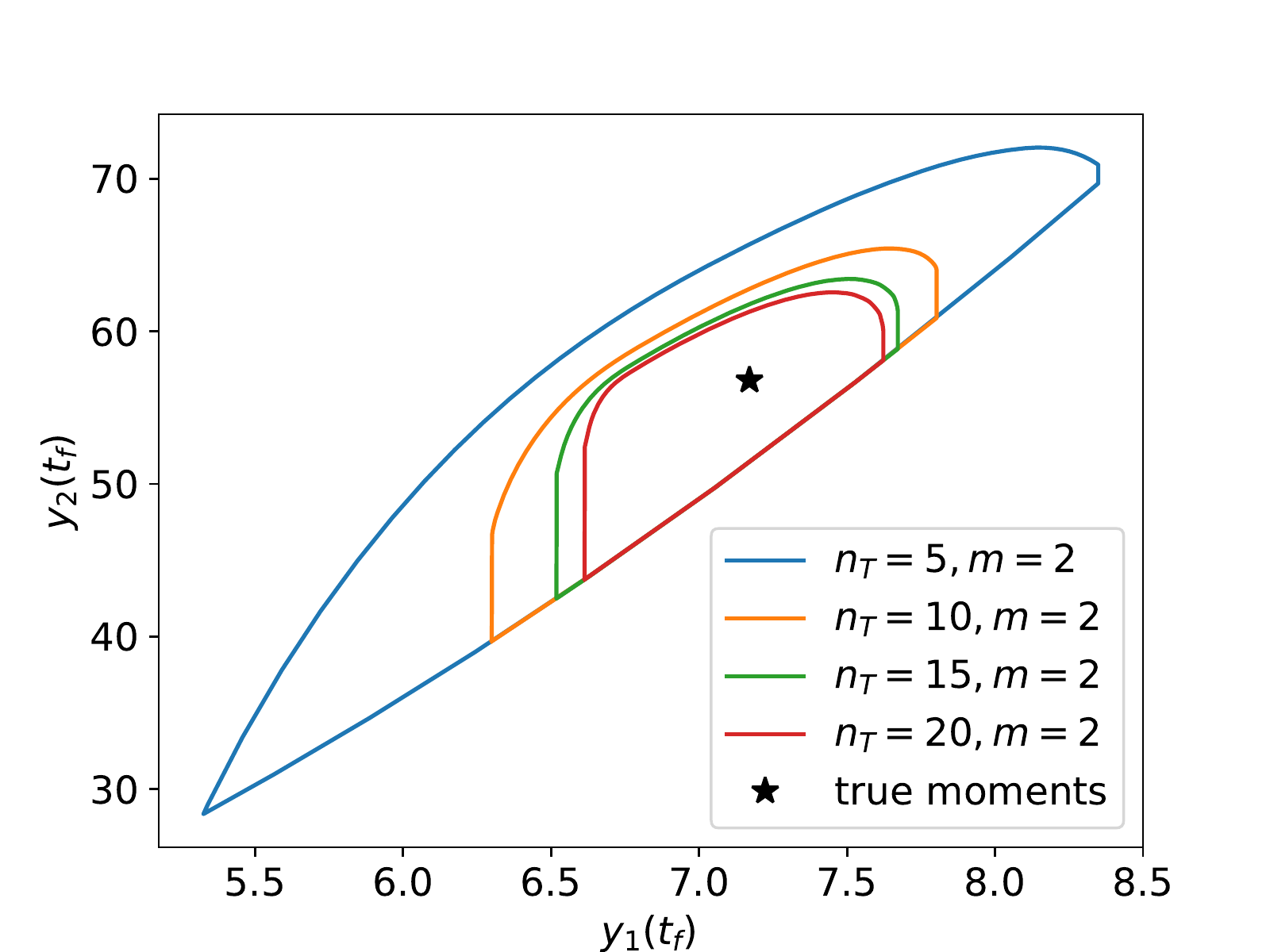}
        \caption{}    
    \end{subfigure}
    \begin{subfigure}{0.49\textwidth}
        \centering
        \includegraphics[width=1.0\textwidth]{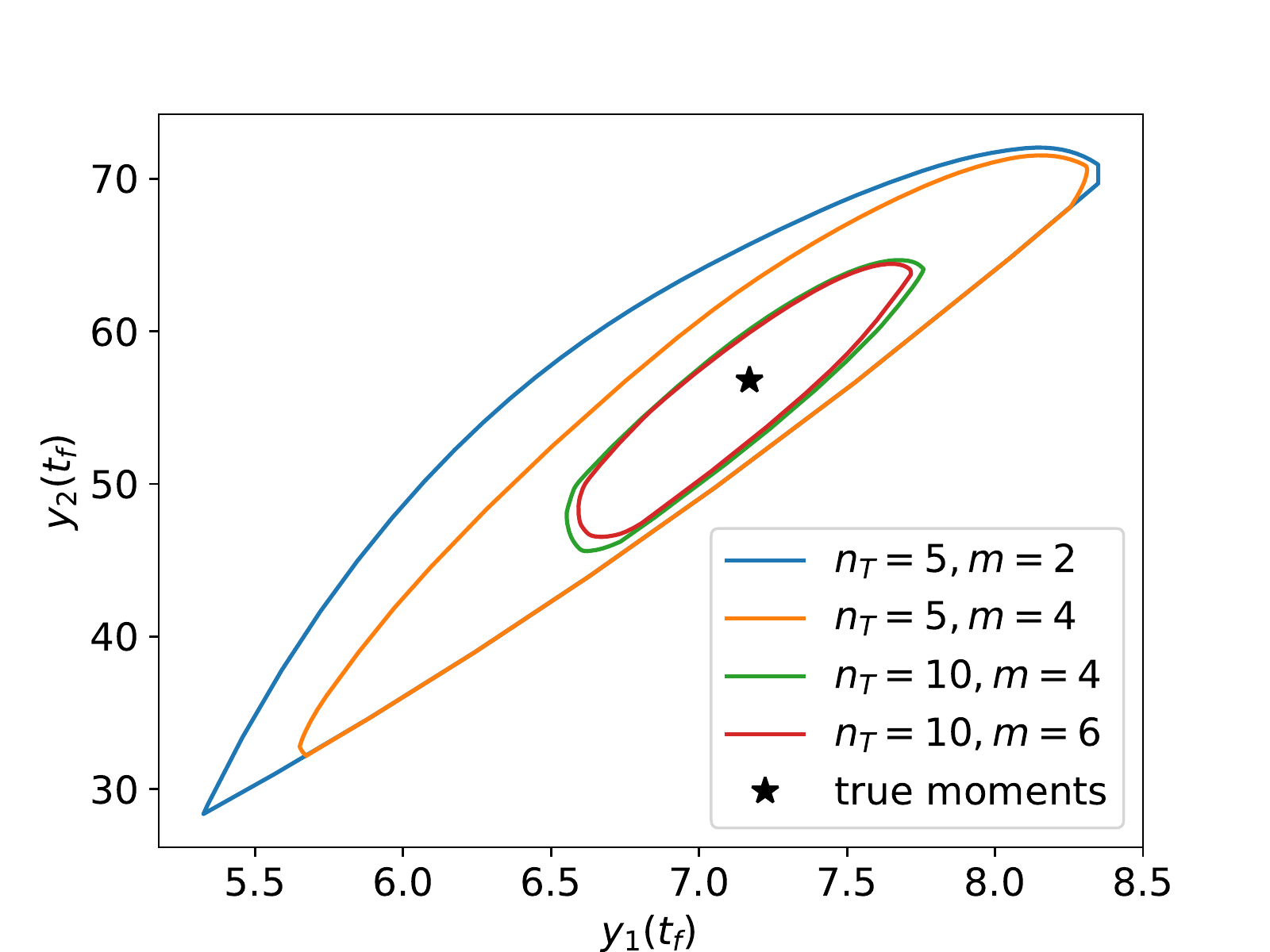}
        \caption{}    
    \end{subfigure}
    \caption{Projection of the feasible set of \eqref{eq:SDP} corresponding to the birth-death process \eqref{sys:birth_death} for (a) increasing truncation order, (b) increasing number of time points $n_{\mathsf{T}}$, and (c) jointly increasing truncation order and number of time points. All projections are obtained for $n_{I} = 2$ and $t_f = 20$.}
    \label{fig:projections}
\end{figure*}

\section{Conclusion}\label{sec:conclusion}
\subsection{Summary}
We have extended the results of \citet{dowdy2018dynamic} by constructing a new hierarchy of convex necessary moment conditions for the moment trajectories of stochastic chemical systems described by the CME. Building on a discretization of the time domain of the problem, the conditions reflect temporal causality and regularity properties of the true moment trajectories. It is proved that the conditions give rise to a hierarchy of highly structured SDPs whose optimal values form a sequence of monotonically improving bounds on the true moment trajectories. Furthermore, the conditions provide new mechanisms to tighten the obtained bounds when compared to the original conditions proposed by \citet{dowdy2018dynamic}. These tightening mechanisms are often a substantially more scalable alternative to the primary tightening mechanism of increasing the truncation order in Dowdy and Barton's approach \cite{dowdy2018dynamic}; most notably, refining the time discretization results in linearly increasing problem sizes, independent of the state space dimension of the system. As an additional advantage, this bound tightening mechanism provides a way to sidestep the poor numerical conditioning of moment-based SDPs featuring high-order moments. Finally, it is demonstrated with several examples that the proposed hierarchy provides bounds that may indeed be useful in practice.

\subsection{Open Questions}\label{sec:questions} 
We close by stating some open questions motivated by our results.
\begin{enumerate} 
	\item In the presented case studies, we naively chose the time points at which the proposed necessary moment conditions were imposed as equidistant. Several results from numerical integration, perhaps most notably Gau{\ss} quadrature, suggest that this is likely not the optimal choice. It would be interesting to examine if and how results from numerical integration can inform improvements of this choice. 
	\item The choice of the hierarchy parameters in the proposed bounding scheme is crucial to achieve a good trade-off between bound quality and computational cost. As indicated by the discussion in Section \ref{sec:bounding_mechanisms}, however, the interplay between the bound tightening mechanisms associated with the different hierarchy parameters and their effect on the bound quality remains poorly understood. Accordingly, we believe that assessing the trade-offs offered by the different bound tightening mechanisms in greater detail and developing more rigorous guidelines on how to utilize them effectively constitutes an important step towards improving the practicality of the proposed method.
	
	\item The ideas discussed in Section \ref{sec:practicalities} constitute promising research avenues towards improving practicality of the proposed method. Specifically, there are several open questions pertaining to the concrete implementation of Algorithm \ref{alg:successive_oa} and the way Problem \eqref{eq:pwpOCP} can be used to inform an effective use of the different bound tightening mechanisms. Furthermore, the decomposable, weakly coupled structure of the bounding SDPs motivates other forms of exploitation than Algorithm \ref{alg:successive_oa}; in particular the use of distributed optimization techniques such as ADMM \cite{boyd2011distributed} or Schwarz-like approaches \cite{shin2020schwarz,na2020overlapping} appears promising. 

\end{enumerate} 

\section*{Supplementary Information}
The Supplementary Information to this article can be found at the end of this document.
\section*{References}
\bibliography{references}

\end{document}

\maketitle 
	In this document, we state detailed arguments for claims made in the main article. % citation TBD
\section{Proof of Proposition 1}\label{app:prop_increment_conditions}
\begin{namedthm}{Proposition 1}
Let $0 \leq t_1 \leq t_2 < + \infty$ and $n_I \leq \floor{\frac{m}{q}}$. Further, consider test functions of the form $g_l(t) = \mathds{1}_{[t_1,t_2]}(t)(t_2-t)^l$. If $\bm{y}_{t_1}, \bm{y}_{t_2} \in \R^{n_L+n_H}$ and $\bm{z}_{g_l,t_2} \in \R^{n_L+n_H}$ satisfy 
		\begin{align}
		\bm{K}\left(g_l(t_2) \bm{y}_{t_2} - g_l(t_1)\bm{y}_{t_1} \right) = \bm{A} \bm{z}_{g_l,t_2} - l\bm{K}\bm{z}_{g_{l-1},t_2} \label{eq:taylor_refined_appendix}
		\end{align}
		for $l = 0,\dots,n_I$, then $\bm{y}_{t_1}, \bm{y}_{t_2}$ and $\bm{z}_{g_l,t_2}$ also satisfy % Condition \eqref{eq:taylor} for $n=0,1,\dots,n$.
		\begin{align}
		%\bm{K}_{l} \bm{y}_{t_2} = \bm{T}_l(\bm{K}_l\bm{y};t_1,t_2)\bm{y}_{t_1} + \bm{R}_l(\bm{K}_l\bm{y};t_1,t_2)\bm{z}_{g_l,t_2} \label{eq:taylor}
		%
		y_{\bm{j},t_2} = \bm{c}_{l,\bm{j}}(t_1,t_2)^\top\bm{y}_{t_1} + \bm{d}_{l,\bm{j}}(t_1,t_2)^\top\bm{z}_{g_l,t_2} \label{eq:taylor_appendix}
		\end{align} 
		for $l=0,\dots,n_I$ and $\bm{j}$ such that $|\bm{j}|\leq m-lq$, where $\bm{c}_{l,\bm{j}}$ and $\bm{d}_{l,\bm{j}}$ are defined as in $(3)$. 
\end{namedthm}

\begin{proof}
	We will first introduce a convenient formalism to described the conditions in the form of \eqref{eq:taylor_appendix} in terms of the coefficient matrix $\bm{A}$ of the moment dynamics
	\begin{align*}
	    \bm{K}\frac{d\bm{y}}{dt}(t) = \bm{A} \bm{y}(t).
	\end{align*}
	To that end, we assume that the moments are ordered according to their corresponding multi-indices, i.e., the first entry corresponds to $y_{\bm{0}}$, the following $n$ entries correspond to the moments $y_{\bm{j}}$ with $|\bm{j}| \leq 1$ and so on. Then, we define the matrices $\bm{K}_l = [\bm{I}_{n_L(l)\times n_L(l)} \ \bm{0}_{n_L(l)\times n_H(l)}]$ with $n_L(l) = {m - lq + n \choose n}$ and $n_H(l) = {m+q+n \choose n} - n_L(l)$. In words, applying $\bm{K}_l$ to $\bm{y}$ extracts all moments up to order $m-lq$. Now we can express the higher derivatives $\frac{d^k}{dt^k}\left(\bm{K}_l \bm{y}\right)$ for any $k \leq l+1$ in the following way: let $\bm{A}_i$ be the coefficient matrix such that $\frac{d}{dt}\left( \bm{K}_{i+1} \bm{y} \right) = \bm{A}_i \bm{K}_{i} \bm{y}$. By induction, we then arrive at the convenient expression
	\begin{align*}
	\frac{d^k}{dt^k}\left(\bm{K}_l \bm{y}\right) = \prod_{i=1}^{k} \bm{A}_{l-i} \bm{K}_{l-k} \bm{y}.
	\end{align*}
	Recalling the definition of Equation \eqref{eq:taylor_appendix} in terms of the Taylor series expansion of the true moment trajectories, it follows that the equation system 
	\begin{align}
	\bm{K}_l \bm{y}_{t_2} = \sum_{k=0}^l \frac{(t_2-t_1)^k}{k!} \prod_{i=1}^{k} \bm{A}_{l-i} \bm{K}_{l-k} \bm{y}_{t_1} + \frac{1}{l!} \prod_{i=1}^{l+1} \bm{A}_{l-i} \bm{z}_{g_{l},t_2} \tag{$2_l$} \label{eq:prop_1_to_show}
	\end{align}
	corresponds to all equations of the form of $(5)$ with $|\bm{j}| \leq m - lq$ for fixed $l \in \{0,\dots,n_I\}$. Note that we already used the fact that $\bm{K}_{-1}$ reduces to the identity in the above relation. 
	
	%using that $\bm{K}_{-1} = \bm{I}_{(n_L+n_H)\times (n_L+n_H)}$. 
	Now notice that Equation \eqref{eq:prop_1_to_show} holds for all $l =0,\dots,n_I$ if and only if 
	\begin{align}
	\bm{K}\left(\bm{y}_{t_2} - \bm{y}_{t_1}\right) = \bm{A}\bm{z}_{g_0,t_2} \label{eq:prelim_1}
	\end{align}
	holds alongside
	\begin{align}
	%&\prod_{i=1}^{l-1} \bm{A}_{l-i} \bm{K}_{l-1-k} \bm{y}_{t_1} +
	\frac{1}{(l-1)!} \prod_{i=1}^{l} \bm{A}_{l-i} \bm{K}_{0}\bm{z}_{g_{l-1},t_2} = \frac{(t_2-t_1)^l}{l!}\prod_{i=1}^{l} \bm{A}_{l-i} \bm{K}_{0} \bm{y}_{t_1} + \frac{1}{l!} \prod_{i=1}^{l+1} \bm{A}_{l-i} \bm{z}_{g_{l},t_2} \label{eq:prelim_2} 
	\end{align}
	for $l=1,\dots,n_I$ since Equation \eqref{eq:prelim_1} is equivalent to Equation $(2_0)$ and Equation \eqref{eq:prelim_2} is obtained simply by matching coefficients of the Equations $(2_l)$ and $(2_{l-1})$ for each $l \in \{1,\dots,n_I\}$.
	
	Factoring out $\prod_{i=1}^l A_{l-i}$ in Equation \eqref{eq:prelim_2} and noting that $\bm{A}_{-1} = \bm{A}$ as well as $\bm{K}_0 = \bm{K}$ shows that Equation \eqref{eq:prelim_2} is implied by
	\begin{align}
	l \bm{K} \bm{z}_{g_{l-1},t_2} = (t_2-t_1)^l\bm{K} \bm{y}_{t_1} + \bm{A} \bm{z}_{g_{l},t_2} \label{eq:prelim}.
	\end{align}
	Thus, by using that $g_l(t_2) = 0$ for $l=1,\dots,n_I$, it follows that \eqref{eq:prelim} is equivalent to \eqref{eq:taylor_refined_appendix} for $l=1,\dots,n_I$. We finally conclude the proof by noting that \eqref{eq:prelim_1} is equivalent to \eqref{eq:taylor_refined_appendix} for $l=0$.
	
\end{proof}

\section{Commuting Property of $I_L$ and $I_R$}\label{app:commuting_operators}
\begin{lemma}
    The left and right integral operators $I_L, I_R: \mathcal{C}(\mathbb{R}^2) \to \mathcal{C}(\mathbb{R}^2)$ defined by
    \begin{align*}
    	(I_L f)(t_1,t_2) = \int_{t_1}^{t_2} f(t_1, t) \, dt\\
    	(I_R f)(t_1,t_2) = \int_{t_1}^{t_2} f(t, t_2) \, dt.
    \end{align*}
    commute.
\end{lemma}
\begin{proof}
    Using the definition of $I_L$ and $I_R$ we obtain 
    \begin{align*}
        (I_L I_R f)(t_1,t_2) = \int_{t_1}^{t_2} (I_R f)(t_1, t) \, dt = \int_{t_1}^{t_2}  \int_{t_1}^{t} f(s,t) \, ds\, dt  
    \end{align*}
    By defining $\Omega_{LR} = \{(s,t) \in \mathbb{R}^2 :  t_1 \leq  s \leq t,  t_1 \leq t \leq t_2\}$ we can write more concisely
    \begin{align*}
        (I_L I_R f)(t_1,t_2) = \int_{ \Omega_{LR} } f(s,t) \, d(s,t) 
    \end{align*}
    Likewise, we have that 
    \begin{align*}
        (I_R I_L f)(t_1,t_2) = \int_{t_1}^{t_2} (I_L f)(s, t_2) \, ds = \int_{t_1}^{t_2} \int_{s}^{t_2} f(s,t) \, dt\, ds  = \int_{\Omega_{RL}} f(s,t) \, d(s,t) 
    \end{align*}
    where $\Omega_{RL} = \{ (s,t) \in \mathbb{R}^2 : t_1 \leq s \leq t_2, s\leq t \leq t_2 \}$. Observing that $\Omega_{RL} = \Omega_{LR}$ concludes the proof.  
\end{proof}

\section{Proof of Corollary 1}\label{app:cor_implications}
\begin{corollary} \label{cor:app_implications}
	Let $n_I \in \Z_+$ and $t_T > 0$ be fixed. Further suppose $g \in \mathcal{AC}([0,t_T])$ is non-negative, and let $\bm{y}$ and $\bm{z}^l(g;\cdot)$ be arbitrary functions such that $\bm{z}^l(g;\cdot)$ is linear in the first argument and $\bm{z}^0(g;t)= g(t)\bm{y}(t)$ holds. Fix $0 \leq t_1 \leq t_2 \leq t_T$ and define $h_l(t) = \mathds{1}_{[t_1,t_2]}(t)(t_2-t)^l$ for $l=0,1,\dots,n_I$. If Conditions (i) and (ii) of Proposition 2 are satisfied by $\St{\bm{z}^l(g;t_i)}_{l=0}^{n_I+1}$ for $i=1,2$, then there exist functions $\bm{z}(h_lg;\cdot)$ that are linear in the first argument, and satisfy
	\begin{align}
		\bm{K}\left(h_l(t_2)g(t_2)\bm{y}(t_2) - h_l(t_1)g(t_1)\bm{y}(t_1) \right) =\bm{A} \bm{z}(h_l g;t_2)  + \bm{K}\bm{z}((h_lg)';t_2) \label{eq:cor_1_toshow_1} % \label{eq:taylor_g_refined}
	\end{align}
	and
	\begin{align}
		\bm{z}(h_l g; t_2) \in C(X)  \label{eq:cor_1_toshow_2}
	\end{align}
	for all $l \in \St{0,\dots,n_I}$. 
\end{corollary}
\begin{proof}
	Let $l\in \St{0,\dots, n_I}$ be fixed. We show that 
    	\begin{align}
    	\bm{z}(h_l g;t_2) =  l! \left(\bm{z}^{l+1}(g;t_2) - \sum_{k=0}^l \frac{(t_2-t_1)^{l-k}}{(l-k)!} \bm{z}^{k+1}(g;t_1) \right) \label{eq:moment_expansion}
	\end{align}
	satisfies the desired conditions. 
	First note that \eqref{eq:moment_expansion} is equivalent to 
	\begin{align*}
	    \bm{z}(h_l g;t_2) = l! \Omega_{l+1,0}\left(\St{\bm{z}^i(g;t_1)}_{i=1}^{l+1},\St{\bm{z}^i(g;t_2)}_{i=1}^{l+1},t_1,t_2\right).
	\end{align*}
	Since by assumption 
	\begin{align*}
	    \Omega_{l+1,0}\left(\St{\bm{z}^i(g;t_1)}_{i=1}^{l+1},\St{\bm{z}^i(g;t_2)}_{i=1}^{l+1},t_1,t_2\right) \in C(X),
	\end{align*}
	it follows that \eqref{eq:cor_1_toshow_2} holds for the definition in \eqref{eq:moment_expansion}. 
	
% 	Further, integrating $\bm{K} \int_{t_1}^{t_2} (t_2-t)^l g(t) \frac{d\bm{y}}{dt}(t) \, dt$ by parts yields
% 	\begin{align*}
%     	\bm{K} \left( l \int_{t_1}^{t_2} (t_2-t)^{l-1} g(t) \bm{y}(t) - (t_2-t)^l g'(t) \bm{y}(t)\, dt - (t_2-t_1)^l g(t_1)\bm{y}(t_1) \right) = \bm{A} \int_{t_1}^{t_2} (t_2-t)^l \bm{y}(t)dt.
% 	\end{align*}
	Next, note that \eqref{eq:cor_1_toshow_1} can be equivalently written as
	\begin{align*}
    	\bm{A} \bm{z}(h_{l}g;t_2) = \bm{K} \left( l \bm{z}(h_{l-1}g;t_2) - \bm{z}(h_lg';t_2) -(t_2-t_1)^l g(t_1)\bm{y}(t_1) \right)
	\end{align*}
	using the definition of $h_l$ and that by assumption $\bm{z}((h_lg)';t_2) = \bm{z}(h_l'g;t_2) + \bm{z}(h_lg';t_2)$.	Thus, we need to show that the following identity holds:
	\begin{align*}
    	\bm{A}\left(\bm{z}^{l+1}(g;t_2) - \sum_{k=0}^l \frac{(t_2-t_1)^{l-k}}{(l-k)!} \bm{z}^{k+1}(g;t_1) \right) =& \bm{K} \left(\bm{z}^{l}(g;t_2) - \sum_{k=0}^{l-1} \frac{(t_2-t_1)^{l-1-k}}{(l-1-k)!} \bm{z}^{k+1}(g;t_1) \right)\\
    	& - \bm{K} \left(\bm{z}^{l+1}(g';t_2) - \sum_{k=0}^{l} \frac{(t_2-t_1)^{l-k}}{(l-k)!} \bm{z}^{k+1}(g';t_1) \right)\\
    	& -\bm{K} \frac{(t_2-t_1)^{l}}{l!} g(t_1)\bm{y}(t_1).
	\end{align*}
	Reordering the terms in this relation yields
	\begin{align*}
		&\sum_{k=0}^l \frac{(t_2-t_1)^{l-k}}{(l-k)!} \left(\bm{A} \left[ \bm{z}^{k+1}(g;t_2) - \bm{z}^{k+1}(g;t_1) \right] + \bm{K} \left[\bm{z}^{k+1}(g';t_2) - \bm{z}^{k+1}(g';t_1)\right]\right)\\
		& -\sum_{k=1}^{l} \frac{(t_2-t_1)^{l-k}}{(l-k)!} \bm{K}\left(\bm{z}^{k}(g;t_2) - \bm{z}^{k}(g;t_1)\right)\\
		=&  \sum_{k=0}^{l-1} \frac{(t_2-t_1)^{l-k}}{(l-k)!} \left(\bm{A}\bm{z}^{k+1}(g;t_2) + \bm{K} \bm{z}^{k+1}(g';t_2)\right) - \bm{K} \sum_{k=1}^{l-1} \frac{(t_2-t_1)^{l-k}}{(l-k)!} \bm{z}^{k+1}(g;t_2) \\
		& -  \frac{(t_2-t_1)^l}{l!} g(t_1) \bm{y}(t_1)
	\end{align*}
	and finally using that
	\begin{align*}
	    \bm{A}\bm{z}^{k+1}(g;t_i) + \bm{K}\bm{z}^{k+1}(g';t_i) = \bm{K} \left(\bm{z}^k(g;t_i) - \frac{(t_2-t_1)^k}{k!}g(0) \bm{y}(0) \right) 
	\end{align*}
	holds for all $k \in \St{0, \dots, n_I}$ and $i \in \St{1,2}$ with $\bm{z}^{0}(g;t_i) = g(t_i)\bm{y}(t_i)$, it is easily verified that \eqref{eq:cor_1_toshow_1} is indeed satisfied for $\bm{z}(h_lg;t_2)$ as defined in \eqref{eq:moment_expansion}. 
\end{proof}

\section{Proof of Corollary 2}\label{app:cor_linear_scaling}
		Our proof of Corollary 2 relies in large part on explicit computations. To that end, it is first essential to provide an explicit algebraic relation for $\Omega_{l,k}(\Set{\bm{z}^s(g;t_{1})}_{s=1}^{l}, \Set{\bm{z}^s(g;t_{2})}_{s=1}^{l}, t_1, t_{2})$. Recall that $\Omega_{k,l}$ was introduced to compactly denote expressions of the form $(I^n_L I^m_R \bm{f})(t_1,t_2)$ with $\bm{f}(x,y) = \bm{z}^1(g;y) - \bm{z}^1(g;x)$. To avoid unnecessary clutter of notation, we will give an explicit algebraic expression for the latter; adjusting the indices to obtain an expression for $\Omega_{l,k}$ is straightforward and can be found in the proof of Corollary 2. It is easily verified by induction that for any $m,n \in \Z_+$, the following identity holds 	
\begin{multline}	
	(I^n_L I^m_R \bm{f})(t_1,t_2) = \sum_{k=0}^m (-1)^k {n + k \choose n} \frac{(t_2-t_1)^{m-k}}{(m-k)!} \bm{z}^{n+1+k}(g;t_2) + (-1)^{m+1} \sum_{k=0}^n {m+k \choose m} \frac{(t_2- t_1)^{n-k}}{(n-k)!} \bm{z}^{m+1+k}(g;t_1) \label{eq:omega}	
\end{multline}
With this in hand, we will first proof the following Lemma which answers the essential question behind Corollary 2.
\begin{lemma}\label{lemma:scaling}
	Let $\bm{z}^l$ be a sequence of functions that satisfies $\bm{z}^{l}(t) =\int_0^t \bm{z}^{l-1}(\tau)d\tau$ for $l\geq 2$. Further, define $\bm{f}(x,y)= \bm{z}^1(y)- \bm{z}^1(x)$ and consider a convex cone $C$. Then, for any $0\leq t_1 \leq t_2 \leq t_3 < +\infty$ and $n,m \in \Z_+$, the following implications hold:
	\begin{align}
    	\begin{array}{l}
    	(I^{p}_L I^{s}_R \bm{f})(t_1,t_2),(I^{p}_L I^{s}_R \bm{f})(t_2,t_3) \in C \text{ for all } 0\leq p\leq n+1 \text{ and } 0 \leq s\leq m  \\[1em]
    	\implies (I^{n+1}_L I^{m}_R \bm{f})(t_1,t_3) \in C
    	\end{array} \label{eq:to_show_cor2_a}
	\end{align}
	and 
	\begin{align}
    	\begin{array}{l}
    	(I^{p}_L I^{s}_R \bm{f})(t_1,t_2),(I^{p}_L I^{s}_R \bm{f})(t_2,t_3) \in C  \text{ for all } 0\leq p\leq n \text{ and } 0 \leq s\leq m +1 \\[1em]
    	\implies (I^{n}_L I^{m+1}_R \bm{f})(t_1,t_3) \in C
    	\end{array} \label{eq:to_show_cor2_b}
	\end{align}
\end{lemma}
\begin{proof}
	We focus on proving \eqref{eq:to_show_cor2_a}. 
	To that end, note that in order to establish \eqref{eq:to_show_cor2_a}, it suffices to show that
	\begin{align}
	\int_{t_2}^{t_3} (I^n_L I^m_R \bm{f})(t_1,t) - (I^n_L I^m_R \bm{f})(t_2,t)\, dt \in C \label{eq:toshow}
	\end{align}
	since
	\begin{align*}
	(I^{n+1}_L I^{m}_R \bm{f})(t_1,t_3) = \int_{t_1}^{t_2} (I^n_L I^m_R \bm{f})(t_1,t)\, dt + \int_{t_2}^{t_3} (I^n_L I^m_R \bm{f})(t_1,t) \, dt.
	\end{align*}
	We will thus show that \eqref{eq:toshow} can be expressed as a conic combination of $(I^p_L I^s_R \bm{f})(t_1,t_2)$ and $(I^p_L I^s_R \bm{f})(t_2,t_3)$ for all $0 \leq p \leq n +1$ and $0 \leq s \leq m$. From the definition we obtain
	\begin{align*}
	(I^n_L I^m_R \bm{f})(t_1,t) - (I^n_L I^m_R \bm{f})(t_2,t) =& \sum_{k=0}^m (-1)^{m-k} {n+m-k \choose n} \frac{(t-t_1)^k}{k!} \bm{z}^{n+m+1-k}(t)\\
	&- \sum_{k=0}^m (-1)^{m-k} {n+m-k\choose n} \frac{(t-t_2)^k}{k!} \bm{z}^{n+m+1-k}(t) \\
	&+ (-1)^{m+1} \sum_{k=0}^n { m+n-k \choose m} \frac{(t-t_1)^k}{k!} \bm{z}^{n+m+1-k}(t_1)\\
	&- (-1)^{m+1} \sum_{k=0}^n { m+n-k \choose m} \frac{(t-t_2)^k}{k!} \bm{z}^{n+m+1-k}(t_2).
	\end{align*}
	Using the binomial formula, we further obtain  
	\begin{align*}
	&\sum_{k=0}^m (-1)^{m-k} {n+m-k \choose n} \frac{(t-t_1)^k}{k!} \bm{z}^{n+m+1-k}(t) \\
	=&\sum_{k=0}^m (-1)^{m-k} {n+m-k \choose n} \bm{z}^{n+m+1-k}(t) \sum_{l=0}^k \frac{(t-t_2)^{k-l}}{(k-l)!}\frac{(t_2-t_1)^{l}}{l!}\\
	=&\sum_{l=0}^m \frac{(t_2-t_1)^{l}}{l!} \sum_{k=0}^{m-l} (-1)^{m-k-l} {n+m-k-l \choose n} \frac{(t-t_2)^{k}}{k!} \bm{z}^{n+m+1-k-l}(t) 
	\end{align*}
	and likewise 
	\begin{align*}
	&(-1)^{m+1} \sum_{k=0}^n { m+n-k \choose m} \frac{(t-t_1)^k}{k!} \bm{z}^{n+m+1-k}(t_1) \\
	=& (-1)^{m+1} \sum_{k=0}^n { m+n-k \choose m}\bm{z}^{n+m+1-k}(t_1) \sum_{l=0}^k \frac{(t_2-t_1)^{k-l}}{(k-l)!} \frac{(t-t_2)^l}{l!} \\
	=& (-1)^{m+1} \sum_{l=0}^n \frac{(t-t_2)^l}{l!} \sum_{k=0}^{n-l} { m+n-k-l \choose m} \frac{(t_2-t_1)^{k}}{k!} \bm{z}^{n+m+1-k-l}(t_1).
	\end{align*}
	
	Combining the above identities with the definition thus yields
	\begin{align*}
	&(I^n_L I^m_R \bm{f})(t_1,t) - (I^n_L I^m_R \bm{f})(t_2,t) \\
	=&\sum_{l=1}^m \frac{(t_2-t_1)^l}{l!} \sum_{k=0}^{m-l} (-1)^{m-k-l} {n + m -l - k \choose n}\frac{(t-t_2)^k }{k!} \bm{z}^{n+m+1-l-k}(t) \\
	&+ \sum_{l=0}^n \frac{(t-t_2)^l}{l!} \left \lbrace (-1)^{m+1} \sum_{k=0}^{n-l} { m+n-k-l \choose m} \frac{(t_2-t_1)^{k}}{k!} \bm{z}^{n+m+1-k-l}(t_1)\right.\\
	& \left.\qquad \qquad \qquad \qquad -(-1)^{m+1}  { m+n-l \choose m}  \bm{z}^{n+m+1-l}(t_2)\right\rbrace.
	\end{align*}
	
	By adding and subtracting
	\begin{align*}
	&\sum_{l=1}^m (-1)^{m-l+1} \frac{(t_2-t_1)^l}{l!} \sum_{k=0}^{n} {n+m-l-k\choose m-l} \frac{(t-t_2)^k}{k!} \bm{z}^{n+m+1-l-k}(t_2)  \\
	=&\sum_{l=0}^{n}\frac{(t-t_2)^l}{l!} \sum_{k=1}^m (-1)^{m-k+1} {n+m-l-k \choose n-l}  \frac{(t_2-t_1)^{k}}{k!} \bm{z}^{n+m+1-l-k}(t_2),
	\end{align*} 
	we get 
	\begin{align*}
	&(I^n_L I^m_R \bm{f})(t_1,t) - (I^n_L I^m_R \bm{f})(t_2,t)\\
	=&\sum_{l=1}^m \frac{(t_2-t_1)^l}{l!} \left\lbrace (-1)^{m-l+1} \sum_{k=0}^n {n+m-l-k \choose m-l} \frac{(t-t_2)^k}{k!} \bm{z}^{n+m+1-l-k}(t_2)  \right.\\
	&\qquad \qquad \qquad \qquad \left. + \sum_{k=0}^{m-l} (-1)^{m-k-l} {n + m -l - k \choose n}\frac{(t-t_2)^k }{k!} \bm{z}^{n+m+1-l-k}(t) \right\rbrace \\
	&+(-1)^{m+1} \sum_{l=0}^n \frac{(t-t_2)^l}{l!}\left \lbrace \sum_{k=0}^{n-l}{m+n-k-l \choose m} \frac{(t_2-t_1)^k}{k!} \bm{z}^{n+m+1-k-l}(t_1) \right.\\
	&\left.\qquad \qquad  \qquad \qquad  \qquad \qquad \qquad \qquad \qquad \qquad  - {m+n-l\choose m} \bm{z}^{n+m+1-l}(t_2) \right.\\
	&\qquad \qquad \qquad \qquad \qquad \left. - \sum_{k=1}^{m} (-1)^{-k} {n+m-l-k \choose m-k} \frac{(t_2-t_1)^k}{k!} \bm{z}^{n+m+1-l-k}(t_2)  \right\rbrace
	\end{align*}
	which can further be simplified to
	\begin{align*}
	&\sum_{l=1}^m \frac{(t_2-t_1)^l}{l!} \left\lbrace (-1)^{m-l+1} \sum_{k=0}^n {n+m-l-k \choose m-l} \frac{(t-t_2)^k}{k!} \bm{z}^{n+m+1-l-k}(t_2)  \right.\\
	&\qquad \qquad \qquad \qquad \left. + \sum_{k=0}^{m-l} (-1)^{m-k-l} {n + m -l - k \choose n}\frac{(t-t_2)^k }{k!} \bm{z}^{n+m+1-l-k}(t) \right\rbrace \\
	&+ \sum_{l=0}^n \frac{(t-t_2)^l}{l!}\left \lbrace (-1)^{m+1} \sum_{k=0}^{n-l}{m+n-k-l \choose m} \frac{(t_2-t_1)^k}{k!} \bm{z}^{n+m+1-k-l}(t_1) \right.\\
	&\qquad \qquad \qquad \qquad \left. + \sum_{k=0}^{m} (-1)^{m-k} {n+m-l-k \choose m-k} \frac{(t_2-t_1)^k}{k!} \bm{z}^{n+m+1-l-k}(t_2)  \right\rbrace.
	\end{align*}
	Close inspection reveals that this is precisely what we set out to show: every term in the first sum is equivalent to $(I^n_LI^{m-l}_R \bm{f})(t_2,t)$ while every term in the second sum is equivalent to $(I^{n-l}_LI^{m}_R \bm{f})(t_1,t_2)$. Thus, integrating the above relation over $[t_2,t_3]$ yields a conic combination of $(I^p_L I^s_R \bm{f})(t_1,t_2)$ and $(I^p_L I^s_R \bm{f})(t_2,t_3)$ with $0 \leq p \leq n + 1$ and $0 \leq s \leq m$. 
	
	The proof of implication \eqref{eq:to_show_cor2_b} is analogous.
\end{proof}

\begin{namedthm}{Corollary 2}%\label{cor:app_linear_scaling}
	Let $0 \leq t_1 \leq t_2 \leq t_3 < +\infty$ and $n_I$ be a fixed positive integer. Suppose $\St{\bm{z}^s}_{s=1}^{n_I}$ is a set of functions such that
	\begin{align*}
	     \Omega_{l,k}(\Set{\bm{z}^s(t_{i})}_{s=1}^{l}, \Set{\bm{z}^s(t_{i+1})}_{s=1}^{l}, t_i, t_{i+1}) \in C(X)
	\end{align*}
	for all $i \in \{1,2\}$ and $k,l \in \Z_+$ such that $k < l \leq n_I$. Then, 
	\begin{align*}
	\Omega_{l,k}(\Set{\bm{z}^s(t_{1})}_{s=1}^{l}, \Set{\bm{z}^s(t_{3})}_{s=1}^{l}, t_1, t_{3}) \in C(X)
	\end{align*}
	holds for all $k,l \in \Z_+$ such that $k < l \leq n_I$. 
\end{namedthm}
\begin{proof}
	As argued in the beginning of this section, $\Omega_{l,k}(\Set{\bm{z}^s(t_{i})}_{s=1}^{l}, \Set{\bm{z}^s(t_{i+1})}_{s=1}^{l}, t_i, t_{i+1})$ takes the following algebraic form
	\begin{multline*}
		\Omega_{l,k}(\Set{\bm{z}^s(t_{i})}_{s=1}^{l}, \Set{\bm{z}^s(t_{i+1})}_{s=1}^{l}, t_i, t_{i+1}) = \sum_{s=0}^k (-1)^s {l-1-k + s \choose l-1-k} \frac{(t_2-t_1)^s}{s!} \bm{z}^{l-k+s}(t_{i+1}) \\ + (-1)^{k+1} \sum_{s=0}^{l-1-k} {k+s \choose k} \frac{(t_2- t_1)^{l-1-k-s}}{(l-1-k-s)!} \bm{z}^{k+1+s}(t_i).
	\end{multline*}
	The result therefore follows immediately from Lemma \ref{lemma:scaling}. 
\end{proof}

\section{Reformulation of Condition (i) in Proposition 2}\label{app:reformulation}
In the definition of $\mathsf{S}(\mathsf{F},\mathsf{T},n_I)$ we reformulated Condition (i) of Proposition 2 in order to endow  $\mathsf{S}(\mathsf{F},\mathsf{T},n_I)$ with a decomposable structure. Without modification, Condition (i) in Proposition 2 gives rise to the following necessary moment conditions 
\begin{align}
    \Gamma_g\left(\{\bm{z}^{l}(f;t)\}_{f \in \mathsf{F}}\right) = \bm{K} \left(\bm{z}^{l-1}(g;t) - \frac{t^{l-1}}{(l-1)!}g(0) \bm{y}_0\right), \quad \forall t \in \mathsf{T}. \label{eq:cond_i}
\end{align}
In order to turn this condition into a form that lends itself to be decomposed as described in Section IV, we can recognize that the above conditions for evaluated for different time points are only linked through the term $g(0) \bm{K} \bm{y}_0$. To replace this linkage with a link between adjacent time points as required for decomposability, we can rearrange the above equation to obtain 
\begin{align}
    g(0) \bm{K}  \bm{y}_0 = \frac{(l-1)!}{t^{l-1}}\left(\bm{K} \bm{z}^{l-1}(g;t) - \Gamma_g\left(\{\bm{z}^{l}(f;t)\}_{f \in \mathsf{F}}\right) \right), \quad \forall t \in \mathsf{T} \label{eq:link}.
\end{align} 
We then simply replace the linking term in Equation \eqref{eq:cond_i} by the right-hand-side of Equation \eqref{eq:link} evaluated at an appropriately shifted time point. Concretely, for the first time point in $\mathsf{T}$ we retain the unmodified equation 
\begin{align*}
    &\Gamma_g\left(\{\bm{z}^{l}(f;t_1)\}_{f \in \mathsf{F}}\right) = \bm{K} \left(\bm{z}^{l-1}(g;t_1) - \frac{t_1^{l-1}}{(l-1)!}g(0) \bm{y}_0\right)
\end{align*}
and for any other $t \in \mathsf{T}$ we reformulate Equation \eqref{eq:cond_i} by replacing $g(0)\bm{K}\bm{y}_0$ with the right-hand-side of Equation \eqref{eq:link} evaluated at the adjacent time point $n(t)$. This yields the desired relation
\begin{align*}
	 \Gamma_g\left(\{\bm{z}^{l}(f;t)\}_{f \in \mathsf{F}}\right) - \left(\frac{t}{n(t)}\right)^{l-1}\Gamma_g \left(\{\bm{z}^{l}(f;n(t))\}_{f \in \mathsf{F}}\right) = \bm{K} \left(\bm{z}^{l-1}(g;t) - \left(\frac{t}{n(t)}\right)^{l-1} \bm{z}^{l-1}(g;n(t)) \right), \quad \forall t \in \mathsf{T} \setminus \{t_1\}.
\end{align*}